\documentclass[10pt]{article}
\usepackage{amsfonts,amssymb,amsmath,german,style_1,amsthm,mathrsfs,geometry,framed,verbatim}
\usepackage{tikz}
\usepackage[english]{babel}
\newtheorem{Lemma}{Lemma}[section]
\newtheorem{proposition}{Proposition}[section]
\newtheorem{theorem}{Theorem}[section]

\usepackage{color}
\begin{document}
\numberwithin{equation}{section}
\title{Self-Gravitating Relativistic Fluids:\\The Formation of a Free Phase Boundary\\in the Phase Transition from Hard to Soft}
\author{Demetrios Christodoulou, Andr\'e Lisibach\\ ETH Z\"urich}
\maketitle

\section{The Phase Transition from Hard to Soft\label{section1}}

In \cite{I} we introduced a relativistic fluid model with two phases, a soft phase which holds when the density of mass-energy $\rho$ is below a certain constant $\rho_0$ and in which the sound speed is zero, and a hard phase, which holds when $\rho$ is above $\rho_0$ and in which the sound speed is equal to the speed of light. The model is of relevance in the study of the gravitational collapse of the degenerate cores of massive stars, the associated supernova explosions and the formation of neutron stars and black holes. The soft phase corresponds to degenerate stellar matter below nuclear density while the hard phase corresponds to homogeneous nuclear matter at supernuclear densities. The constant $\rho_0$, which by an appropriate choice of units we set equal to unity, corresponds to the nuclear saturation density (see Section 2 of \cite{I}).

In \cite{I} we began the study of the dynamics of the model in the spherically symmetric case. The problem then reduces to one on a 2-dimensional quotient space-time manifold $Q$. Starting with initial conditions which correspond entirely to the soft phase, we showed that predictions based solely on the soft phase break down beyond an achronal boundary $\partial \mathscr{J}^+(\mathscr{K})$, consisting of smooth spacelike segments $\Sigma_i$ along which $\rho=1$ and which form the spacelike part of the phase boundary, joined by pairs $C_i^+$, $C_{i+1}^-$ of outgoing and incoming null segments (see Section 5 of \cite{I}). The end points $N_i^+$, $N_i^-$ of the spacelike segments $\Sigma_i$ at which $\Sigma_i$ turns null, and which are, at the same time, the past end points of the null segments $C_i^+$ and $C_{i+1}^-$ respectively, we called \textit{boundary null points}. The data induced by the soft phase along the spacelike segments $\Sigma_i$ provide the initial conditions for a subsequent hard phase, determined in $\mathscr{D}^+(\Sigma_i)$, the future domain of dependence of $\Sigma_i$. In Section 6 of \cite{I}, we studied the associated hard-phase Cauchy problem in the large and analyzed the behavior near the null points, which, as we showed, are analogous to the \textit{branch points} of minimal surface theory.

In \cite{II} we began the investigation of the problem of extending the solution into the causal future of the boundary null points $N_i^+$, $N_i^-$. We formulated the problems of the formation and continuation of a free phase boundary and, after considering first the case of a vacuum free boundary, we solved the continuation problem in general. We then turned to more global aspects of the free-boundary problem, in particular to the study of how a free phase boundary terminates. This led to the consideration of four cases, the last of which was proved to be non-generic, while the first corresponds to a spacelike transition from hard to soft with a null end point lying on the boundary of the causal past of the point of termination, the second case corresponds to a point of intersection with a spacelike transition from soft to hard, from which point the discontinuity propagates in the hard phase as a null shock, and the third case which corresponds to a point of intersection with a spacelike transition from hard to soft, from which point the discontinuity propagates as a contact discontinuity in the soft phase.

In \cite{III} we solved the problem of formation of a free phase boundary in the phase transition from soft to hard.

However there was one remaining problem in the above program which was not addressed. This was the problem of the formation of a free phase boundary in the phase transition from hard to soft. The formulation and solution of this problem is the content of the present paper.

Let $\Sigma_0$ be a piecewise smooth curve in the quotient manifold $Q$ containing no pair of timelike related points. $\Sigma_0$ then consists of spacelike and null segments. Let initial data for the hard phase equations be given along $\Sigma_0$ so that
\begin{align}
  -g^{ab}\p_a\phi\p_b\phi\geq 1\quad:\quad\textrm{along $\Sigma_0$}.\notag
\end{align}
We denote by $\mathscr{D}^+(\Sigma_0)$ the future domain of dependence of $\Sigma_0$ (this refers to the underlying conformal geometry which is that of the Minkowski plane). Corresponding to the given initial data there is a solution of the hard phase equations defined on a maximal future domain of development $\mathscr{D}^+_\ast (\Sigma_0)\subset \mathscr{D}^+(\Sigma_0)$. We remark that, as a consequence of prop.\ 6.1 of \cite{I}, $\inf_{\mathscr{D}^+_\ast(\Sigma_0)}r>0$ implies $\mathscr{D}^+_\ast(\Sigma_0)=\mathscr{D}^+(\Sigma_0)$.
Recalling that the hard phase is limited by the condition $-g^{ab}\p_a\phi\p_b\phi:=\sigma^2\geq 1$, the set
\begin{align}
  \mathscr{L}=\left\{x\in\mathscr{D}_\ast^+(\Sigma_0):-g^{ab}\p_a\phi\p_b\phi<1\right\}\notag
\end{align}
must be excised from $\mathscr{D}_\ast^+(\Sigma_0)$. This is however not all; for, causality demands that for each excluded point the domain of its influence, i.e.\ its causal future, be also excluded. Thus we must excise from $\mathscr{D}_\ast^+(\Sigma_0)$ the set $\mathscr{J}^+(\mathscr{L})$, the causal future of $\mathscr{L}$ in $\mathscr{D}_\ast^+(\Sigma_0)$. The boundary $\p\mathscr{J}^+(\mathscr{L})$ of $\mathscr{J}^+(\mathscr{L})$ in $\mathscr{D}_\ast^+(\Sigma_0)$ is a Lipschitz curve containing no pair of timelike related points. It consists in general of a finite number or a denumerable infinity of spacelike segments $\Sigma_i:i=1,2,\ldots$ contained with their end points $\p \Sigma_i=\{N_i^-,N_i^+\}$ in $\p \mathscr{L}$, and joined by pairs $C_i^+$, $C_{i+1}^-$ where $C_i^+$ is either empty or an outgoing null segment with past end point $N_i^+$ and $C_{i+1}^-$ is either empty or an incoming null segment with past end point $N_{i+1}^-$. At least one of $C_i^+$, $C_{i+1}^-$ is non-empty. If $C_i^+$ is empty then $N_i^+$ is the future end point of $C_{i+1}^-$ and $\Sigma_i$ is strictly spacelike at $N_i^+$. If $C_{i+1}^-$ is empty then $N_{i+1}^-$ is the future end point of $C_i^+$ and $\Sigma_{i+1}$ is strictly spacelike at $N_{i+1}^-$. If $C_i^-$ is non-empty then $\Sigma_i$ is tangent at $N_i^-$ to $C_i^-$ and if $C_i^+$ is non-empty then $\Sigma_i$ is tangent at $N_i^+$ to $C_i^+$. A point of $\bar{\Sigma}_i$ at which $\bar{\Sigma}_i$ is tangent to a null curve is a \textit{null point}. When both $C_i^-$, $C_i^+$ are non-empty, the boundary points $N_i^-$, $N_i^+$ are both null points. If we assume that $\p\mathscr{L}\cap\p\mathscr{J}^+(\mathscr{L})$ does not contain any critical points of the function $\sigma$, which is true generically, the spacelike segments $\Sigma_i$ are smooth. In the following we shall consider only the case where the boundary points are null points, for otherwise no peculiar difficulties are encountered.

Consider a given spacelike segment $\Sigma_i$. For convenience of notation we omit in the following the suffix $i$. Then $\Sigma$ is described by an equation of the form
\begin{align}
  \label{eq:1}
  v=h(u)
\end{align}
where $h$ is a smooth strictly decreasing function on the interval $(u^+,u^-]$, $h'<0$ on $(u^+,u^-)$, $h'(u^-)=0$. The point $(u^-,v^-=h(u^-))$ is the inner null end point $N^-$. We may also describe $\Sigma$ by the equation
\begin{align}
  u=h^{-1}(v)
\end{align}
and $h^{-1}$ is a smooth strictly decreasing function on the interval $(v^-,v^+]$, $(h^{-1})'<0$ on $(v^-,v^+)$, $(h^{-1})'(v^+)=0$. The point $(u^+=h^{-1}(v^+),v^+)$ is the outer null end point $N^+$. The former description is regular in a neighborhood of $N^-$, while the latter description is regular in a neighborhood of $N^+$. In the following we restrict attention to the former.

The hard phase solution defines the functions $\Omega$, $\phi$, $r$, $\dot{r}=Ur$, along $\Sigma$
\begin{align}
  \Omega&=\Omega_\ast(u)\\
  \phi&=\phi_\ast(u)\\
  r&=r_\ast(u),\qquad \dot{r}=\dot{r}_\ast(u)
\end{align} where
\begin{align}
    \Omega_\ast(u)&=\Omega(u,h(u))\\
    \phi_\ast(u)&=\phi(u,h(u))\\
\label{eq:2}
    r_\ast(u)&=r(u,h(u)),\qquad \dot{r}_\ast(u)=(Ur)(u,h(u)).
\end{align}
These functions constitute initial data for a soft phase solution in the future of $\Sigma$. The soft phase solution is described in comoving coordinates $\tau$, $\chi$ where we set
\begin{align}
  \label{eq:3}
  \chi=-u\,:\,\textrm{along $\Sigma$}
\end{align}
and $\tau$ is defined by the conditions that it coincides with $\phi_\ast$ along $\Sigma$. Thus $\Sigma$ is given in terms of the coordinates $\tau$, $\chi$ by
\begin{align}
  \label{eq:4}
  \tau=-f(\chi),
\end{align}
where
\begin{align}
  \label{eq:5}
  f(\chi)  =-\phi_\ast(-\chi).
\end{align}

We must require that the arc length along $\Sigma$ is that given by the prior hard phase solution
\begin{align*}
    ds^2=-\Omega^2_\ast h'du^2,
\end{align*}
the restriction to $v=h(u)$ of the metric $-\Omega^2dudv$ in null coordinates. This must coincide with
\begin{align*}
    ds^2=\left(-(f')^2+e^{2\omega_\ast}\right)d\chi^2,
\end{align*}
the restriction to $\tau=-f(\chi)$ of the soft phase metric $-d\tau^2+e^{2\omega}d\chi^2$ in comoving coordinates. Hence $\omega_\ast(\chi)=\omega(-f(\chi),\chi)$ is defined by
\begin{align}
  \label{eq:6}
e^{2\omega_\ast}=(f')^2-\Omega_\ast^2h'
\end{align}
where the argument of $\Omega_\ast$ and $h'$ is $u=-\chi$.

The initial data for the soft phase equation (eq. (3.45e) of \cite{I})
\begin{align}
  \label{eq:7}
  \frac{\p^2 r}{\p\tau^2}=-\frac{m}{r^2}
\end{align}
are
\begin{align}
  \label{eq:8}
  r(-f(\chi),\chi)=r_\ast(-\chi),\qquad  \frac{\p r}{\p \tau} (-f(\chi),\chi)=\dot{r}_\ast(-\chi).
\end{align}

We shall show below that the mass function $m$ along $\Sigma$ defined by the above soft phase initial data according to
\begin{align*}
  1-\frac{2m} {r}=-\left(\frac{\p r}{\p\tau}\right)^2+e^{-2\omega}\left(\frac{\p r}{\p \chi}\right)^2
\end{align*}
coincides with the mass function $m$ along $\Sigma$ induced by the prior hard phase,
\begin{align*}
  1-\frac{2m}{r}=-\frac{4}{\Omega^2}\frac{\p r}{\p u}\frac{\p r}{\p v}.
\end{align*}
To accomplish this we extend the null coordinates $u$, $v$ to the soft phase and show that the conformal factor $\Omega$ and the partial derivatives $\p r/\p u$, $\p r/\p v$ extend continuously across $\Sigma$, by virtue of the conditions \eqref{eq:6}, \eqref{eq:8}.

The functions $u(\tau,\chi)$, $v(\tau,\chi)$ are defined by the conditions
\begin{align}
  \label{eq:9}
  \frac{\p u}{\p\tau}+e^{-\omega}\frac{\p u}{\p\chi}&=0,\qquad u(-f(\chi),\chi)=-\chi\\
  \label{eq:10}
    \frac{\p v}{\p\tau}-e^{-\omega}\frac{\p v}{\p\chi}&=0,\qquad v(-f(\chi),\chi)=h(-\chi)
\end{align}
In the coordinates $u$, $v$ the soft phase metric $-d\tau^2+e^{2\omega}d\chi^2$ takes the form $-\Omega^2dudv$ with the conformal factor given by
\begin{align}
  \label{eq:11}
  \Omega^2\frac{\p u}{\p \tau}\frac{\p v}{\p \tau}=1.
\end{align}
The conditions \eqref{eq:9}, \eqref{eq:10} imply that along $\Sigma$
\begin{align}
  \label{eq:12}
  \frac{\p u}{\p \tau}=\frac{1}{f'+e^\omega},\qquad \frac{\p v}{\p \tau}=\frac{h'}{f'-e^\omega}.
\end{align}
Therefore the conformal factor induced by the soft phase along $\Sigma$ is
\begin{align}
  \label{eq:13}
  \Omega^2=\frac{(f')^2-e^{2\omega}}{h'}
\end{align}
in agreement with \eqref{eq:6}. Thus $\Omega$ is continuous across $\Sigma$.

In the soft phase the wave function $\phi$ coincides with $\tau$. Using equations \eqref{eq:9}, \eqref{eq:10} we find that the partial derivatives of the inverse of the transformation $(\tau,\chi)\mapsto(u,v)$ are given by
\begin{align}
  \label{eq:14}
  \frac{\p \tau}{\p u}=\frac{1}{2\p u/\p\tau},\qquad \frac{\p \tau}{\p v}=\frac{1}{2\p v/\p\tau},
\end{align}
\begin{align}
  \label{eq:15}
    \frac{\p \chi}{\p u}=\frac{1}{2\p u/\p\chi},\qquad \frac{\p \chi}{\p v}=\frac{1}{2\p v/\p\chi}.
\end{align}
In view of \eqref{eq:12} we deduce in particular, that along $\Sigma$
\begin{align}
  \label{eq:16}
  \frac{\p \tau}{\p u}=\frac{1}{2}(f'+e^{\omega_\ast}),\qquad   \frac{\p \tau}{\p v}=\frac{1}{2h'}(f'-e^{\omega_\ast}).
\end{align}
Differentiating \eqref{eq:5} we obtain
\begin{align}
  \label{eq:17}
  f'=\phi'_\ast=\phi_u+\phi_vh'.
\end{align}
Here and in the following paragraph we denote by subscripts $u$ and $v$ the partial derivatives with respect to $u$ and $v$ respectively. Using \eqref{eq:6} and the fact that along $\Sigma$
\begin{align}
  \label{eq:18}
  \Omega^2=4\phi_u\phi_v,
\end{align}
we then obtain
\begin{align}
  \label{eq:19}
  e^{2\omega_\ast}=(f')^2-\Omega_\ast^2h'=(\phi_u-\phi_vh')^2.
\end{align}
Hence (since $\phi_u-\phi_vh'>0$)
\begin{align}
  \label{eq:20}
  e^{\omega_\ast}=\phi_u-\phi_vh'.
\end{align}
Substituting \eqref{eq:17} and \eqref{eq:20} yields, along $\Sigma$
\begin{align}
  \label{eq:21}
  \frac{\p \tau}{\p u}=\frac{\p\phi}{\p u},\qquad   \frac{\p \tau}{\p v}=\frac{\p\phi}{\p v},
\end{align}
i.e.\ the partial derivatives of $\phi$ are continuous across $\Sigma$.

We have (see \eqref{eq:2})
\begin{align}
  \label{eq:22}
  r_\ast'=r_u+r_vh'.
\end{align}
In the following, we denote by a subscript $-$ a quantity along $\Sigma$ induced by the prior hard phase, by a subscript $+$ a quantity along $\Sigma$ induced by the subsequent soft phase, and by a prefix $\Delta$ the difference $+$ minus $-$ of corresponding quantities. Then by virtue of the first of conditions \eqref{eq:8} we have
\begin{align}
  \label{eq:23}
  \Delta r_u+h'\Delta r_v=0.
\end{align}
Also by virtue of the second of conditions \eqref{eq:8} we have
\begin{align*}
  \left(\frac{\p r}{\p \tau}\right)_+=(Ur)_-=\left[\frac{2}{\Omega^2}(\phi_v r_u+\phi_u r_u)\right]_-=\frac{1}{2}\left(\frac{r_u}{\phi_u}+\frac{r_v}{\phi_v}\right)_-.
\end{align*}
On the other hand, in view of \eqref{eq:14},
\begin{align*}
  \left(\frac{\p r}{\p\tau}\right)_+=\left(r_u\frac{\p u}{\p \tau}+r_v\frac{\p v}{\p\tau}\right)_+=\frac{1}{2}\left(\frac{r_u}{\tau_u}+\frac{r_v}{\tau_v}\right)_+,
\end{align*}
while, as we have shown above: $\tau_{u+}=\phi_{u-}$, $\tau_{v+}=\phi_{v-}$. It follows that
\begin{align}
  \label{eq:24}
  \frac{\Delta r_u}{\phi_{u-}}+\frac{\Delta r_v}{\phi_{v_-}}=0.
\end{align}
Equations \eqref{eq:23}, \eqref{eq:24} constitute a linear homogeneous system for the differences $\Delta r_u$, $\Delta r_v$. The matrix
\begin{align*}
  M=\left(
    \begin{array}{cc}
      1 & h'\\
      \frac{1}{\phi_u} & \frac{1}{\phi_v}
    \end{array}
\right)
\end{align*}
has determinant
\begin{align*}
  \det M=\frac{1}{\phi_v}-\frac{h'}{\phi_u}>0
\end{align*}
(since $h'\leq 0$). Consequently
\begin{align*}
  \Delta r_u=\Delta r_v=0,
\end{align*}
i.e.\ $dr$ is continuous across $\Sigma$.

Since the mass function $m$ is in general defined by
\begin{align*}
  1-\frac{2m}{r}=-\frac{4}{\Omega^2}\frac{\p r}{\p u}\frac{\p r}{\p v},
\end{align*}
and $\Omega$ as well as $dr$ are continuous across $\Sigma$, so is $m$
\begin{align}
  \label{eq:25}
  m(-f(\chi),\chi)=m_\ast(-\chi),\qquad m_\ast(u)=m(u,h(u)).
\end{align}

Next we shall show that the function $\rho$ determined in the soft phase by $m$ through the equations (eqs. (3.45c), (3.45d) of \cite{I})
\begin{align}
  \label{eq:26}
  \frac{\p m}{\p \tau}=0,\qquad \frac{\p m}{\p\chi}=4\pi r^2\rho\frac{\p r}{\p \chi}
\end{align}
is equal to 1 along $\Sigma$, therefore also $\rho$ is continuous across $\Sigma$. In view of the fact that along $\Sigma$: $\Omega^2=4\phi_u\phi_v$, the hard phase equations
\begin{align}
  \label{eq:27}
  m_u=2\pi r^2(r_u-4\Omega^{-2}\phi_u^2r_v),\qquad m_v=2\pi r^2(r_v-4\Omega^{-2}\phi_v^2r_u),
\end{align}
(eqs. (3.49a), (3.49b) of \cite{I}), reduce along $\Sigma$ to
\begin{align}
  \label{eq:28}
  m_u=\frac{2\pi r^2}{\phi_v}(\phi_v r_u-\phi_u r_v),\qquad m_v=\frac{2\pi r^2}{\phi_u}(\phi_u r_v-\phi_v r_u).
\end{align}
It follows that
\begin{align*}
  \frac{dm_\ast}{du}&=m_u+m_vh'\\
  &=2\pi r^2\left[ r_u+r_v h'-\left(\frac{\phi_u r_v}{\phi_v}+\frac{\phi_v r_u}{\phi_u}h'\right)\right].
\end{align*}
Since, along $\Sigma$,
\begin{align*}
  \frac{1}{2}\left(\frac{r_u}{\phi_u}+\frac{r_v}{\phi_v}\right)=\dot{r}_\ast,
\end{align*}
we can write this in the form
\begin{align*}
  \frac{dm_\ast}{du}&=4\pi r^2[r_u+r_vh'-\dot{r}_\ast(\phi_u+\phi_vh')]\\
  &=4\pi r^2\left(\frac{dr_\ast}{du}-\dot{r}_\ast\frac{d\phi_\ast}{du}\right).
\end{align*}
Therefore, by the soft phase equations \eqref{eq:26} and \eqref{eq:3}, \eqref{eq:4}, \eqref{eq:5}, \eqref{eq:8} together with \eqref{eq:25}, we have along $\Sigma$
\begin{align*}
  4\pi r^2\rho\frac{\p r}{\p\chi}=\frac{\p m}{\p \chi}=\frac{dm}{d\chi}=4\pi r^2\left(\frac{dr}{d\chi}+\frac{\p r}{\p \tau}\frac{df}{d\chi}\right)=4\pi r^2\frac{\p r}{\p \chi}.
\end{align*}
Consequently, $\rho=1$ along $\Sigma$ with respect to the soft phase and $\rho$ is continuous across $\Sigma$.

We have established this far the continuity across $\Sigma$ of $\Omega$, $\phi$, $r$, $d\phi$, $dr$, $m$ and $\rho$. It follows that besides the metric $g$, $r$ and $dr$, the thermodynamic variables $\rho$, $p$ and the fluid velocity $U$ are also continuous across $\Sigma$. This implies the continuity across $\Sigma$ of the energy tensor $T$, $S$ (see equations (3.19a), (3.19b) of \cite{I})
\begin{align*}
  T^{ab}  =(\rho+p)U^aU^b+pg^{ab},\qquad S=p.
\end{align*}
The Hessian system (equation (3.6b) of \cite{I}) reads
\begin{align*}
  \nabla_a\nabla_b r=(1/2r)(1-\p^cr\p_cr)g_{ab}-4\pi r(T_{ab}-g_{ab}\textrm{tr}T).
\end{align*}
The right hand side has been shown to be continuous across $\Sigma$. Therefore so is the left hand side, i.e.~the Hessian of $r$ is continuous across $\Sigma$. Now the derivative with respect to the parameter $u$ of the restriction to $\Sigma$ of $\p r/\p u$, that is
\begin{align*}
 \left( \frac{\p}{\p u}+h'\frac{\p}{\p v}\right)\frac{\p r}{\p u}=(\nabla\nabla r)_{uu}+h'(\nabla\nabla r)_{uv}+\frac{2}{\Omega}\frac{\p \Omega}{\p u}\frac{\p r}{\p u}
\end{align*}
must be continuous across $\Sigma$. It follows that $\p\Omega/\p u$, therefore also $\p\Omega/\p v$, the partial derivative of the metric $g_{ab}$ in the null coordinate system, and the 2nd partial derivative of $r$, are all continuous across $\Sigma$.

The unit future directed normal to $\Sigma$ is given by
\begin{align}
T=\frac{1}{\Omega_\ast\sqrt{-h'}}\left(\frac{\p}{\p u}-h'\frac{\p }{\p v}\right).
\end{align}
Also, the unit outward directed tangent to $\Sigma$ is
\begin{align}
  L=-\frac{1}{\Omega_\ast\sqrt{-h'}}\left(\frac{\p}{\p u}+h'\frac{\p }{\p v}\right).
\end{align}
The continuity of the connection across $\Sigma$ implies that the geodesic curvature of $\Sigma$, $g(\nabla_L T,L)$, induced from the two sides coincide.

Let $V=U^\ast$ be the outward directed spacelike vectorfield tangent to the simultaneous curves of the fluid. The geodesic curvature of the simultaneous curves is given by
\begin{align}
  \kappa=g(\nabla_VU,V).
\end{align}
In the soft phase, we have, in comoving coordinates
\begin{align}
\label{eq:29}
  U=\frac{\p }{\p \tau},\qquad V=e^{-\omega}\frac{\p}{\p \chi}.
\end{align}
Hence in the soft phase
\begin{align}
 \kappa =\frac{\p \omega}{\p \tau}.
\end{align}
Now $g(\nabla_UU,V)$ suffers a jump across $\Sigma$, for $\nabla_UU$ is given along $\Sigma$ on the hard phase side by
\begin{align*}
  (\nabla_UU)_-=-V(\log \|d\phi\|)_-V
\end{align*}
while $\nabla_UU$ vanishes in the soft phase
\begin{align*}
  (\nabla_UU)_+=0.
\end{align*}
Thus
\begin{align*}
  \Delta g(\nabla_UU,V)=g(\nabla_UU,V)_+-g(\nabla_UU,V)_-=V(\log \|d\phi\|)_-
\end{align*}
which is non-zero at a point of $\Sigma$ unless $\Sigma$ is tangent to a simultaneous curve at that point. Expressing, along $\Sigma$,
\begin{align*}
  V=\frac{L+g(L,U)U}{g(L,V)}
\end{align*}
and noting that $\nabla_LU$ is continuous across $\Sigma$, we conclude that the geodesic curvature of the simultaneous curves suffers a jump across $\Sigma$ given by
\begin{align}
  \Delta \kappa=\frac{g(L,U)}{g(L,V)}\Delta g(\nabla_UU,V)=\frac{g(L,U)}{g(L,V)}V(\log \|d\phi\|)_-.
\end{align}

Consider now a point $(-f(\chi_0),\chi_0)$ on $\Sigma$ at which
\begin{align*}
  \frac{\p r}{\p \chi}(-f(\chi_0),\chi_0)=0.
\end{align*}
This means that in the 3-dimensional spacelike hypersurface $H_0$ orthogonal to the fluid flow lines, which corresponds, in the 4-dimensional spacetime manifold $M$, to the simultaneous curve passing through the given point on $\Sigma$, there is a minimal sphere $S_0$ corresponding to that point. The $\tau\chi$- component of the Hessian equations (equations (3.33b) of \cite{I}), which in the soft phase reads
\begin{align*}
 \frac{\p^2 r}{\p\tau\p\chi}-\frac{\p \omega}{\p \tau}\frac{\p r}{\p\chi}=0,
\end{align*}
then implies that along the flow line $\chi=\chi_0$,
\begin{align*}
  \frac{\p r}{\p\chi}(\tau,\chi_0)=0
\end{align*}
as long as that flow line remains in the soft phase. Thus the flow lines through $S_0$ in $M$ intersects each orthogonal hypersurface $H$ in a neighborhood of $H_0$ to the future, at a sphere $S$ which is minimal in $H$.

We now focus attention to the behavior of the soft phase solution in a neighborhood of a boundary null point. We have the following remarkable fact:

\begin{proposition}
The boundary null points of a spacelike curve of transition from hard to soft are critical points of the function $\rho$ with respect to the subsequent soft phase.\label{proposition1}
\end{proposition}

\begin{proof}
Consider the case of an incoming boundary null point $N^-$. Along the spacelike curve $\Sigma$ we have $\sigma=1$ and $\Sigma$ is tangent at $N^-$ to the incoming null curve through $N^-$. Hence
\begin{align}
  \label{eq:30}
  \left(\frac{\p\sigma}{\p u}\right)_{N^-}=0.
\end{align}
In view of the fact that $\sigma$ is given in the hard phase by
\begin{align}
  \label{eq:31}
  \sigma^2=\|d\phi\|^2=\frac{4}{\Omega^2}\frac{\p\phi}{\p u}\frac{\p\phi}{\p v},
\end{align}
equation \eqref{eq:30} reads
\begin{align}
 \frac{\p \phi}{\p v}\frac{\p}{\p u}\left(\frac{1}{\Omega^2}\frac{\p\phi}{\p u}\right)+\frac{1}{\Omega^2}\frac{\p\phi}{\p u}\frac{\p^2\phi}{\p u\p v}=0\quad \textrm{:}\quad \textrm{at $N^-$}.
\end{align}
Substituting for $\p^2\phi /\p u\p v$ from the wave equation (equation (3.47) of \cite{I})
\begin{align}
  \frac{\p^2\phi}{\p u\p v}+\frac{1}{r}\frac{\p r}{\p u}\frac{\p \phi}{\p v}+\frac{1}{r}\frac{\p r}{\p v}\frac{\p \phi}{\p u}=0,
\end{align}
and using \eqref{eq:31} we obtain, in view of the fact that $\sigma_{N^-}=1$,
\begin{align}
  \label{eq:32}
  \frac{\p }{\p u}\left(\frac{1}{\Omega^2}\frac{\p\phi}{\p u}\right)=\frac{1}{\Omega^2}\frac{\p\phi}{\p u}\left[\frac{1}{r}\frac{\p r}{\p u}+\frac{4}{\Omega^2}\frac{1}{r}\frac{\p r}{\p v}\left(\frac{\p \phi}{\p u}\right)^2\right]\quad\textrm{:}\quad\textrm{at $N^-$}.
\end{align}

Since $\sigma=1$ along $\Sigma$ the fluid velocity is given by
\begin{align}
  \label{eq:33}
  U^a=-g^{ab}\frac{\p \phi}{\p x^b}
\end{align}
that is
\begin{align}
  \label{eq:34}
  U=\frac{2}{\Omega^2}\left(\frac{\p \phi}{\p v}\frac{\p}{\p u}+\frac{\p\phi}{\p u}\frac{\p}{\p v}\right).
\end{align}
The covariant derivative of $U$ along $\Sigma$ is given by the data along $\Sigma$ induced by the prior hard phase. Since $\p/\p u$ is tangent to $\Sigma$ at $N^-$, the derivative $\nabla_{\p/\p u}U$ at $N^-$ is given by the prior hard phase. Taking into account the fact that the only non-vanishing connection coefficients are
\begin{align*}
\Gamma_{uu}^u=\frac{2}{\Omega}\frac{\p\Omega}{\p u},\qquad \Gamma_{vv}^v=\frac{2}{\Omega}\frac{\p\Omega}{\p v}
\end{align*}
we find
\begin{align}
  \label{eq:35}
  (\nabla_{\p/\p u}U)^u&=\frac{\p U^u}{\p u}+\Gamma_{uu}^uU^u\notag\\
&=\frac{2}{\Omega^2}\frac{\p^2\phi}{\p u\p v}=-\frac{2}{r\Omega^2}\left(\frac{\p r}{\p u}\frac{\p \phi}{\p v}+\frac{\p r}{\p v}\frac{\p \phi}{\p u}\right)
\end{align}
at $N^-$, while, by virtue of \eqref{eq:32},
\begin{align}
  \label{eq:36}
  (\nabla_{\p/\p u}U)^v&=\frac{\p U^v}{\p u}\notag\\
&=2\frac{\p}{\p u}\left(\frac{1}{\Omega^2}\frac{\p\phi}{\p u}\right)=\frac{2}{\Omega^2}\frac{\p  \phi}{\p u}\left[\frac{1}{r}\frac{\p r}{\p u}+\frac{4}{\Omega^2}\frac{1}{r}\frac{\p r}{\p v}\left(\frac{\p \phi}{\p u}\right)^2\right]
\end{align}
at $N^-$. In view of \eqref{eq:34} we may write \eqref{eq:35} and \eqref{eq:36} in the form
\begin{align}
  (\nabla_{\p/\p u}U)^u&=-r^{-1}Ur\quad \textrm{:}\quad\textrm{at $N^-$,}\label{eq:37}\\
(\nabla_{\p/\p u}U)^v&=\frac{U^v}{r}\left[\frac{\p r}{\p u}+\Omega^2(U^v)^2\frac{\p r}{\p v}\right]\quad \textrm{:}\quad\textrm{at $N^-$.}\label{eq:38}
\end{align}

Now in the subsequent soft phase the velocity vectorfield $U$ is geodesic
\begin{align}
  \label{eq:39}
  \nabla_UU=0.
\end{align}
The $v$-component of this equation reads
\begin{align*}
  U^v(\nabla_{\p/\p v}U)^v=-U^u(\nabla_{\p/\p u}U)^v.
\end{align*}
Multiplying by $\Omega^2U^u$ and using the normalization condition
\begin{align*}
  \|U\|^2=\Omega^2 U^uU^v=1
\end{align*}
yields
\begin{align*}
  (\nabla_{\p/\p v}U)^v=-\Omega^2(U^u)^2(\nabla_{\p/\p u}U)^v.
\end{align*}
Thus, by \eqref{eq:38} at $N^-$ we have, coming from the soft phase
\begin{align*}
  (\nabla_{\p/\p v}U)^v=-\frac{U^u}{r}\left[\frac{\p r}{\p u}+\Omega^2(U^v)^2\frac{\p r}{\p v}\right]
\end{align*}
that is, in view of the normalization condition, simply
\begin{align}
  \label{eq:40}
  (\nabla_{\p/\p v}U)^v=-r^{-1}Ur\quad\textrm{:}\quad\textrm{at $N^-$}
\end{align}
coming from the soft phase. In conjunction with \eqref{eq:37} we then deduce
\begin{align*}
  \nabla\cdot U=(\nabla_{\p/\p u}U)^u+(\nabla_{\p/\p v}U)^v=-2r^{-1}Ur
\end{align*}
i.e.
\begin{align}
  \nabla\cdot (r^2 U)=0\quad\textrm{:}\quad\textrm{at $N^-$}\label{eq:41}
\end{align}
coming from the soft phase.

Now in the soft phase the function $\rho$ satisfies the equation
\begin{align}
  U\rho+\rho r^{-2}\nabla\cdot (r^2U)=0.\label{eq:42}
\end{align}
Thus \eqref{eq:41} implies that
\begin{align}
  U\rho=0\quad\textrm{:}\quad\textrm{at $N^-$}\label{eq:43}
\end{align}
coming from the soft phase. Since $\rho=1$ along $\Sigma$ and $\p/\p u$ is tangent to $\Sigma$ at $N^-$ we also have
\begin{align}
  \frac{\p \rho}{\p u}=0\quad\textrm{:}\quad\textrm{at $N^-$}.\label{eq:44}
\end{align}
We conclude that
\begin{align*}
  d\rho=0\quad\textrm{:}\quad\textrm{at $N^-$}
\end{align*}
coming from the soft phase, i.e.~$N^-$ is a critical point of the function $\rho$ with respect to the subsequent soft phase.
\end{proof}

We now define null coordinates $u$, $v$ which are canonical in a neighborhood of the point $N^-$, by requiring that
\begin{enumerate}
\item [1.]$u=v=0$ at $N^-$,
\item [2.]$u$ is an affine parameter along $C^{\ast -}$, the incoming null curve through $N^-$, and $v$ is an affine parameter along $C^-$, the outgoing null curve through $N^-$,
\item[3.]At $N^-$: $Uu=Uv=1$.
\end{enumerate}
These requirements determine the coordinates $u$, $v$ uniquely. The second requirement is equivalent to
\begin{align*}
  \Omega(u,0)=\Omega(0,v)=c
\end{align*}
where $c$ is a constant. The third requirement reads
\begin{align}
  \label{eq:45}
  U_{N^-}=\left(\frac{\p}{\p u}+\frac{\p}{\p v}\right)_{N^-}.
\end{align}
Since $\Omega^2U^uU^v=1$, this is equivalent to
\begin{align}
  \label{eq:46}
  \Omega(0,0)=1
\end{align}
therefore $c=1$ and we have
\begin{align}
  \label{eq:47}
  \Omega(u,0)=\Omega(0,v)=1.
\end{align}
Also, comparing \eqref{eq:34} with \eqref{eq:45} we obtain, in view of \eqref{eq:46},
\begin{align}
  \label{eq:48}
  \left(\frac{\p\phi}{\p u}\right)(0,0)=\left(\frac{\p\phi}{\p v}\right)(0,0)=\frac{1}{2}.
\end{align}

\footnote{It is not possible to define the geodesic curvature of $\Sigma$ at $N^-$, because there is no unit normal at a null point.}The canonical null coordinates allow us to define a notion of curvature of $\Sigma$ at $N^-$. $\Sigma$ is given by equation \eqref{eq:1}, where, since $N^-=(u^-,v^-)=(0,0)$, $h$ is a smooth function on $(u^+,0]$, $h'<0$ on $(u^+,0)$, $h'(0)=0$. We then define the curvature of $\Sigma$ at $N^-$ to be
\begin{align}
  \label{eq:49}
  k:=h''(0).
\end{align}
Note that $k\geq 0$, and we shall assume in the following that the generic case holds, i.e.
\begin{align}
  \label{eq:50}
  k>0.
\end{align}

By a transformation of the form $\phi \mapsto \phi-c$ we can set
\begin{align}
  \label{eq:51}
  \phi(0,0)=0.
\end{align}
Then the point $N^-$ is also the origin of the system of comoving coordinates $\tau$, $\chi$ and
\begin{align}
  \label{eq:52}
  f(0)=0.
\end{align}
Recalling that $\Sigma$ is given by
\begin{align*}
  \tau=-f(\chi)\,:\,\chi\in (\chi^-,\chi^+),\quad\chi^-=-u^-=0,\quad\chi^+=-u^+,
\end{align*}
we define (compare equation (5.13) of \cite{I})
\begin{align}
  \label{eq:53}
  \delta:=e^{-\omega_\ast}\frac{d f}{d \chi}.
\end{align}
We have
\begin{align}
  \label{eq:54}
  \delta^2<1\,:\,\textrm{in $(0,\chi^+)$},\quad \delta(0)=1,\quad \delta(\chi^+)=-1.
\end{align}
We also define (compare equation (6.39) of \cite{I})
\begin{align}
  \label{eq:55}
  q:=-e^{-\omega_\ast}\frac{d\delta}{d\chi}.
\end{align}

The vectorfields $U-V$, $U+V$ are both null and future directed, $U-V$ is incoming, $U+V$ is outgoing. Thus $U-V$ is in the direction of $\p/\p u$ and $U+V$ is in the direction of $\p/\p v$. Consequently,
\begin{align}
  \label{eq:56}
  \frac{r_u}{\phi_u}=\frac{Ur-Vr}{U\phi-V\phi},\qquad \frac{r_v}{\phi_v}=\frac{Ur+Vr}{U\phi+V\phi}.
\end{align}
Now,
\begin{align}
  \label{eq:57}
  U\phi=\|d\phi\|,\qquad V\phi=0.
\end{align}
Defining as in section 5 of \cite{I},
\begin{align}
  \label{eq:58}
  a_-:=-Ur+Vr,\qquad a_+:=Ur+Vr
\end{align}
we have
\begin{align}
  \label{eq:59}
  a_-a_+=-(Ur)^2+(Vr)^2=1-\mu,\qquad \mu=\frac{2m}{r}
\end{align}
and, by proposition 4.1 of \cite{I}
\begin{align}
  \label{eq:60}
  a_->0.
\end{align}

Since
\begin{align*}
  Ur=\dot{r},
\end{align*}
we can express, according to \eqref{eq:59},
\begin{align}
  \label{eq:61}
  Vr=\sqrt{1-\mu+\dot{r}^2}\quad\textrm{if}\quad Vr\geq 0,\qquad Vr=-\sqrt{1-\mu+\dot{r}^2}\quad\textrm{if}\quad Vr\leq 0.
\end{align}
Therefore, assuming that the first case holds, we can write
\begin{align}
  \label{eq:62}
  a_-=-\dot{r}+\sqrt{1-\mu+\dot{r}^2},\qquad a_+=\dot{r}+\sqrt{1-\mu+\dot{r}^2}.
\end{align}
By \eqref{eq:56}, \eqref{eq:57}, along $\Sigma$ we have
\begin{align}
  \label{eq:63}
  r_u=-a_{-\ast}\phi_u,\qquad r_v=a_{+\ast}\phi_v.
\end{align}
In terms of the functions
\begin{align}
\label{eq:64}
  \zeta=\frac{1}{\nu}\frac{\p \phi}{\p u},\quad\eta=\frac{1}{\kappa}\frac{\p\phi}{\p v},\quad\nu=-\frac{\p r}{\p u},\quad\kappa=\frac{1}{1-\mu}\frac{\p r}{\p v}
\end{align}
these equations read
\begin{align}
  \label{eq:65}
  \zeta_\ast=\frac{1}{a_{-\ast}},\qquad\eta_\ast=a_{-\ast}
\end{align}
(compare with equations 6.11c of \cite{I}).

According to the definition \eqref{eq:53},
\begin{align}
  \label{eq:66}
  f'=e^{\omega_\ast}\delta.
\end{align}
Substituting in equation \eqref{eq:6} yields
\begin{align}
  \label{eq:67}
  e^{\omega_\ast}=\frac{\Omega_\ast\sqrt{-h'}}{\sqrt{1-\delta^2}}.
\end{align}
From \eqref{eq:16}, by virtue of the continuity of the first derivatives of $\phi$, we have along $\Sigma$
\begin{align}
  \label{eq:68}
  \phi_u=\frac{1}{2}(f'+e^{\omega_\ast}),\qquad\phi_v=\frac{1}{2h'}(f'-e^{\omega_\ast})
\end{align}
that is
\begin{align}
  \label{eq:69}
  f'=\phi_u+h'\phi_v,\qquad e^{\omega_\ast}=\phi_u-h'\phi_v.
\end{align}
Note that in equations \eqref{eq:69}, when the argument of the left hand sides is $\chi$, the argument of the right hand sides is $(u=-\chi,v=h(-\chi))$.

At $N^-$, in view of \eqref{eq:48} we obtain
\begin{align}
  \label{eq:70}
  f'(0)=e^{\omega_\ast(0)}=\frac{1}{2}.
\end{align}
Also differentiating \eqref{eq:69} and evaluating the result at $N^-$, taking into account \eqref{eq:49}, yields
\begin{align}
  \label{eq:71}
  f''(0)=-\phi_{uu}(0,0)-\frac{k}{2},\qquad \frac{\omega_\ast'(0)}{2}=-\phi_{uu}(0,0)+\frac{k}{2}.
\end{align}
Hence, since by \eqref{eq:53},
\begin{align*}
  \delta'=e^{-\omega_\ast}(f''-\omega_\ast'f'),
\end{align*}
we have
\begin{align*}
  \delta'(0)=2\left(f''(0)-\frac{\omega_\ast'(0)}{2}\right)=-2k
\end{align*}
therefore, by \eqref{eq:55}, at $N^-$
\begin{align}
  \label{eq:72}
  q(0)=4k.
\end{align}

The Hessian equations in the null coordinate system (equations 3.48a, 3.48b, 3.48c of \cite{I}) read
\begin{align}
  r_{uu}-2\Omega^{-1}\Omega_ur_u&=-4\pi r\phi_u^2,  \label{eq:73}\\
 rr_{uv}+r_ur_v&=(1/4)(4\pi r^2-1)\Omega^2,\label{eq:74}\\
  r_{vv}-2\Omega^{-1}\Omega_vr_v&=-4\pi r\phi_v^2.\label{eq:75}
\end{align}
These equations are to be considered together with the wave equation in the same system (equation 3.47 of \cite{I})
\begin{align}
  \label{eq:76}
  r\phi_{uv}+r_u\phi_v+r_v\phi_u=0.
\end{align}

By \eqref{eq:47} we have
\begin{align}
  \label{eq:77}
  \Omega_u(0,0)=\Omega_v(0,0)=0.
\end{align}
Let us denote
\begin{align}
  \label{eq:78}
  r(0,0)&=r_0,& m(0,0)&=m_0,& \mu(0,0)&=\mu_0=\frac{2m_0}{r_0},\\
  \label{eq:79}
a_-(0,0)&=a_{-0},& a_+(0,0)&=a_{+0},& \dot{r}(0,0)&=\dot{r}_0=\frac{1}{2}(a_{+0}-a_{-0}).
\end{align}
Then by \eqref{eq:48} and \eqref{eq:63} we have at $N^-$
\begin{align}
  \label{eq:80}
  r_u(0,0)=-\frac{a_{-0}}{2},\qquad r_v(0,0)=\frac{a_{+0}}{2}.
\end{align}
Equations \eqref{eq:73}, \eqref{eq:74}, \eqref{eq:75} together with \eqref{eq:48} \eqref{eq:77} and \eqref{eq:80} yield the following expression for the Hessian of $r$ at the point $N^-$
\begin{align}
  r_{uu}(0,0)&=r_{vv}(0,0)=-\pi r_0,  \label{eq:81}\\
  r_{uv}(0,0)&=(1/4 r_0)(4\pi r_0^2-\mu_0).  \label{eq:82}
\end{align}

From \eqref{eq:32} we obtain
\begin{align}
  \label{eq:83}
  \phi_{uu}(0,0)=\frac{\dot{r}_0}{2r_0}.
\end{align}
Also,
\begin{align}
  \label{eq:84}
  \phi_{uv}(0,0)=-\frac{\dot{r}_0}{2r_0}.
\end{align}
From \eqref{eq:71} we then obtain
\begin{align}
  \label{eq:85}
  f''(0)=-\frac{\dot{r}_0}{2r_0}-\frac{k}{2},\qquad \frac{\omega_\ast'(0)}{2}=-\frac{\dot{r}_0}{2r_0}+\frac{k}{2}.
\end{align}
As a consequence of proposition \ref{proposition1}, the Taylor expansion of the soft phase energy density function $\rho$ at the point $N^-=(0,0)$ begins with quadratic terms. According to \eqref{eq:26} $\rho$ is given by
\begin{align}
  \label{eq:86}
  \rho=\frac{\p m/\p \chi}{4\pi r^2\p r/\p\chi}
\end{align}
and by \eqref{eq:26},
\begin{align}
  \label{eq:87}
  m(\tau,\chi)=m_\ast(-\chi),
\end{align}
therefore we can write
\begin{align}
  \label{eq:88}
  \rho=\frac{dm_\ast/d\chi}{4\pi r^2 \p r/\p \chi}.
\end{align}
In view of the fact that $\rho(-f(\chi),\chi)=1$, we have, along $\Sigma$
\begin{align}
  \label{eq:89}
  \frac{dm_\ast}{d\chi}=4\pi r^2\left(\frac{dr_\ast}{d\chi}+\dot{r}_\ast\frac{df}{d\chi
}\right).
\end{align}
It follows from the above that to obtain the 2nd partial derivatives of $\rho$ at $N^-$ we must determine at $N^-$, relative to the soft phase, the partial derivatives of $r$ up to the 3rd order, with the exception of $r_{\chi\chi\chi}$. Here and in the following paragraph we denote by the subscripts $\tau$ and $\chi$ the partial derivatives with respect to $\tau$ and $\chi$ respectively.

Differentiating the equation \eqref{eq:2}
\begin{align}
  \label{eq:90}
  r_\ast(u)=r(u,h(u)),
\end{align}
with respect to $\chi=-u$ twice, we obtain, in turn,
\begin{align}
  \label{eq:91}
  -\frac{dr_\ast}{d\chi}=r_u+r_vh',\qquad  \frac{d^2r_\ast}{d\chi^2}=r_{uu}+2r_{uv}h'+r_{vv}h'^2+r_vh''.
\end{align}
Evaluating the results at $N^-=(0,0)$ using \eqref{eq:80}, \eqref{eq:81}, \eqref{eq:82}, \eqref{eq:49} yields
\begin{align}
  \label{eq:92}
  \left(\frac{dr_\ast}{d\chi}\right)(0)=\frac{a_{-0}}{2},\qquad  \left(\frac{d^2r_\ast}{d\chi^2}\right)(0)=-\pi r_0+\frac{a_{+0}}{2}k.
\end{align}
Differentiating the equation \eqref{eq:2}
\begin{align}
  \label{eq:93}
  \dot{r}_\ast(u)=(Ur)(u,h(u))=(2\Omega^{-2}(\phi_vr_u+\phi_u r_v))(u,h(u))
\end{align}
with respect to $\chi=-u$, yields
\begin{align}
  \label{eq:94}
  -\frac{d\dot{r}_\ast}{d\chi}&=\,2\Omega^{-2}\Big[\phi_v(r_{uu}-2\Omega^{-1}\Omega_ur_u)+h'\phi_u(r_{vv}-2\Omega^{-1}\Omega_vr_v)\notag\\
&\qquad\qquad+r_v(\phi_{uu}-2\Omega^{-1}\Omega_u\phi_u)+h'r_u(\phi_{vv}-2\Omega^{-1}\Omega_v\phi_v)\notag\\
&\hspace{40mm}+\phi_{uv}(r_u+h'r_v)+r_{uv}(\phi_u+h'\phi_v)\Big].
\end{align}
Evaluating the result at $N^-$ we obtain, by \eqref{eq:77}, \eqref{eq:80}, \eqref{eq:81}, \eqref{eq:82}, \eqref{eq:83}, \eqref{eq:84},

\begin{align}
  \left(\frac{d\dot{r}_\ast}{d\chi}\right)(0)&=-2(\phi_v r_{uu}+\phi_u r_{uv}+r_v \phi_{uu}+r_u\phi_{uv})(0,0)\notag\\
&=\frac{1}{4r_0}(\mu_0-a_{+0}^2+a_{-0}^2).  \label{eq:95}
\end{align}

Differentiating the equation
\begin{align}
  \label{eq:96}
  r(-f(\chi),\chi)=r_\ast(-\chi),
\end{align}
we obtain, in turn,
\begin{align}
  \label{eq:97}
  (-r_\tau f'+r_\chi)(-f(\chi),\chi)=\frac{dr_\ast}{d\chi}(-\chi),
\end{align}
and differentiating once more, in view of equation \eqref{eq:7},
\begin{align}
  \label{eq:98}
  \left(-\frac{m}{r^2}f'^2-2r_{\tau\chi}f'+r_{\chi\chi}-r_\tau f''\right)(-f(\chi),\chi)=\frac{d^2r_\ast}{d\chi^2}(-\chi).
\end{align}
Also, differentiating the equation
\begin{align}
  \label{eq:99}
  r_\tau(-f(\chi),\chi)=\dot{r}_\ast(-\chi)
\end{align}
we obtain, in view of equation \eqref{eq:7},
\begin{align}
  \label{eq:100}
  \left(\frac{m}{r^2}f'+r_{\tau\chi}\right)(-f(\chi),\chi)=\frac{d\dot{r}_\ast}{d\chi}(-\chi).
\end{align}
Evaluating the above equations at $N^-=(0,0)$ and using \eqref{eq:92}, \eqref{eq:95}, \eqref{eq:85} yields
\begin{align}
  r_\tau(0,0)&=\dot{r}_0=\frac{1}{2}(a_{+0}-a_{-0}),\qquad\qquad\qquad r_\chi(0,0)=\frac{1}{4}(a_{+0}+a_{-0}),  \label{eq:101}\\
r_{\tau\tau}(0,0)&=-\frac{\mu_0}{2r_0}=-\frac{1-a_{+0}a_{-0}}{2r_0},\,\qquad\qquad r_{\tau\chi}(0,0)=-\frac{a_{+0}^2-a_{-0}^2}{4r_0},  \label{eq:102}\\
r_{\chi\chi}(0,0)&=\frac{1}{8r_0}(-3a_{+0}^2+a_{+0}a_{-0}+a_{-0}^2+1)-\pi r_0+\frac{k}{4}(a_{+0}+a_{-0}).\label{eq:103}
\end{align}

Differentiating equation \eqref{eq:86} with respect to $\tau$ we obtain
\begin{align}
  \label{eq:104}
  \rho_\tau=-\frac{dm_\ast/d\chi}{4\pi r^3 r_\chi^2}(rr_{\tau\chi}+2r_\tau r_\chi).
\end{align}
Also, differentiating the equation
\begin{align}
  \label{eq:105}
  \rho(-f(\chi),\chi)=1
\end{align}
yields
\begin{align}
  \label{eq:106}
  (-\rho_\tau f'+\rho_\chi)(-f(\chi),\chi)=0.
\end{align}
Therefore
\begin{align}
  \label{eq:107}
  \rho_\tau(0,0)=\rho_\chi(0,0)=0
\end{align}
in agreement with proposition \ref{proposition1}.

Since the factor in parenthesis in \eqref{eq:104} vanishes at $N^-$, we have
\begin{align}
  \label{eq:108}
  \rho_{\tau\tau}(0,0)=-\left(\frac{dm_\ast/d\chi}{4\pi r^3r_\chi^2}\right)(rr_{\tau\tau\chi}+3r_\tau r_{\tau\chi}+2r_\chi r_{\tau\tau})(0,0).
\end{align}
Using \eqref{eq:7}, \eqref{eq:101}, \eqref{eq:102} we find that the second factor on the right is equal to
\begin{align*}
  -\frac{(a_{+0}+a_{-0})}{4r_0}\left[\frac{3}{2}(a_{+0}-a_{-0})^2+4\pi r_0^2\right]
\end{align*}
while according to \eqref{eq:88} evaluated at $N^-$, the first factor is equal to
\begin{align*}
  \frac{4}{r_0(a_{+0}+a_{-0})}.
\end{align*}
Hence
\begin{align}
  \label{eq:109}
  \rho_{\tau\tau}(0,0)=\frac{1}{r_0^2}\left[\frac{3}{2}(a_{+0}-a_{-0})^2+4\pi r_0^2\right].
\end{align}

We have
\begin{align}
  \label{eq:110}
  \rho_{\tau\chi}(0,0)=-\left(\frac{dm_\ast /d\chi}{4\pi r^3 r_\chi^2}\right)(rr_{\tau\chi\chi}+3r_\chi r_{\tau\chi}+2r_\tau r_{\chi\chi})(0,0).
\end{align}
To determine $r_{\tau\chi\chi}(0,0)$ we differentiate equation \eqref{eq:100} obtaining
\begin{align*}
  \left[\frac{2dm_\ast/d\chi}{r^2}-\frac{2m}{r^3}\left(\frac{dr_\ast}{d\chi}+r_\chi\right)\right]\frac{df}{d\chi}+\frac{m}{r^2}\frac{d^2f}{d\chi^2}+r_{\tau\chi\chi}=\frac{d^2\dot{r}_\ast}{d\chi^2}.
\end{align*}
Evaluating this at $N^-=(0,0)$ with the help of \eqref{eq:85} yields
\begin{align}
  \label{eq:111}
  r_{\tau\chi\chi}(0,0)=\frac{d^2\dot{r}_\ast}{d\chi^2}(0)+\frac{(a_{+0}+a_{-0})}{4r_0^2}(\mu_0-4\pi r_0^2)+\frac{\mu_0}{4r_0}k.
\end{align}

To proceed we must determine $(d^2\dot{r}_\ast/d\chi^2)(0)$. This is to be obtained by differentiating equation \eqref{eq:94} and evaluating the result at $N^-$. We first substitute equations \eqref{eq:73}-\eqref{eq:76} into \eqref{eq:94} reducing the latter to the form
\begin{align}
  \label{eq:112}
  -\frac{d\dot{r}_\ast}{d\chi}&=\,2\Omega^{-2}\Big[-4\pi r\phi_v\phi_u^2-4\pi rh'\phi_u\phi_v^2\notag+r_v(\phi_{uu}-2\Omega^{-1}\Omega_u\phi_u)\notag\\
&\qquad\hspace{9mm}+h'r_u(\phi_{vv}-2\Omega^{-1}\Omega_v\phi_v)-r^{-1}(r_u\phi_v+r_v\phi_u)(r_u+h'r_v)\notag\\
&\qquad\hspace{35mm}-r^{-1}\left(r_ur_v+\tfrac{1}{4}\Omega^2(1-4\pi r^2)\right)(\phi_u+h'\phi_v)\Big].
\end{align}
Here the right hand side is evaluated at $v=h(u)$, $u=-\chi$.

Let us define
\begin{align}
  \label{eq:113}
  j=-\frac{\p(\sigma^2)}{\p v}(0,0).
\end{align}
The fact that $\sigma^2=1$ on $\Sigma$ while $\sigma^2\geq 1$ in the past of $\Sigma$ implies $j\geq0$ and we shall assume that the generic case holds, i.e.
\begin{align}
  \label{eq:114}
  j>0.
\end{align}
Differentiating the equation
\begin{align}
  \label{eq:115}
  \sigma^2(u,h(u))=1
\end{align}
twice, we obtain
\begin{align}
  \label{eq:116}
  \left(\frac{\p^2 (\sigma^2)}{\p u^2}+2h'\frac{\p^2 (\sigma^2)}{\p u\p v}+(h')^2\frac{\p^2 (\sigma^2)}{\p v^2}+h''\frac{\p (\sigma^2)}{\p v}\right)(u,h(u))=0.
\end{align}
Evaluating the above at $N^-$ using \eqref{eq:49}, \eqref{eq:113} yields
\begin{align}
  \label{eq:117}
  \frac{\p^2 (\sigma^2)}{\p u^2}(0,0)=jk=:i.
\end{align}

We have (recall \eqref{eq:31})
\begin{align*}
 \frac{\p (\sigma^2)}{\p v}=4\Omega^{-2}[\phi_u(\phi_{vv}-2\Omega^{-1}\Omega_v\phi_v)-r^{-1}\phi_v(r_u\phi_v+r_v\phi_u)].
\end{align*}
Evaluating this at $N^-$ using \eqref{eq:113} yields
\begin{align}
  \label{eq:118}
  \phi_{vv}(0,0)=-\frac{j}{2}+\frac{1}{4r_0}(a_{+0}-a_{-0}).
\end{align}
Also differentiating the equation
\begin{align*}
  \frac{\p(\sigma^2)}{\p u}=4\Omega^{-2}[\phi_v(\phi_{uu}-2\Omega^{-1}\Omega_u\phi_u)-r^{-1}\phi_u(r_u\phi_v+r_v\phi_u)]
\end{align*}
with respect to $u$ and evaluating the result at $N^-$ with the help of \eqref{eq:81} \eqref{eq:82}, \eqref{eq:83}, \eqref{eq:84}, we obtain, in view of \eqref{eq:117}, and the fact that by \eqref{eq:47} $\Omega_{uu}=0,$
\begin{align}
  \label{eq:119}
  \phi_{uuu}(0,0)=\frac{i}{2}+\frac{1}{8 r_0^2}[3a_{+0}(a_{+0}-a_{-0})-\mu_0].
\end{align}

Going back to \eqref{eq:112} we differentiate once again and evaluate the result at $\chi=0$ with the help of \eqref{eq:118} and \eqref{eq:119}. This yields
\begin{align}
  \frac{d^2\dot{r}_\ast}{d\chi^2}(0)=&\,\frac{1}{2}(a_{+0}+a_{-0})i+\frac{k}{4r_0}(a_{-0}^2-a_{+0}^2+a_{-0}a_{+0}-1)\notag\\
&\qquad+\frac{1}{8r_0^2}(3a_{+0}^3+a_{+0}^2a_{-0}+2a_{+0}a_{-0}^2-2a_{-0}^3-3a_{+0}-a_{-0})+\frac{\pi}{2}(3a_{+0}-a_{-0}).\label{eq:120}
\end{align}
Substituting in \eqref{eq:111} we then obtain
\begin{align}
  \label{eq:121}
  r_{\tau\chi\chi}(0,0)&=\frac{i}{2}(a_{+0}+a_{-0})+\frac{k}{4r_0}(a_{-0}^2-a_{+0}^2)\notag\\
&\qquad+\frac{1}{8 r_0^2}(3a_{+0}^3-a_{+0}^2a_{-0}-2a_{-0}^3-a_{+0}+a_{-0})+\frac{\pi}{2}(a_{+0}-3a_{-0}.)
\end{align}
Substituting in turn in \eqref{eq:110} and taking into account also \eqref{eq:102}, \eqref{eq:103}, we arrive at the expression
\begin{align}
  \label{eq:122}
  \rho_{\tau\chi}(0,0)=-2i+\frac{3}{4r_0^2}(a_{+0}-a_{-0})^2+2\pi.
\end{align}
Finally $\rho_{\chi\chi}(0,0)$ is obtained by differentiating equation \eqref{eq:106} and evaluating the result at $\chi=0$. This yields
\begin{align}
  \label{eq:123}
  ((1/4)\rho_{\tau\tau}-\rho_{\tau\chi}+\rho_{\chi\chi})(0,0)=0.
\end{align}
Setting
\begin{align}
  \label{eq:124}
  i_0:=\frac{3}{2}\frac{\dot{r}_0^2}{r_0^2}+\pi=\frac{1}{r_0^2}\left[\frac{3}{8}(a_{+0}-a_{-0})^2+\pi r_0^2\right]
\end{align}
we can express \eqref{eq:109}, \eqref{eq:122} in the form
\begin{align}
  \label{eq:125}
  \rho_{\tau\tau}(0,0)=4i_0,\qquad \rho_{\tau\chi}(0,0)=2i_0-2i,
\end{align}
and, by \eqref{eq:123}, we have
\begin{align}
  \label{eq:126}
  \rho_{\chi\chi}(0,0)=i_0-2i.
\end{align}
The initial data along $\Sigma$ determine a soft phase solution in the domain $\mathcal{U}_0=\{(\tau,\chi):\tau>-f(\chi),0<\chi<\chi_+\}$ in the future of $\Sigma$. Here $\chi_+=-u_+$. By \eqref{eq:107} and \eqref{eq:125}, \eqref{eq:126}, the Taylor expansion of the soft phase energy density function $\rho$ in a neighborhood of the point $N^-=(0,0)$ takes the form
\begin{align}
  \label{eq:127}
  \rho(\tau,\chi)=1+2i_0\left(\tau+\frac{\chi}{2}\right)\left(\tau-\frac{l\chi}{2}\right)+\ldots.
\end{align}
Here,
\begin{align}
  \label{eq:128}
  l:=2(i/i_0)-1
\end{align}
and $\ldots$ denotes terms of degree higher than the second.

\label{discussion_1}It follows that only a subdomain $\mathcal{U}\subset \mathcal{U}_0$ bounded in the past by $\Sigma$ and in the future by a curve tangent at $N^-$ to the line $\tau=l\chi/2$, corresponds to genuine soft phase, i.e.~to $\rho<1$. We note that, in view of \eqref{eq:70}, the line $\tau=-\chi/2$ is tangent to $\Sigma$ at $N^-$, while the line $\tau=l\chi/2$, which lies in its future, is spacelike, null outgoing, or timelike at $N^-$ according as to whether $l<1$, $l=1$, or $l>1$ respectively. Thus in the case $l<1$ the future boundary of $\mathcal{U}$ is a spacelike curve $\Sigma'$ through $N^-$ along which the soft phase changes continuously back to a hard phase and no new difficulties are encountered. Discounting the exceptional case $l=1$, we shall therefore confine our attention in the following to the other generic case, namely the case $l>1$. The future boundary of $\mathcal{U}$ is then a timelike curve $\tau=\hat{\tau}(\chi)$ through $N^-$, $\hat{\tau}_\ast(0)=0$, $\hat{\tau}_\ast'(0)=l/2$, and in the causal future of $N^-$ the following \textit{free boundary problem} is posed:
\vspace{5mm}

\textit{Find a future domain of development $\mathcal{V}$ of $N^-$, a timelike curve $\mathscr{B}$ in $\mathcal{V}$ issuing from $N^-$, and a solution in $\mathcal{V}$ such that in $\mathcal{V}_1$, the part of $\mathcal{V}$ bounded by $\mathscr{B}$ and $C^{\ast-}$, the incoming null curve issuing from $N^-$, it is a hard phase solution taking given initial data along $C^{\ast-}$, in $\mathcal{V}_0=\mathcal{V}\setminus\overline{\mathcal{V}_1}$ it coincides with the above soft phase solution restricted to a subdomain of $\mathcal{U}$ bounded in the past by $C^-$, the outgoing null curve issuing from $N^-$, while across $\mathscr{B}$ the jump conditions are fulfilled.}

\section{Existence}

We recall the barrier function $e$ defined in \cite{III} (equation (1.23a)):
\begin{align}
  \label{eq:129}
  e=\frac{\xi}{2\zeta^2}+\frac{1}{2r}\left[\frac{1}{\zeta}-(1-\mu-4\pi r^2)\zeta\right],
\end{align}
where $\xi$ is given by
\begin{align}
  \label{eq:130}
  \frac{\p \zeta}{\p u}=\nu\xi.
\end{align}
Let $u$, $v$ be the canonical null coordinates in a neighborhood of the point $N^-$ defined in the preceding section.
\begin{proposition}
At the boundary null point $N^-$ the barrier function $e$ vanishes\label{proposition2}
\begin{align}
  \label{eq:131}
  e|_{N^-}=0,
\end{align}
while its partial derivatives with respect to the prior hard phase are given by
\begin{align}
  \label{eq:132}
  \left(\frac{\p e}{\p u}\right)_{N^-}=i,\qquad \left(\frac{\p e}{\p v}\right)_{N^-}=2i_0,
\end{align}
where $i$ and $i_0$ are the positive real numbers defined by \eqref{eq:117} and \eqref{eq:124} respectively.
\end{proposition}
\begin{proof}
From \eqref{eq:65} at $N^-=(0,0)$ we have
\begin{align}
  \label{eq:133}
  \zeta(0,0)=\frac{1}{a_{-0}}.
\end{align}
Now,
\begin{align}
  \label{eq:134}
  \xi=\frac{1}{r_u}\frac{\p}{\p u}\left(\frac{\phi_u}{r_u}\right)=r_u^{-3}(r_u\phi_{uu}-\phi_ur_{uu}).
\end{align}
Substituting \eqref{eq:48}, \eqref{eq:80}, \eqref{eq:81}, \eqref{eq:83}, we find that at $N^-$
\begin{align}
  \label{eq:135}
  \xi(0,0)=\frac{2}{a_{-0}^3r_0}(a_{-0}\dot{r}_0-2\pi r_0^2).
\end{align}
Evaluating \eqref{eq:129} at $N^-$ using \eqref{eq:133}, \eqref{eq:135}, and taking account of the formulas
\begin{align}
  \label{eq:136}
  \dot{r}=\frac{1}{2}(a_+-a_-),\qquad 1-\mu=a_+a_-,
\end{align}
yields
\begin{align}
  \label{eq:137}
  e(0,0)=0.
\end{align}
We define the function $\xi_-$ by
\begin{align}
  \label{eq:138}
  \frac{\p \xi}{\p u}=\nu\xi_-.
\end{align}
Taking into account the fact that by equations (6.5a), (6.3a) of \cite{I},
\begin{align}
  \label{eq:139}
  \frac{\p\mu}{\p u}=\frac{\nu}{r}\{\mu-4\pi r^2[(1-\mu)\zeta^2+1]\}
\end{align}
we can express, in terms of $\xi_-$,
\begin{align}
  \label{eq:140}
  e_-&:=\frac{1}{\nu}\frac{\p e}{\p u}\notag\\
&\,\,=\frac{\xi_-}{2\zeta^2}-\frac{\xi^2}{\zeta^3}-\frac{\xi}{2r}\left(\frac{1}{\zeta^2}+1-\mu-4\pi r^2\right)+\frac{1}{2r^2}\left(\frac{1}{\zeta}+(2\mu-1-8\pi r^2)\zeta-4\pi r^2(1-\mu)\zeta^3\right).
\end{align}
Differentiating \eqref{eq:134} with respect to $u$ we obtain
\begin{align}
  \label{eq:141}
  \frac{\p \xi}{\p u}=-3r_u^{-1}r_{uu}\xi+r_u^{-3}(r_u\phi_{uuu}-\phi_ur_{uuu}).
\end{align}
At $N^-$, by \eqref{eq:73} and \eqref{eq:83}, and the fact that by \eqref{eq:47} $\Omega_{uu}=0$,
\begin{align}
  \label{eq:142}
  r_{uuu}(0,0)=-4\pi(r_u\phi_u^2+2r\phi_u\phi_{uu})(0,0)=\frac{\pi}{2}(a_{-0}-4\dot{r}_0).
\end{align}
Substituting \eqref{eq:142} as well as \eqref{eq:134} and \eqref{eq:119} in \eqref{eq:141} evaluated at $N^-$ we obtain
\begin{align}
  \label{eq:143}
  \xi_-(0,0)=\frac{4i}{a_{-0}^3}+\frac{1}{r_0^2}\left[\frac{12}{a_{-0}^3}\dot{r}^2_0+\left(\frac{6}{a_{-0}^2}-\frac{40\pi r_0^2}{a_{-0}^4}\right)\dot{r}_0+\frac{1}{a_{-0}^3}\left(-\mu_0+4\pi r_0^2+\frac{48\pi^2 r_0^4}{a_{-0}^2}\right)\right].
\end{align}
Substituting in turn \eqref{eq:143} in \eqref{eq:140} evaluated at $N^-$ and taking account of \eqref{eq:133}, \eqref{eq:135}, yields simply
\begin{align}
  \label{eq:144}
  \frac{\p e}{\p u}(0,0)=i.
\end{align}
We define the function $\xi_+$ by
\begin{align}
  \label{eq:145}
  \frac{\p \xi}{\p v}=\kappa\xi_+.
\end{align}
Taking into account the fact that by equations (6.5b), (6.3b) of \cite{I},
\begin{align}
  \label{eq:146}
  \frac{\p \mu}{\p v}=\frac{\kappa}{r}[(1-\mu)(4\pi r^2-\mu)+4\pi r^2\eta^2]
\end{align}
and according to equation (6.7b) of \cite{I}
\begin{align}
  \label{eq:147}
  \frac{\p \zeta}{\p v}=\frac{\kappa}{r}[\eta-(1-4\pi r^2)\zeta]
\end{align}
we can express, in terms of $\xi_+$
\begin{align}
e_+&:=\frac{1}{\kappa}\frac{\p e}{\p v}\notag\\
&\,\,=\,\frac{\xi_+}{2\zeta^2}-\frac{1}{r}[\eta-(1-4\pi r^2)\zeta]\left\{\frac{\xi}{\zeta^3}+\frac{1}{2r}\left[\frac{1}{\zeta^2}+(1-\mu-4\pi r^2)\right]\right\}\notag\\
&\qquad\,\,-\frac{(1-\mu)}{2r^2}\left[\frac{1}{\zeta}-(1-\mu-4\pi r^2)\zeta\right]+\frac{\zeta}{2r^2}\bigg[(1-\mu)(4\pi r^2-\mu)+4\pi r^2\eta^2+8\pi r^2(1-\mu)\bigg].\label{eq:148}
\end{align}
Differentiating \eqref{eq:134} with respect to $v$ we obtain
\begin{align}
  \label{eq:149}
  \frac{\p\xi}{\p v}=-3r_u^{-1}r_{uv}\xi+r_u^{-3}(r_{uv}\phi_{uu}-\phi_{uv}r_{uu}+r_u\phi_{uuv}-\phi_ur_{uuv}).
\end{align}
To obtain $r_{uuv}$, $\phi_{uuv}$ at $N^-$ we differentiate equations \eqref{eq:74}, \eqref{eq:76} with respect to $u$
\begin{align*}
rr_{uuv}+2r_ur_{uv}+r_{uu}r_v&=2\pi r\Omega^2r_u+(1/2)(4\pi r^2-1)\Omega\Omega_u,\\
r\phi_{uuv}+2r_u\phi_{uv}+r_v\phi_{uu}+r_{uu}\phi_v+r_{uv}\phi_u&=0.
\end{align*}
Evaluating these at $N^-=(0,0)$ with the help of \eqref{eq:77}, \eqref{eq:80}, \eqref{eq:81}, \eqref{eq:82}, \eqref{eq:83}, \eqref{eq:84}, yields
\begin{align}
  \label{eq:150}
r_{uuv}(0,0)&=\frac{1}{4r_0^2}(2\pi r_0^2a_{+0}-a_{-0}\mu_0),\\
\phi_{uuv}(0,0)&=\frac{1}{8r_0^2}(1-a_{+0}^2-2a_{+0}a_{-0}+2a_{-0}^2).
\end{align}
Substituting in \eqref{eq:149} and taking account of \eqref{eq:135} we obtain
\begin{align}
\xi_+(0,0)&=\,\frac{1}{a_{-0}^3r_0^2}\big[-48\pi^2r_0^4+12\pi r_0^2(1-a_{+0}a_{-0})+4\pi r_0^2a_{-0}(4a_{+0}-3a_{-0})\notag\\
&\qquad+a_{-0}(a_{-0}-2a_{+0}+a_{+0}^2a_{-0}-2a_{+0}a_{-0}^2+2a_{-0}^3)\big].  \label{eq:151}
\end{align}

Substituting in turn in \eqref{eq:148} and taking again into account \eqref{eq:135} and the fact that from \eqref{eq:65}
\begin{align}
  \label{eq:152}
  \eta(0,0)=a_{-0},
\end{align}
yields simply
\begin{align}
  \label{eq:153}
  e_+(0,0)=4a_{-0}i_0.
\end{align}
Taking then into account the fact that, by the second of \eqref{eq:80} and the relation \eqref{eq:59},
\begin{align}
  \label{eq:154}
  \kappa(0,0)=\frac{1}{2a_{-0}},
\end{align}
we obtain
\begin{align}
  \label{eq:155}
  \frac{\p e}{\p v}(0,0)=2i_0.
\end{align}

\end{proof}

In the following we use the same coordinate system as in \cite{III} which was introduced at the beginning of section 4 of \cite{II}.

Let us define $\mathscr{V}$ by
\begin{align}
  \label{eq:156}
  \mathscr{V}:=\{(\tau,\chi):\chi\geq 0,\tau_-(\chi)\leq\tau\leq\hat{\tau}(\chi)\}.
\end{align}
The point $\tau=\chi=0$ is the inner null end point $N^-$ of $\Sigma$, $\tau=\tau_-(\chi)$ is the equation of $C^-$, while $\tau=\hat{\tau}(\chi)$ is the curve issuing at $N^-$, where, corresponding to the given soft-phase initial data, $\rho(\tau,\chi)$ along each flow line first becomes equal to 1. Note that the curve $\tau=\hat{\tau}(\chi)$ corresponds to the future boundary of $\mathcal{U}$ (cf. the discussion on page \pageref{discussion_1}). We have
\begin{align}
  \label{eq:157}
  \tau_-(0)=\hat{\tau}(0)=0.
\end{align}
We cannot apply theorem 3.1 of \cite{II} to conclude the local existence of a solution to the free boundary problem. For, according to \eqref{eq:133}
\begin{align}
  \label{eq:158}
  \gamma(0)=\frac{1}{\zeta_\ast(0)a_{-\ast}(0)}=1.
\end{align}
In view of \eqref{eq:157} and \eqref{eq:158}, the hypotheses of theorem 3.1 of \cite{II} are not fulfilled.

To prove the existence of an expansion shock we construct an appropriate sequence of regular problems to which the continuation argument of \cite{II} applies. The solution to the formation problem will then be shown to be the limit of this sequence.

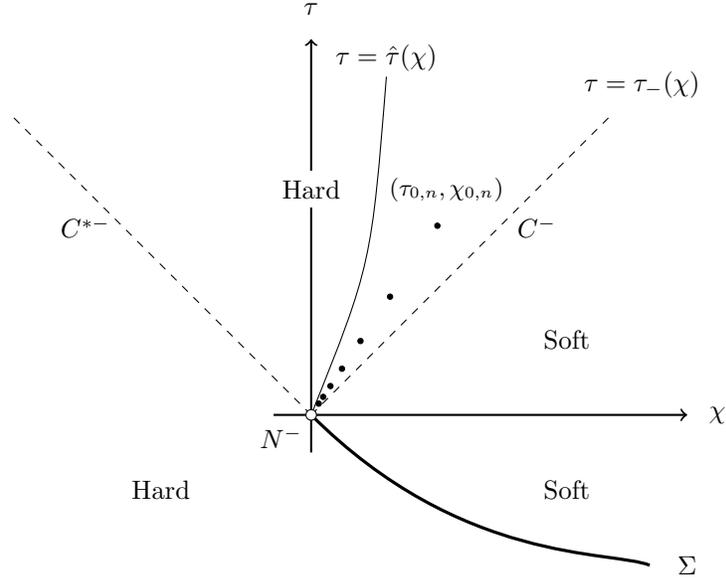
\begin{center}
\begin{tikzpicture}
\draw [->, thick](-0.5,0) -- (5,0);
\draw  [->, thick](0,-0.5) -- (0,5);
\node at (5.4,0) {$\chi$};
\node at (0,5.4) {$\tau$};
\draw [dashed] (0,0) -- (4,4);
\draw [dashed] (0,0) -- (-4,4);
\node at (-3,2.5) {$C^{\ast -}$};
\node at (3,2.5) {$C^-$};
\node at (4.4,4.4) {$\tau=\tau_-(\chi)$};
\draw [very thick] (0,0) .. controls (2,-2) and (4,-1.8) .. (4.5,-2);
\node at (5,-2) {$\Sigma$};
\draw (0,0) .. controls (0.8,2).. (1,4.5);
\node at (1,4.75) {$\tau=\hat{\tau}(\chi)$};
\filldraw (0.1,0.15) circle (1pt);
\filldraw (0.16,0.24) circle (1pt);
\filldraw (0.256,0.384) circle (1pt);
\filldraw (0.4096, 0.6144) circle (1pt);
\filldraw (0.65536, 0.98304) circle (1pt);
\filldraw (1.04858, 1.57286) circle (1pt);
\filldraw (1.67772, 2.51658) circle (1pt);
\node at (1.8,3) {\small{$(\tau_{0,n},\chi_{0,n})$}};
\node at (3.4,1) {Soft};
\node at (3.4,-1) {Soft};
\node at (0,3) [fill=white] {Hard};
\node at (-2,-1) {Hard};
\filldraw (0,0) [fill=white] circle (2pt);
\node at (-0.4,-0.3) {$N^-$};
\end{tikzpicture}
\end{center}

We pick a sequence of points $(\tau_{0,n},\chi_{0,n})$, monotone in both arguments, contained in the interior $\mathscr{V}$. Joining each point $(\tau_{0,n},\chi_{0,n})$ to $(0,0)$ by a timelike geodesic, we define $\beta_{0,n}$ to be the velocity at $(0,0)$ of the corresponding geodesic relative to the flow lines. We then have $0<\beta_{0,n}<1$. Furthermore, we arrange the way the sequence of points approach $N^-$ such that the following condition is satisfied.
\begin{align}
  \label{eq:159}
  \lim_{n\rightarrow\infty}\beta_{0,n}=\beta_0:=\frac{5+l}{5l+1}.
\end{align}
We note that, since $l>1$,
\begin{align}
  \label{eq:160}
  1<\frac{1}{\beta_0}<l,
\end{align}
i.e.~the above conditions are not in conflict with each other.

Let $\tau=\tau_{-n}(\chi)$ be the equation of $C_n^-$, the outgoing null curve issuing from the point $(\tau_{0,n},\chi_{0,n})$. We define the domains
 \begin{align}
   \label{eq:161}
   \mathscr{V}_n:=\{(\tau,\chi):\chi\geq \chi_{0,n},\tau_{-n}(\chi)\leq \tau\leq\hat{\tau}(\chi)\}.
 \end{align}
The domains $\mathscr{V}_n\subset \mathscr{V}$ are the images under the translation $(\tau,\chi)\mapsto (\tau+\tau_{0,n},\chi+\chi_{0,n})$ of the domains
\begin{align}
  \label{eq:162}
  \tilde{\mathscr{V}}_n:=\{(\tau,\chi):(\tau+\tau_{0,n},\chi+\chi_{0,n})\in\mathscr{V}_n\}.
\end{align}
On the domains $\tilde{\mathscr{V}}_n$ we define the soft-phase solution $(r_n,\omega_n,\rho_n)$ to be the corresponding pullbacks of the restriction of the soft-phase solution $(r,\omega,\rho)$ to the domains $\mathscr{V}_n$
\begin{align}
  \label{eq:163}
  r_n(\tau,\chi)&:=r(\tau+\tau_{0,n},\chi+\chi_{0,n}),\\
\omega_n(\tau,\chi)&:=\omega(\tau+\tau_{0,n},\chi+\chi_{0,n}),\\
\rho_n(\tau,\chi)&:=\rho(\tau+\tau_{0,n},\chi+\chi_{0,n}).\quad
\end{align}

We also must construct the corresponding hardphase initial data. This is accomplished by constructing the sequence $R_n(t)$, which defines $r$ as a function of $\phi$ along $C^{\ast -}$ for each regularized problem. This sequence must satisfy the conditions
\begin{align}
  \label{eq:164}
  R_n(0)=r_{0,n},\qquad \gamma_n(0)<1,\qquad e_{0,n}>0.
\end{align}
The first condition makes sure that $R_n(0)\rightarrow r_0$ as $n\rightarrow \infty$. The second condition is a hypothesis of theorem 3.1 of \cite{II}. The hypothesis $\hat{\tau}(\chi_{0,n}) > \tau_{0,n}$ of theorem 3.1 of \cite{II} is fulfilled because of the way the sequence of points converging to $N^-$ is chosen. The third condition provides a barrier which is used to prove lemma \ref{lemma_analogue} below.

We need three constants as we have three conditions. Let us set
\begin{align}
  \label{eq:165}
R_n(t):=R(t)+c_n+k_nt+\frac{1}{2}l_nt^2,
\end{align}
with
\begin{align}
  \label{eq:166}
c_n\rightarrow 0,\quad k_n\rightarrow 0 ,\quad l_n \rightarrow 0\quad\textrm{as}\quad n\rightarrow \infty.  
\end{align}
The first condition yields 
\begin{align}
    c_n=r_{0,n}-r_0.\label{eq:167}
\end{align}
The second condition reads
\begin{align} \gamma_n(0)=-\frac{\dot{R}_n(0)}{a_{-0,n}}<1.\end{align}
Therefore this condition becomes 

\begin{align}
    k_n>a_{-0}-a_{-0,n}.\label{eq:168}
\end{align}
The third condition can be written as
\begin{align}
  0<2\zeta_{0,n}^2e_{0,n}=\xi_{0,n}+\frac{\zeta_{0,n}}{r_{0,n}}[1-(1-\mu_{0,n}-4\pi r_{0,n}^2)\zeta_{0,n}^2].\label{eq:169}
  \end{align}
Here
\begin{align}
  \zeta_{0,n}=-\frac{1}{\dot{R}_n(0)}=-\frac{1}{\dot{R}(0)+k_n},\qquad \xi_{0,n}=-\frac{\ddot{R}_n(0)}{\dot{R}_n^3(0)}=-\frac{\ddot{R}(0)+l_n}{(\dot{R}(0)+k_n)^3}.
\end{align}
Since $e_0=\lim_{n\rightarrow\infty}e_{0,n}=0$, $R(0)=r_0$, $\dot{R}(0)=-a_{-0}$ it follows
\begin{align}
  \ddot{R}(0)=-\frac{1}{r_0}[a_{-0}^2-(1-\mu_0-4\pi r_0^2)].\label{eq:170}
\end{align}
Rewriting equation \eqref{eq:169} using equation \eqref{eq:170} we get
\begin{align}
  2\zeta_{0,n}^2e_{0,n}&=\,\frac{\ddot{R}(0)+l_n}{(a_{-0}-k_n)^3}-\frac{\ddot{R}(0)}{a_{-0}^3}\notag\\
&\qquad+\frac{1}{r_{0,n}(a_{-0}-k_n)}\left[1-\frac{1-\mu_{0,n}-4\pi r_{0,n}^2}{(a_{-0}-k_n)^2}\right]
-\frac{1}{r_0a_{-0}}\left[1-\frac{1-\mu_0-4\pi r_0^2}{a_{-0}^2}\right]>0.
\end{align}
Thus \eqref{eq:169} is equivalent to
\begin{align}
  l_n&>\left\{\frac{1}{r_0a_{-0}}\left[1-\frac{1-\mu_0-4\pi r_0^2}{a_{-0}^2}\right]-\frac{1}{r_{0,n}(a_{-0}-k_n)}\left[1-\frac{1-\mu_{0,n}-4\pi r_{0,n}^2}{(a_{-0}-k_n)^2}\right]\right\}(a_{-0}-k_n)^3\notag\\
&\qquad+\frac{\ddot{R}(0)}{a_{-0}^3}(a_{-0}-k_n)^3-\ddot{R}(0)=:F(r_{0,n},\mu_{0,n},k_n).
\end{align}
The third condition is then
\begin{align}
  l_n>F(r_{0,n},\mu_{0,n},k_n).\label{eq:171}
\end{align}
We have
\begin{align}
F(r_{0,n},\mu_{0,n},k_n)\rightarrow 0 \quad\textrm{as}\quad n\rightarrow \infty\quad \textrm{i.e.~as}\quad r_{0,n}\rightarrow r_0,\,\,\mu_{0,n}\rightarrow \mu_0,\,\,k_{0,n}\rightarrow 0.
\end{align}
Adjusting now the three constants $c_n$, $k_n$ and $l_n$ for each member of the sequence such that \eqref{eq:167}, \eqref{eq:168} and \eqref{eq:171} together with \eqref{eq:166} hold, which is certainly possible, implies the three conditions \eqref{eq:164} and therefore completes the construction of the hardphase initial data.

In the next few paragraphs are stated things that will be used in the proof of lemma \ref{lemma_analogue}.

On $C^{\ast -}$ for the $n$'th regularized problem we have
\begin{align}
  \label{eq:172}
  \phi&=t,& r(u,0)&=R_n(\phi(u,0)),& \zeta_n(u,0)&=Z_n(\phi(u,0)),\\
  m_n(u,0)&=M_n(\phi(u,0)),&\xi_n(u,0)&=\Xi_n(\phi(u,0)),& \xi_{-n}(u,0)&=\Xi_{-n}(\phi(u,0)).\label{eq:173}
\end{align}
By construction of $R_n(t)$ we have $R_n\rightarrow R$ uniformly as $n\rightarrow\infty$ with its derivatives up to third order. Since $Z_n=-1/\dot{R}_n$, it follows that $Z_n\rightarrow Z$ uniformly as $n\rightarrow\infty$. Similarly for $\Xi_n=-\dot{Z}_n/\dot{R}_n$ it follows that $\Xi_n \rightarrow \Xi$ uniformly as $n\rightarrow \infty$. Since
\begin{align}
  M_n(0)=m_{0,n},\qquad \dot{M}_n=2\pi R_n^2 \dot{R}_n[(1-2M_nR_n^{-1})Z_n^2+1],
\end{align}
(cf. equation (6.5a) of \cite{I}) it follows that $M_n\rightarrow M$ uniformly as $n\rightarrow\infty$. Since $\Xi_{-n}=-\dot{\Xi}_n/\dot{R}_n$, it follows that $\Xi_{-n}\rightarrow \Xi_-$ uniformly as $n\rightarrow \infty$. Along $C^{\ast -}$ we have $e_n=E_n(t)$, where
\begin{align}
  E_n=\frac{\Xi_n}{2Z_n^2}+\frac{1}{2R_n}\left[\frac{1}{Z_n}-\left(1-\frac{2M_n}{R_n}-4\pi R_n^2\right)Z_n\right].
\end{align}
It follows that $E_n\rightarrow E$ uniformly in $t$ as $n\rightarrow \infty$. By \eqref{eq:140} $e_-=e_-(\xi,\xi_-,\zeta,r,m)$ so along $C^{\ast -}$ we have $E_{-n}=E_{-n}(\Xi_n,\Xi_{-n},Z_n,R_n,M_n)$. It follows therefore from the form of $E_{-n}$ and from the uniform convergence of the arguments of $E_{-n}$ that $E_{-n}\rightarrow E_-$ uniformly in $t$ as $n\rightarrow \infty$.

We have
\begin{align}
    \xi_+=\frac{1}{\kappa}\frac{\p \xi}{\p v}=\frac{1}{\kappa}\frac{\p}{\p v}\left(\frac{1}{\nu}\frac{\p \zeta}{\p u}\right).
\end{align}
Using equations (6.6a), (6.6b), (6.7a) and (6.7b) yields
\begin{align}
    \xi_+=-\frac{\xi}{r}(1+\mu-9\pi r^2)+\frac{2\eta}{r^2}-\frac{\zeta}{r^2}[1-\mu+1+4\pi r^2-4\pi r^2\zeta^2(1-4\pi r^2)].\label{eq:174}
\end{align}

Having appropriately constructed the approximating sequence, we can apply theorem 3.1 of \cite{II} to conclude that there is a positive integer $n_0$ such that for each $n\geq n_0$ the $n$'th regularized problem has a solution with a maximal existence interval $\bar{\tau}_n>0$.

\begin{Lemma}\label{lemma_analogue}
\begin{align} \liminf_{n\rightarrow\infty} \bar{\tau}_n=:\bar{\tau}>0.\end{align}
\end{Lemma}
\begin{proof}
As in the proof of lemma 1.1 of \cite{III} we assume that there exists a subsequence $(n_i:i=1,2,3,\ldots)$ $n_i\rightarrow \infty$ as $i\rightarrow \infty$ such that $\lim_{i\rightarrow\infty}\bar{\tau}_{n_i}=0$. According to lemma 4.3 of \cite{II}, for each $n\geq n_0$, $\gamma_n(\tau)\rightarrow \overline{\gamma}_n\in(0,1]$ as well as $\rho_{\ast n}(\tau)\rightarrow \overline{\rho}_n\in[0,1]$ as $\tau\rightarrow\overline{\tau}_n$, and we have the following four cases to consider:

\begin{center} Case 1: $\overline{\rho}_n<1,\overline{\gamma}_n<1$.

Case 2: $\overline{\rho}_n=1,\overline{\gamma}_n<1$.

Case 3: $\overline{\rho}_n<1,\overline{\gamma}_n=1$.

Case 4: $\overline{\rho}_n=1,\overline{\gamma}_n=1$.\end{center}

We shall show that under the hypothesis $\lim_{i\rightarrow \infty}\overline{\tau}_{n_i}=0$ each of the four cases leads to a contradiction.

The second case can be treated in the exact same way as in the proof of lemma 1.1 of \cite{III}. For the other cases we show that $\inf_{\mathscr{U}(\bar{\tau}_{n_i})}e_{n_i}>0$ for large enough $i$ (recall from p. 340 of \cite{II} that $\mathscr{U}(t):=\left\{(u,v):t\geq u\geq v\geq 0\right\}$). Once we have this, equations (1.34a) and (1.36g) of \cite{III} hold and the contradictions can be established in the same way as in the proof of lemma 1.1 of \cite{III}.
Looking at 
\begin{align} e_{n_i}(u,v)=\int_0^v\kappa_{n_i}(u,v')e_{+n_i}(u,v')dv'+e_{n_i}(u,0),\end{align}
where $e_+=\frac{1}{\kappa}\frac{\p e}{\p v}$, we see that the desired inequality for $e_{n_i}$ holds, once the following three inequalities hold
\begin{align} \inf_{\mathscr{U}(\bar{\tau}_{n_i})}\kappa_{n_i}>0,\qquad \inf_{\mathscr{U}(\bar{\tau}_{n_i})}e_{+n_i}>0,\qquad\inf_{u\in [0,\bar{\tau}_{n_i}]}e_{n_i}(u,0)>0\label{eq:175},\end{align}
for all sufficiently large $i$. Since $\kappa(u,v)=\kappa_\ast(v)\textrm{e}^{-K(u,v)}$ (cf. (6.12a) of \cite{I}), the first inequality follows from the second of (4.5a) of \cite{II}. For the third inequality we consider $e_{n_i}$ along the incoming null curve $C^{\ast -}$. Along $C^{\ast -}$ we have
\begin{align}e_{n_i}(u,0)=E_{n_i}(\phi(u,0))\qquad e_{-n_i}(u,0)=E_{-n_i}(\phi(u,0)),\qquad r_{n_i}(u,0)=R_{n_i}(\phi(u,0)),\qquad \phi_{n_i}=t.\end{align}
Since $E_{-n_i}=-\dot{E}_{n_i}/\dot{R}_{n_i}$, we have  $\dot{E}_{n_i}=-\dot{R}_{n_i}E_{-n_i}$. We also have
\begin{align} \lim_{n\rightarrow\infty}-\dot{R}_{n}(0)=a_{-0}>0,\qquad \lim_{n\rightarrow\infty}E_{-n}(0)=e_{-0}>0,\end{align}
For the latter inequality, $e_{-0}>0$, we use $e_{-0}=\left.\frac{1}{\nu}\frac{\p e}{\p u}\right|_{N^-}$, $\nu>0$ (by proposition 4.1 of \cite{I}) and $\left.\p e/\p u\right|_{N^-}>0$ (by proposition \ref{proposition2}). Hence,
\begin{align} -\dot{R}_n(0)\geq\frac{a_{-0}}{2}>0, \qquad E_{-n}(0)\geq\frac{e_{-0}}{2}>0,\end{align}
for all sufficiently large $n$. Since $E_{-n}$ and $\dot{R}_n$ are both equicontinuous (They converge uniformly as discussed in the previous section. The uniform convergence is equivalent to being equicontinuous and converging pointwise (corollary of Ascoli's theorem)), it follows that there exists an $\varepsilon>0$ independent of $n$ such that
\begin{align} -\dot{R}_n(t)\geq\frac{a_{-0}}{4},\qquad E_{-n}(t)\geq\frac{e_{-0}}{4}\qquad  :\forall t\in[0,\varepsilon].\end{align}
Thus,
\begin{align} \dot{E}_n(t)\geq \frac{a_{-0}e_{-0}}{16}>0\qquad:\forall t\in[0,\varepsilon],\end{align}
and since $E_n(0)=e_{0,n}>0$ (confer the third condition put on the initial data), it follows that
\begin{align}
  E_n(t)\geq\frac{a_{-0}e_{-0}}{16}t+e_{0,n}>0\qquad:\forall t\in[0,\varepsilon].\label{eq:176}
\end{align}
Now because of the fact that $\phi$ is a time function, we have that $\sup_{\mathscr{U}(\bar{\tau}_n)}\phi_n<\bar{\tau}_n$ (cf. (4.15b) of \cite{II}). By our assumption that $\lim_{i\rightarrow\infty}\bar{\tau}_{n_i}=0$ we have that $\bar{\tau}_{n_i}\leq \varepsilon$ for all sufficiently large $i$. Taking now the infimum on $[0,\varepsilon]$ of \eqref{eq:176} and taking into account that $\phi_{n_i}=t$ along $C^{\ast -}$, it follows, in view of the third of \eqref{eq:164}, that $\inf_{u\in[0,\bar{\tau}_{n_i}]}e_{n_i}(u,0)>0$. Therefore the third inequality of \eqref{eq:175} is shown.

To show the second inequality of \eqref{eq:175} we first show that the oscillations of $r_{n_i}$, $m_{n_i}$, $\zeta_{n_i}$, $\xi_{n_i}$ and  $\eta_{n_i}$ in $\mathscr{U}(\bar{\tau}_{n_i})$ tend to 0 as $i\rightarrow\infty$. Since $e_{+n_i}$ is given in terms of those quantities (cf. \eqref{eq:148} together with \eqref{eq:174}), it then follows that $\lim_{i\rightarrow \infty}e_{+n_i}= e_+=\textrm{const.}=e_{+0}=4a_{-0}i_0>0$ on $\mathscr{U}(\bar{\tau})$ (see proposition \ref{proposition2}). The right hand sides of the equations (1.17a) and (1.17e) of \cite{III} with $n=n_i$ approach $r_0$ as $i\rightarrow\infty$. This shows that 
\begin{align}
\lim_{i\rightarrow \infty}\osc_{\mathscr{U}(\bar{\tau}_{n_i})}r_{n_i}=0.
\end{align}
Also
\begin{align}\label{osci_m}
\lim_{i\rightarrow\infty}\osc_{u\in[0,\bar{\tau}_{n_i}]}m_{n_i}(u,0)=0.   
\end{align}
Using the bounds of page 103 of \cite{III} and equation (6.5b) of \cite{I}, it follows that
\begin{align}
\sup_{\mathscr{U}(\bar{\tau}_{n_i})}\left|\frac{\p m_{n_i}}{\p v} \right|\leq C.\end{align}
This inequality, together with (\ref{osci_m}),  then yields 
\begin{align}
  \lim_{i\rightarrow\infty}\osc_{\mathscr{U}(\bar{\tau}_{n_i})}m_{n_i}=0.\end{align}
Similarly
\begin{align}\label{osci_zeta}
\lim_{i\rightarrow\infty}\osc_{u\in[0,\bar{\tau}_{n_i}]}\zeta_{n_i}(u,0)=0
\end{align}
Using the bounds of page 103 of \cite{III} and equation (6.7b) of \cite{I}, it follows that
\begin{align}
\sup_{\mathscr{U}(\bar{\tau}_{n_i})}\left|\frac{\p \zeta_{n_i}}{\p v}\right|\leq C.
  \end{align}
This inequality, together with (\ref{osci_zeta}), then yields
\begin{align}
  \lim_{i\rightarrow \infty}\osc_{\mathscr{U}(\bar{\tau}_{n_i})}\zeta_{n_i}=0.
\end{align}
Also, since by \eqref{eq:173}
\begin{align}
  \label{eq:177}
  \xi_n(u,0)=-\frac{\ddot{R}(\phi(u,0))}{(\dot{R}(\phi(u,0)))^3},
\end{align}
we have, in view of the fact that the range of $\phi_{n_i}$ is contained in $[0,\overline{\tau}_{n_i}]$ and $\overline{\tau}_{n_i}\rightarrow 0$,
\begin{align}
\lim_{i\rightarrow\infty}\osc_{u\in[0,\bar{\tau}_{n_i}]}\xi_{n_i}(u,0)=0.\label{osci_xi}
\end{align}
Since $\frac{\p \xi_n}{\p v}=\kappa_n\xi_{+n}$, it follows by the expression \eqref{eq:174} for $\xi_+$ and the bound (1.29c) of \cite{III} for the difference
\begin{align}
  \label{eq:178}
  \xi_{n_i}(u,v)-\exp(-2N_{n_i}(u,v))\xi_{n_i}(u,0),
\end{align}
in conjunction with (\ref{osci_xi}), that
\begin{align}
  \lim_{i\rightarrow\infty}\osc_{\mathscr{U}(\bar{\tau}_{n_i})}\xi_{n_i}=0.
  \end{align}
Let us finally consider $\eta_{n_i}$. From equation (6.7a) of \cite{I} and the bounds on page 103 of \cite{III} we see that
\begin{align}
  \sup_{\mathscr{U}(\bar{\tau}_{n_i})}\left|\frac{\p \eta_{n_i}}{\p u}\right|\leq C.  
\end{align}
So we only need to show that
\begin{align}
\label{eq:179}
  \lim_{i\rightarrow \infty}\osc_{\tau\in [0,\bar{\tau}_{n_i}]}\eta_{\ast n_i}(\tau)=0.
\end{align}
Now, according to equation (4.5e) of \cite{II}, we can express
\begin{align}
  \label{eq:180}
  \eta_{\ast n}(\tau)=a_{-\ast n}(\tau)\sqrt{q_n(\tau)},
\end{align}
where
\begin{align}
  \label{eq:181}
  q_n(\tau)=\frac{1-2z_n(\tau)}{1-z_n(\tau)}
\end{align}
and
\begin{align}
  \label{eq:182}
  x_n(\tau):=2(1-\rho_{\ast n}(\tau)),\qquad y_n(\tau):=1-\gamma_n^2(\tau),\qquad z_n(\tau):=\frac{x_n(\tau)y_n(\tau)}{x_n(\tau)+2y_n(\tau)}.
  \end{align}
Since
\begin{align}
  0\leq z_n(\tau)  \leq \frac{1}{2}x_n(\tau)=1-\rho_{\ast n}(\tau)
\end{align}
and
\begin{align}
  \lim_{i\rightarrow\infty}\sup_{\tau\in[0,\bar{\tau}_{n_i}]}(1-\rho_{\ast n_i}(\tau))=0
\end{align}
it follows that
\begin{align}
  \lim_{i\rightarrow \infty}\osc_{\tau\in[0,\bar{\tau}_{n_i}]}q_{n_i}(\tau)=0,
  \end{align}
which, together with
\begin{align}
  \label{eq:183}
  \lim_{i\rightarrow \infty}\osc_{\tau\in[0,\bar{\tau}_{n_i}]}a_{-\ast n_i}(\tau)=0,
\end{align}
implies \eqref{eq:179}.
\end{proof}

In the arguments to follow we shall make use of the following proposition which is identical to proposition 1.1 of \cite{III}.
\begin{proposition}
  Let $\mathscr{B}$, $(r,m,\nu,\kappa,\zeta,\eta)$, defined on $(\tau_1,\tau_2)$, $\{(u,v):u\in (\tau_1,\tau_2), v\in (\tau_1,u]\}$ respectively, be a solution of the free-boundary problem corresponding to a soft phase solution $(r,\omega,\rho)$, such that $0<\beta<1$, $\rho_\ast<1$, whence also $0<\gamma<1$. Then at each $\tau\in(\tau_1,\tau_2)$,
  \begin{align}
    \label{eq:184}
    -\frac{1}{\gamma}\frac{d\gamma}{d\tau}=(1-\beta)E+\beta B,
  \end{align}
where $E$ and $B$ are the functions given in the statement of proposition 1.1 of \cite{III}.
\label{proposition11}
\end{proposition}

\textit{Remark.} In the limiting case $\gamma\rightarrow 1$ the function $E$ reduces to the boundary barrier $\bar{e}$ given by (1.23c) of \cite{III}. If also $\rho_\ast\rightarrow 1$, the function $B$ reduces to $\bar{b}$ (see lemma 4.4 of \cite{II}).

Now let $\bar{\tau}$ be as in the statement of lemma \ref{lemma_analogue}. Take any $\hat{\tau}\in(0,\bar{\tau})$. Then $\bar{\tau}_n\geq\hat{\tau}$ if $n$ is large enough. Let
\begin{align}
  \label{eq:185}
  \mathscr{V}^a:=\{(\tau,\chi)\in\mathscr{V}:\tau\leq a\},
\end{align}
where $\mathscr{V}$ is the wedge-shaped region as defined by \eqref{eq:156}. In the following we will make use of the functions $N_n$, $K_n$. Those are given by (1.41b) and (1.42a) of \cite{III} and correspond to the functions $N$ and $K$ of the $n$'th member of the sequence of regular problems, defined by (2.4b) and (2.3b) of \cite{II}.
\begin{Lemma}
  There is a subsequence $(n_i:i=1,2,\ldots)$ with the following properties:
  \begin{enumerate}
  \item The corresponding sequence of curves $(\mathscr{B}_{n_i})$ converges uniformly on $[0,\hat{\tau}]$ to a non-spacelike curve $\mathscr{B}$ contained in $\mathscr{V}^{\hat{\tau}}$.
\item The sequence of functions $(r_{n_i},\phi_{n_i},m_{n_i})$ converges uniformly in $\mathscr{U}(\hat{\tau})$ to $(r,\phi,m)$, Lipschitz functions on $\mathscr{U}(\hat{\tau})$.
\item The sequence of functions $(N_{n_i},K_{n_i})$ converges uniformly in $\mathscr{U}(\hat{\tau})$ to $(N,K)$, Lipschitz functions on $\mathscr{U}(\hat{\tau})$.
\item The sequence of functions $(\zeta_{n_i},\eta_{n_i})$ converges uniformly in $\mathscr{U}(\hat{\tau})$ to $(\zeta,\eta)$, continuous functions on $\mathscr{U}(\hat{\tau})$.
  \end{enumerate}
\label{lemma_sequence}
\end{Lemma}
\begin{proof}
The proof of the first three statements is identical to the contraction case. The proof of the uniform convergence of the sequence of functions $(\zeta_{n_i},\eta_{n_i})$ involves the functions $\zeta_n(\cdot,0)$ on $[0,\hat{\tau}]$. We have
\begin{align}
  \label{eq:186}
  \zeta_n(u,0)=Z_n(\phi_n(u,0))=-\frac{1}{\dot{R}_n(\phi_n(u,0))}.
\end{align}
Using \eqref{eq:165} yields
\begin{align}
  \label{eq:187}
  \zeta_n(u,0)=-\frac{1}{\dot{R}(\phi_n(u,0))+k_n+l_n\phi_n(u,0)}.
\end{align}
Since $\dot{R}$ is continuous and the functions $\phi_{n_i}(\cdot,0)$ converge uniformly on $[0,\hat{\tau})$ to $\phi(\cdot,0)$ while the constants $k_n,l_n\rightarrow 0$, it follows that the functions $\zeta_{n_i}(\cdot,0)$ converge uniformly on $[0,\hat{\tau})$ to $\zeta(\cdot,0)$, a continuous function on $[0,\hat{\tau}]$. All the other parts of the proof of the uniform convergence of $(\zeta_{n_i},\eta_{n_i})$ can be copied from the contraction case. An intermediate result of the proof (cf. (1.45f) of \cite{III}) is the uniform convergence of the sequence $(x_{n_i},y_{n_i})$, i.e.
\begin{align}
  \label{eq:188}
  (x_{n_i},y_{n_i})\rightarrow (x,y)\quad \textrm{as}\quad i\rightarrow \infty
\end{align}
uniformly in $[0,\hat{\tau}]$.
\end{proof}

To proceed further we must show that we can select a subsequence $(N_i)$ such that the corresponding functions $\beta_{n_i}$ converge uniformly on $[0,\hat{\tau}]$. This will be accomplished in two steps, lemma \ref{lemma_every_point} and lemma \ref{lemma_sufficiently}. According to equation (4.5c) of \cite{II}, $\beta_n$ is given by
\begin{align}
  \label{eq:189}
  \beta_n=\frac{1-\gamma_n^2}{1+\gamma_n^2-2\gamma_n^2\rho_{\ast n}}.
\end{align}
In terms of the variables $x_n$, $y_n$, (cf. \eqref{eq:182}) we can write this in the form
\begin{align}
  \label{eq:190}
  \beta_n=\frac{y_n}{(1-y_n)x_n+y_n}.
\end{align}
Now the function
\begin{align}
  \label{eq:191}
  g(x,y):=\frac{y}{(1-y)x+y}
\end{align}
is continuous on $\bar{Q}_\varepsilon\setminus (0,0)$ where
\begin{align}
  \label{eq:192}
  \bar{Q}_\varepsilon :=\{(x,y):x\geq 0,1-\varepsilon^2\geq y\geq 0\},
\end{align}
but is not continuous at the point $(0,0)$. Therefore \eqref{eq:188} does not imply the uniform convergence of the $\beta_n$.

Let us first make the following general remarks. The functions $\beta_n$ belong to $C^0[0,\hat{\tau}]\subset L^\infty [0,\hat{\tau}]$ and satisfy $0<\beta_n<1$. By Alaoglu's theorem, $\bar{B}_1(L^\infty[0,\hat{\tau}])$, the closed unit ball of $L^\infty[0,\hat{\tau}]=L^1[0,\hat{\tau}]^\star$, is weak-star compact. We can therefore select a subsequence from the subsequence $(n_i)$, which for convenience of notation we still denote by $(n_i)$, such that the corresponding functions $\beta_{n_i}$ converge weak-star to a function $\beta\in \bar{B}_1(L^\infty[0,\hat{\tau}])$, that is
\begin{align}
  \label{eq:193}
  \int_0^{\hat{\tau}}\beta_{n_i}fd\tau\rightarrow\int_0^{\hat{\tau}}\beta fd\tau\quad \forall f\in L^1[0,\hat{\tau}].
\end{align}
Since $\beta\in \bar{B}_1(L^\infty[0,\hat{\tau}])$, we have $|\beta|\leq 1$ almost everywhere. We shall show that also $\beta\geq 0$ almost everywhere. For, let
\begin{align}
  \label{eq:194}
  N:=\{\tau\in[0,\hat{\tau}]:\beta(\tau)<0\}.
\end{align}
Then
\begin{align}
  \label{eq:195}
  N=\bigcup_{m=1}^\infty N_m,\qquad \textrm{where}\qquad N_m:=\left\{\tau\in[0,\hat{\tau}]:\beta(\tau)\leq-\frac{1}{m}\right\}.
\end{align}
Let $\chi_{N_m}$ denote the characteristic function of $N_m$. Setting $f=\chi_{N_m}$ in \eqref{eq:193}, we obtain
\begin{align}
  \label{eq:196}
  0\leq \int_{N_m}\beta_{n_i}d\tau\rightarrow \int_{N_m}\beta d\tau\leq -\frac{1}{m} \textrm{meas}(N_m).
\end{align}
Hence, $\textrm{meas}(N_m)=0$ and $N$, being the countable union of sets of measure zero, is itself of measure zero. We conclude that $0\leq\beta\leq 1$ almost everywhere and we can adjust $\beta$ on a set of measure zero to achieve $0\leq\beta\leq 1$ everywhere. Now, suppose that $M$ is a measurable subset of $[0,\hat{\tau}]$ and suppose that
\begin{align}
  \label{eq:197}
  \beta_{n_i}\rightarrow \tilde{\beta}\quad\textrm{pointwise on $M$}.
\end{align}
Then $\beta|_M=\tilde{\beta}$ almost everywhere. For, on one hand, by the dominated convergence theorem, for any $f\in L^1[0,\hat{\tau}]$,
\begin{align}
  \label{eq:198}
  \int_M\beta_{n_i}fd\tau\rightarrow\int_M\tilde{\beta}fd\tau\quad \forall f\in L^1[0,\hat{\tau}],
\end{align}
while on the other hand, replacing $f$ in \eqref{eq:193} by $f\chi_M$, where $\chi_M$ is the characteristic function of $M$, we obtain
\begin{align}
  \label{eq:199}
  \int_M\beta_{n_i}fd\tau\rightarrow \int_M\beta fd\tau\quad \forall f\in L^1[0,\hat{\tau}].
\end{align}
Consequently,
\begin{align}
  \label{eq:200}
  \int_M(\beta-\tilde{\beta})fd\tau=0\quad\forall f\in L^1[0,\hat{\tau}].
\end{align}
In particular, setting $f=\chi_M(\beta-\tilde{\beta})$, we obtain the conclusion $\beta|_M=\tilde{\beta}$ almost everywhere. Thus we can adjust $\beta$ on a subset of measure zero to coincide with $\tilde{\beta}$ on $M$. In view of the fact that $0\leq\beta\leq 1$, this does not affect our previous adjustment.

As a consequence of the above, the functions
\begin{align}
  \label{eq:201}
  \nu_\ast:=\frac{1}{2}a_{-\ast}(1-\beta),\qquad\kappa_\ast:=\frac{1}{2a_{-\ast}}(1+\beta)
\end{align}
belong to $L^\infty[0,\hat{\tau}]$ and we have
\begin{align}
  \label{eq:202}
  \nu_{\ast n_i}\rightarrow \nu_\ast,\qquad \kappa_{\ast n_i}\rightarrow \kappa_\ast
\end{align}
in the weak-star sense. We set
\begin{align}
  \label{eq:203}
  \nu(u,v)=\nu_\ast(u)e^{N(u,v)-N(u,u)},\qquad \kappa(u,v)=\kappa_\ast(v)e^{-K(u,v)}
\end{align}
where $N$, $K$ are the uniform limits of the 
functions of $N_{n_i}$, $K_{n_i}$ as in the statement of lemma \ref{lemma_sequence}. We shall show that these functions are in fact given by the expressions (2.4b), (2.3b) of \cite{II}, i.e.
\begin{align}
  \label{eq:204}
    N(u,v)=\int_0^v\left((\mu-4\pi r^2)\frac{\kappa}{r}\right)(u,v')dv',\qquad K(u,v)=4\pi \int_v^u(r\nu\zeta^2)(u',v)du'.
\end{align}

The proof relies on the following lemma, which is identical to lemma 1.3 of \cite{III}
\begin{Lemma}\label{lemma_convergence}
  Let $(f_{n_i})$ be a sequence of functions in $C^0(\mathscr{U}(\hat{\tau}))$ converging uniformly to a function $f$ and let
  \begin{align}
    \label{eq:205}
    g_{n_i}(u,v):=\int_0^v(f_{n_i}\kappa_{n_i})(u,v')dv',\qquad g(u,v):=\int_0^v(f\kappa)(u,v')dv',
  \end{align}
  \begin{align}
    \label{eq:206}
    h_{n_i}(u,v):=\int_v^u(f_{n_i}\nu_{n_i})(u',v)du',\qquad h(u,v):=\int_v^u(f\nu)(u',v)du',
  \end{align}
Then
\begin{align}
  \label{eq:207}
  g_{n_i}\rightarrow g,\quad h_{n_i}\rightarrow h\quad \textrm{uniformly on $\mathscr{U}(\hat{\tau})$}.
\end{align}
\end{Lemma}

Applying lemma \ref{lemma_convergence} to the cases
\begin{align}
  \label{eq:208}
  f_{n_i}=\frac{1}{r_{n_i}}(\mu_{n_i}-4\pi r_{n_i}^2),\qquad f=\frac{1}{r}(\mu-4\pi r^2)
\end{align}
and
\begin{align}
  \label{eq:209}
  f_{n_i}=4\pi r_{n_i}\nu_{n_i}\zeta_{n_i}^2,\qquad f=4\pi r\nu\zeta^2
\end{align}
we obtain (using the conclusion for $g_{n_i}$, $g$)
\begin{align}
  \label{eq:210}
  N_{n_i}\rightarrow \int_0^v(\mu-4\pi r^2)\frac{\kappa}{r}dv\quad\textrm{uniformly on $\mathscr{U}(\hat{\tau})$}
\end{align}
and (using the conclusion for $h_{n_i}$, $h$)
\begin{align}
  \label{eq:211}
  K_{n_i}\rightarrow 4\pi \int_v^ur\nu\zeta^2du\quad\textrm{uniformly on $\mathscr{U}(\hat{\tau})$}.
\end{align}
Since $N_{n_i}\rightarrow N$ and $K_{n_i}\rightarrow K$ by lemma \ref{lemma_sequence}, the functions $N$ and $K$ are thus given by \eqref{eq:204}. It follows that $\nu$ and $\kappa$ satisfy equations (6.6b) and (6.6a) of \cite{I} respectively.

Applying lemma \ref{lemma_convergence} to the cases
\begin{align}
  \label{eq:212}
  (f_{n_i},f)&=(1-\mu_{n_i},1-\mu),(1,1),\notag\\
  &\hspace{4.7mm}(\eta_{n_i},\eta),(\zeta_{n_i},\zeta),\notag\\
  &\hspace{4.7mm}\left(2\pi r_{n_i}^2[\eta_{n_i}^2+(1-\mu_{n_i})],2\pi r^2[\eta^2+(1-\mu)]\right),\notag\\
  &\hspace{4.7mm}\left(-2\pi r_{n_i}^2[(1-\mu_{n_i})\zeta_{n_i}^2+1],-2\pi r^2[(1-\mu)\zeta^2+1]\right),\notag\\
  &\hspace{4.7mm}\left(\frac{1}{r_{n_i}}[\eta_{n_i}-(1-4\pi r_{n_i}^2)\zeta_{n_i}],\frac{1}{r}[\eta-(1-4\pi r^2)\zeta]\right),\notag\\
  &\hspace{4.7mm}\left(\frac{1}{r_{n_i}}[(1+4\pi r_{n_i}^2\zeta_{n_i}^2)\eta_{n_i}-(1-\mu_{n_i})\zeta_{n_i}],\frac{1}{r}[(1+4\pi r^2\zeta^2)\eta-(1-\mu)\zeta]\right),
\end{align}
we obtain the result that the hard phase system (cf. (6.3a)\ldots  (6.7b) of \cite{I}) is satisfied and therefore $(r,\phi,m,\nu,\kappa,\zeta,\eta)$ (this is the limit of the sequence of regular problems) is a solution of the hard-phase equations. The details are identical to the treatment in  \cite{III}.

Let us now define the closed subsets
\begin{align}
  \label{eq:213}
  X:=\{\tau\in[0,\hat{\tau}]:x(\tau)=0\},\qquad Y:=\{\tau\in[0,\hat{\tau}]:y(\tau)=0\}.
\end{align}
We then define the subsets
\begin{align}
  \label{eq:214}
  C_1:=X^c\cap Y^c.\qquad C_2:=X\cap Y^c,\qquad C_3:=X^c\cap Y,\qquad C_4:=X\cap Y.
\end{align}
Since
\begin{align}
  \label{eq:215}
  x=2(1-\rho_\ast),\qquad y=1-\gamma^2,
\end{align}
the subsets $C_1$, $C_2$, $C_3$, $C_4$ correspond to the four cases of lemma 4.3 of \cite{II}. The subset $C_1$ is open and, by virtue of \eqref{eq:188} and \eqref{eq:191}, $\beta_{n_i}\rightarrow \beta$ pointwise on $C_1$ and
\begin{align}
  \label{eq:216}
  0<\beta|_{C_1}<1.
\end{align}
The convergence is uniform in closed subsets of $C_1$. Thus $\beta$ is continuous on $C_1$. By \eqref{eq:188} and \eqref{eq:191}, $\beta_{n_i}\rightarrow \beta$ pointwise on $C_2$ and $C_3$, and
\begin{align}
  \label{eq:217}
  \beta|_{C_2}=1,\qquad \beta|_{C_3}=0.
\end{align}
Thus only on the closed subset
\begin{align}
  \label{eq:218}
  C_4=\{\tau\in[0,\hat{\tau}]:(x(\tau),y(\tau))=(0,0)\},
\end{align}
which corresponds to the point of non-continuity of the function $g$ on $\bar{Q}_\varepsilon$, does \eqref{eq:188} fail to give pointwise convergence. The argument instead rests on the following lemma.
\begin{Lemma}\label{lemma_every_point}
  For every $\tau_1\in(0,\hat{\tau}]$, there exists a $\tau_0\in(0,\tau_1)$ such that $\tau_0\in C_1$.
\end{Lemma}
\begin{proof}
Suppose on the contrary that there is a $\tau_1\in(0,\hat{\tau}]$ such that for all $\tau\in(0,\tau_1)$ we have $\tau\notin C_1$, that is,
\begin{align}
  \label{eq:219}
  \tau\in\bigcup_{i=2}^4C_i.
\end{align}
We may assume that $\tau_1$ is as small as we wish.

Consider first the case that there is a $\check{\tau}\in(0,\tau_1)$ such that $\check{\tau}\in C_2$, that is, $\check{\tau}\in X$, $\check{\tau}\in Y^c$. Since $Y^c$ is open, there is an $\varepsilon>0$ and an interval
\begin{align}
  \label{eq:220}
 I_\varepsilon:=[\check{\tau}-\varepsilon,\check{\tau}+\varepsilon]
\end{align}
such that $I_\varepsilon\subset (0,\tau_1)\cap Y^c$. Thus $I_\varepsilon \cap X^c=\emptyset$ because
\begin{align}
  \label{eq:221}
  I_\varepsilon \cap X^c\subset (0,\tau_1)\cap Y^c\cap X^c=(0,\tau_1)\cap C_1,
\end{align}
and according to our hypothesis the last set is empty. We conclude that $I_\varepsilon \subset X$; hence
\begin{align}
  \label{eq:222}
  I_\varepsilon \subset (0,\tau_1)\cap C_2
\end{align}
also. It follows that
\begin{align}
  \label{eq:223}
  \rho_\ast =1\quad \textrm{on $I_\varepsilon$},
\end{align}
and also that
\begin{align}
  \label{eq:224}
  \beta_{n_i}\rightarrow 1\quad\textrm{uniformly on $I_\varepsilon$},
\end{align}
which implies that for any $\tau',\tau''\in I_\varepsilon$,
\begin{align}
  \label{eq:225}
  \chi_{\ast n_i}(\tau'')-\chi_{\ast n_i}(\tau')=\int_{\tau'}^{\tau''}\exp (-\omega_{\ast n_i})\beta_{n_i}d\tau\rightarrow \int_{\tau'}^{\tau''}\exp (-\omega_\ast)\beta d\tau.
\end{align}
Since
\begin{align}
  \label{eq:226}
  \chi_{\ast n_i}(\tau'')-\chi_{\ast n_i}(\tau')\rightarrow \chi_\ast(\tau'')-\chi_\ast(\tau'),
\end{align}
$\chi_\ast$ is continuously differentiable in $I_\varepsilon$ and
\begin{align}
  \label{eq:227}
  e^{\omega_\ast}\frac{d\chi_\ast}{d\tau}=1.
\end{align}
This means that the segment of the curve $\mathscr{B}$ corresponding to the interval $I_\varepsilon$ is null outgoing. On the other hand, according to \eqref{eq:223}, this segment must coincide with the corresponding segment of the curve $\tau=\hat{\tau}(\chi)$, which is strictly timelike.

Consider next the case that there is a $\check{\tau}\in (0,\tau_1)$ such that $\check{\tau}\in C_3$, that is, $\check{\tau}\in Y$, $\check{\tau}\in X^c$. Since $X^c$ is open, there is an $\varepsilon_1>0$ and an interval $I_{\varepsilon_1}$ as in \eqref{eq:220}, but with $\varepsilon_1$ in the role of $\varepsilon$, such that $I_{\varepsilon_1}\subset (0,\tau_1)\cap X^c$. We have $I_{\varepsilon_1}\cap Y^c=\emptyset$, because
\begin{align}
  \label{eq:228}
  I_{\varepsilon_1} \cap Y^c\subset (0,\tau_1)\cap X^c\cap Y^c=(0,\tau_1)\cap C_1,
\end{align}
and according to our hypothesis the last set is empty. We conclude that $I_{\varepsilon_1}\subset Y$; hence
\begin{align}
  \label{eq:229}
  I_{\varepsilon_1}\subset(0,\tau_1)\cap C_3.
\end{align}
It follows that
\begin{align}
  \label{eq:230}
  \gamma=1 \quad\textrm{on $I_{\varepsilon_1}$}
\end{align}
and also that
\begin{align}
  \label{eq:231}
  \beta_{n_i}\rightarrow 0\quad\textrm{uniformly on $I_{\varepsilon_1}$}.
\end{align}
Let $0<\varepsilon<\varepsilon_1$. According to proposition \ref{proposition11} we can express
\begin{align}
  \label{eq:232}
  \gamma_{n_i}(\check{\tau}+\varepsilon)-\gamma_{n_i}(\check{\tau}-\varepsilon)=-\int_{\check{\tau}-\varepsilon}^{\check{\tau}+\varepsilon}\gamma_{n_i}[(1-\beta_{n_i})E_{n_i}+\beta_{n_i}B_{n_i}]d\tau.
\end{align}
Now, by \eqref{eq:230} the left hand side tends to zero as $i\rightarrow \infty$. Consequently, in conjunction with \eqref{eq:231},
\begin{align}
  \label{eq:233}
  \lim_{i\rightarrow \infty}\int_{\check{\tau}-\varepsilon}^{\check{\tau}+\varepsilon}e_{\ast n_i}d\tau=0.
\end{align}
Here $e_{\ast n_i}$ is the restriction to $\mathscr{B_{n_i}}$ of the barrier function \eqref{eq:129} corresponding to the $n_i$'th solution
\begin{align}
  \label{eq:234}
  e_{\ast n_i}=\frac{1}{2}a_{-\ast n_i}^2\xi_{\ast n_i}+\frac{1}{2r_{\ast n_i}}\left[a_{-\ast n_i}-\frac{1-4\pi r_{\ast n_i}^2}{a_{- \ast n_i}}\right]+\frac{\mu_{\ast n_i}}{2r_{\ast n_i}a_{- \ast n_i}}\quad\textrm{on $I_\varepsilon$}.
\end{align}
We will now show that
\begin{align}
  \label{eq:235}
  \liminf_{i\rightarrow \infty}\left(\inf_{I_\varepsilon}e_{\ast n_i}\right)>0,
\end{align}
contradicting \eqref{eq:233}. We have
\begin{align}
  \label{eq:236}
  e_{\ast n_i}(u)=e_{n_i}(u,0)+\int_0^ue_{+n_i}(u,v)\kappa_{n_i}(u,v)dv.
\end{align}
Along $C^{\ast -}$ we have $e_{n_i}=E_{n_i}(t)$ and $E_{n_i}\rightarrow E$ uniformly as $i\rightarrow \infty$. Therefore, it follows from proposition \ref{proposition2} that $E(t)>Ct$ for $\tau_1$ small enough, which as we have remarked, can in the present context be assumed. This implies that
\begin{align}
  \label{eq:237}
  \liminf_{i\rightarrow\infty}\left(\inf_{I_\varepsilon}e_{n_i}(u,0)\right)>0.
\end{align}
To deal with the second term in \eqref{eq:236} we first show that
\begin{align}
  \label{eq:238}
  \xi_{n_i}\rightarrow \xi:=\frac{1}{r}(\alpha_-+2\zeta)\quad\textrm{uniformly in $\mathscr{U}_{I_\varepsilon}^+$},
\end{align}
where (cf. (1.55a) of \cite{III})
\begin{align}
  \label{eq:239}
  \mathscr{U}^+_{I_\varepsilon}=\left\{(u,v):u\in I_\varepsilon,v\in[0,u]\right\}
\end{align}
and (cf. (1.26a) of \cite{III} and (6.23) of \cite{I})
\begin{align}
  \label{eq:240}
  \alpha_-:=\frac{1}{\nu}\frac{\p \alpha}{\p u},\qquad \alpha:=\frac{1}{\nu}\frac{\p (r\phi)}{\p u}.
\end{align}
This will then imply, through \eqref{eq:174}, that
\begin{align}
  \label{eq:241}
  \xi_{+n_i}\rightarrow \xi_+\quad\textrm{uniformly in $\mathscr{U}_{I_\varepsilon}^+$},
\end{align}
where $\xi_+$ is given by the right hand side of \eqref{eq:174}. This in turn will imply, through \eqref{eq:148}, that
\begin{align}
  \label{eq:242}
  e_{+n_i}\rightarrow e_+\quad\textrm{uniformly in $\mathscr{U}_{I_\varepsilon}^+$},
\end{align}
where $e_+$ is given by the right hand side of \eqref{eq:148}. Since proposition \ref{proposition2} gives $e_+>0$ in $\mathscr{U}_{I_\varepsilon}^+$, if $\tau_1$ is sufficiently small, which as we have remarked, can in the present context be assumed, it will then follow that the second term in \eqref{eq:236} is bounded from below by a positive constant for sufficiently large $i$. This together with \eqref{eq:237} will imply \eqref{eq:235}. To show the uniform convergence of $\xi_{n_i}$ we recall from (1.26c) of \cite{III} that we can write
\begin{align}
  \label{eq:243}
  \alpha_{-n_i}(u,v)=e^{-2N_{n_i}(u,v)}\alpha_{-n_i}(u,0)-A_{-n_i}(u,v),
\end{align}
where
\begin{align}
  \label{eq:244}
  A_{-n_i}&=\alpha(u,v)N_{-n_i}(u,v)e^{-N_{n_i}(u,v)}\notag\\
&\qquad+\int_0^v\left[\zeta_{n_i}N_{+n_i}+\phi_{n_i}(N_{+-n_i}+N_{+n_i}N_{-n_i})\right](u,v')e^{2N_{n_i}(u,v')-2N_{n_i}(u,v)}\kappa_{n_i}(u,v')dv'.
\end{align}
Now, as in the proof of lemma 1.5 of \cite{III} it follows that
\begin{align}
  \label{eq:245}
  A_{-n_i}\rightarrow A_-\quad\textrm{uniformly on $\mathscr{U}(\hat{\tau})$},
\end{align}
where
\begin{align}
  \label{eq:246}
  A_-(u,v)&:=\alpha(u,v)N_{-}(u,v)e^{-N_{}(u,v)}\notag\\
&\qquad+\int_0^v\left[\zeta N_{+}+\phi(N_{+-}+N_{+}N_{-})\right](u,v')e^{2N(u,v')-2N(u,v)}\kappa(u,v')dv'.
\end{align}
Using the initial conditions given by \eqref{eq:165} and $\dot{R}Z=-1$ it follows that
\begin{align}
  \label{eq:247}
  \xi_n(u,0)=-\frac{\dot{Z}_n}{\dot{R}_n}(\phi_n(u,0))=-\frac{\ddot{R}_n}{\dot{R}_n^3}(\phi_n(u,0))=-\left(\frac{\dot{Z}\dot{R}^2+l_n}{(\dot{R}+k_n+l_n\phi_n(u,0))^3}\right)(\phi_n(u,0)).
\end{align}
Proposition 4.1 of \cite{II} implies (see (1.39) of \cite{III}) that
\begin{align}
  \label{eq:248}
  0=\phi_{n_i}(0,0)\leq \phi_{n_i}(u,0)\leq \phi_{n_i}(u,u)\leq u.
\end{align}
Let $J_0\subset I_{\varepsilon_1} \setminus I_{\varepsilon}$ be a closed subinterval on the left of $I_{\varepsilon}$. From \eqref{eq:231} it follows that $\beta=0$ on $J_0$, which implies $\nu_\ast>0$ on $J_0$. It follows that on $J_0\times \{0\}$ the functions $\nu_{n_i}$ are uniformly bounded from below by a positive constant. Since the same is true for the functions $\zeta_{n_i}$ and we have $\p\phi_{n_i}/\p u=\nu_{n_i}\zeta_{n_i}$, it follows that
\begin{align}
  \label{eq:249}
  \inf_i\left(\inf_{u\in J_0}\frac{\p\phi_{n_i}}{\p u}\right):=\varepsilon_0>0.
\end{align}
Consequently
\begin{align}
  \label{eq:250}
  \inf_{u\in I_{\varepsilon}}\phi_{n_i}(u,0),\inf_{u\in I_{\varepsilon}}\phi(u,0)\geq \varepsilon_0|J_0|.
\end{align}
Since $\dot{Z}(t)$ is continuous away from $t=0$ we conclude that the sequence $(\xi_{n_i}(\cdot,0))$ converges uniformly in $I_{\varepsilon}$ to $\xi(\cdot,0)$, where
\begin{align}
  \label{eq:251}
  \xi(u,0)=-\frac{\dot{Z}(\phi(u,0))}{\dot{R}(\phi(u,0))}
\end{align}
is a continuous function in $I_{\varepsilon}$. Hence the sequence $\alpha_{-n_i}(\cdot,0)$ converges uniformly in $I_{\varepsilon}$ to $\alpha_-(\cdot,0)=r(\cdot,0)\xi(\cdot,0)-2\zeta(\cdot,0)$. In conjunction with \eqref{eq:245} this allows us to conclude that
\begin{align}
  \label{eq:252}
  \alpha_{-n_i}\rightarrow \alpha_-\quad\textrm{uniformly in $\mathscr{U}_{I_{\varepsilon}}^+$}
\end{align}
where
\begin{align}
  \label{eq:253}
  \alpha_-(u,v)=e^{-2N(u,v)}\alpha_-(u,0)-A_-(u,v)
\end{align}
is a continuous function in $\mathscr{U}_{I_{\varepsilon}}^+$. Consequently, \eqref{eq:238} follows.

The preceding development leads to the conclusion that $(0,\tau_1)\subset C_4$. Then for any closed interval $I_\varepsilon\subset(0,\tau_1)$ we have
\begin{align}
  \label{eq:254}
  \rho_\ast=\gamma=1\quad\textrm{on $I_\varepsilon$}
\end{align}
and \eqref{eq:232} holds. Therefore, considering the remark at the end of proposition \ref{proposition11}, since the left-hand side of \eqref{eq:232} tends to zero as before, we obtain
\begin{align}
  \label{eq:255}
  \lim_{i\rightarrow \infty}\int_{\check{\tau}-\varepsilon}^{\check{\tau}+\varepsilon}[(1-\beta_{n_i})e_{\ast n_i}+\beta_{n_i}b_{\ast n_i}]d\tau=0.
\end{align}
Now, the left-hand side of \eqref{eq:255} is not less than
\begin{align}
  \label{eq:256}
  2\varepsilon \liminf_{i\rightarrow \infty}\left(\inf_{I_\varepsilon}\min\{e_{\ast n_i},b_{\ast n_i}\}\right),
\end{align}
and, since \eqref{eq:235} holds, while by proposition 5.1 of \cite{I}, in view of \eqref{eq:254}, $b_{\ast n_i}\rightarrow b_\ast$ uniformly in $I_\varepsilon$ and
\begin{align}
  \label{eq:257}
  \inf_{I_\varepsilon}b_\ast>0.
\end{align}
We have therefore again reached a contradiction and the proof of the lemma is complete.

\end{proof}
Next we show
\begin{Lemma}\label{lemma_sufficiently}
  If $\hat{\tau}$ is sufficiently small, then $C_1$ coincides with $(0,\hat{\tau}]$.
\end{Lemma}
\begin{proof}
The proof of this lemma is identical to the proof of lemma 1.5 of \cite{III} up to two points.

The first point that has to be changed is the proof of the fact that the sequence of functions $\xi_{n_i}(\cdot,0)$ is uniformly convergent in any closed subinterval $J\subset I$ where $I$ is a component of the open set $C_1$ as in the proof in \cite{III}. Using the initial conditions given by \eqref{eq:165} and $\dot{R}Z=-1$ it follows that
\begin{align}
  \label{eq:258}
  \xi_n(u,0)=-\frac{\dot{Z}_n}{\dot{R}_n}(\phi_n(u,0))=-\frac{\ddot{R}_n}{\dot{R}_n^3}(\phi_n(u,0))=-\left(\frac{\dot{Z}\dot{R}^2+l_n}{(\dot{R}+k_n+l_n\phi_n(u,0))^3}\right)(\phi_n(u,0)).
\end{align}
Since $\dot{Z}$, $\dot{R}$ are continuous and applying the same reasoning as in the proof of lemma 1.5 of \cite{III} that yields
\begin{align}
  \label{eq:259}
  \inf_{u\in J}\phi_{n_i}(u,0),\inf_{u\in J}\phi(u,0)\geq \varepsilon_0|J_0|,
\end{align}
it follows that the sequence $(\xi_{n_i}(\cdot,0))$ converges uniformly in $J$ to $\xi(\cdot,0)$, where
\begin{align}
  \label{eq:260}
  \xi(u,0)=-\frac{\dot{Z}(\phi(u,0))}{\dot{R}(\phi(u,0))}
\end{align}
is a continuous function in $I$ extending continuously to the right end point $\tau_1$ of $I$.

The second point that has to be changed is the proof of $e_1>0$ for small enough $\hat{\tau}$, where $e_1$ is the boundary barrier function at the point $\tau_1$. This then contradicts (1.58a) through (1.58c) of \cite{III} using also proposition 5.1 of \cite{I}. The proof of $e_1>0$ for $\hat{\tau}$ small enough can be given using the argument in the proof of lemma \ref{lemma_every_point} (the second case), but with the interval $I$ (whose right end point coincides with $\tau_1$) in the role of the interval $I_\varepsilon$.

\end{proof}

Up to this point we have
\begin{align}
  \label{eq:261}
  \beta_{n_i}(\tau)\rightarrow \beta(\tau)\quad\textrm{uniformly on $(0,\hat{\tau}]$}.
\end{align}
Therefore, $\beta$ is continuous and bounded on $(0,\hat{\tau}]$ and $\chi_\ast\in C^1(0,\hat{\tau}]$. Furthermore
\begin{align}
  \label{eq:262}
  0\leq \beta(\tau)\leq 1\qquad\textrm{for $\tau\in (0,\hat{\tau}]$}.
\end{align}
This implies
\begin{align}
  \label{eq:263}
  \chi_\ast(\tau)=\mathcal{O}(\tau).
\end{align}
By construction we also have
\begin{align}
  \label{eq:264}
  \lim_{i\rightarrow \infty}\beta_{n_i}(0)=\beta_0,
\end{align}
where $\beta_0$ is given by \eqref{eq:159}. In the following we will show
\begin{align}
  \label{eq:265}
  \lim_{\tau\rightarrow 0}\beta(\tau)=\beta_0.
\end{align}
Hence
\begin{align}
  \label{eq:266}
  \beta\in C^0[0,\hat{\tau}] 
\end{align}
and
\begin{align}
  \label{eq:267}
  \chi_\ast\in C^1[0,\hat{\tau}].
\end{align}

Defining the variables $z$, $E$ and $X$ by
\begin{align}
  \label{eq:268}
  z^2:=\frac{1-\beta}{1+\beta},\qquad E:=\frac{1}{\gamma^2}-1,\qquad X:=1-\rho_\ast,
\end{align}
the equation for the free phase boundary, equation (1.24a) of \cite{II}, becomes
\begin{align}
  \label{eq:269}
  z^2=\frac{X}{E+X}.
\end{align}
Using (cf.~(6.23) of \cite{I})
\begin{align}
  \label{eq:270}
  \alpha=-\frac{\partial (r\phi)/\partial u}{\partial r/\partial u}=r\zeta-\phi
\end{align}
and $\gamma=1/(\zeta_\ast a_{-\ast})$ (see (1.23b) of \cite{II}), we can express
\begin{align}
  \label{eq:271}
 E=a_{-\ast}^2\left(\frac{\alpha_\ast+\phi_\ast}{r_\ast}\right)^2-1.
\end{align}
The function $\alpha$ satisfies (see (4.14b) of \cite{II})
\begin{align}
  \label{eq:272}
  \alpha(u,v)e^{N(u,v)}=\alpha(u,0)-\int_0^v\left(\phi\frac{\p N}{\p v}e^N\right)(u,v')dv',
\end{align}
where $N$ is given by the first of \eqref{eq:204}. Let
\begin{align}
  \label{eq:273}
  M(u,v):=\int_0^vg(u,v')dv'.
\end{align}
where
\begin{align}
  \label{eq:274}
  g:=\phi f e^N,\qquad f:=\frac{\partial N}{\partial v}=(\mu-4\pi r^2)\frac{\kappa}{r}.
\end{align}
Thus
\begin{align}
  \label{eq:275}
  \alpha(u,v)=\left(\alpha(u,0)-M(u,v)\right)e^{-N(u,v)}.
\end{align}
In particular
\begin{align}
  \label{eq:276}
  \alpha_\ast(u)=\left(\alpha(u,0)-M_\ast(u)\right)e^{-N_\ast(u)}
\end{align}
and we can express
\begin{align}
  \label{eq:277}
  \alpha(u,0)=\alpha_\ast(u)e^{N_\ast(u)}+M_\ast(u).
\end{align}
Substituting this into \eqref{eq:275}, we obtain
\begin{align}
  \label{eq:278}
  \alpha(u,v)=\left(\alpha_\ast(u)e^{N_\ast(u)}+M_\ast(u)-M(u,v)\right)e^{-N(u,v)}.
\end{align}

In the following we will express $\alpha$ along $C^{\ast -}$ as a function of $s=r_0-r$. We note that $s$ is increasing along $C^{\ast -}$ and vanishes at $N^-$. We recall that $R(t)$ describes $r$ as a given function of $\phi$ along $C^{\ast -}$. Therefore, $\phi$ along $C^{\ast -}$ as a function of $s$  is given by the function
\begin{align}
  \label{eq:279}
  \Phi(s):=R^{-1}(r_0-s).
\end{align}
Now, in view of equation (6.4a) of \cite{I}, i.e.~$\partial \phi/\partial u=\nu\zeta$, we have for $\zeta$ along $C^{\ast -}$ as a function of $s$
\begin{align}
  \label{eq:280}
  \zeta|_{C^{\ast -}}(s)=\frac{d\Phi}{ds}(s).
\end{align}
Let us denote by $H$ the function which defines $-\alpha$ as a function of $s$ along $C^{\ast-}$. Hence, using \eqref{eq:279} and \eqref{eq:280} in \eqref{eq:270},
\begin{align}
  \label{eq:281}
  H(s)=(s-r_0)\frac{d\Phi}{ds}(s)+\Phi(s).
\end{align}

Let us denote by $Y$ the function which describes $r$ along $C^{\ast -}$ according to
\begin{align}
  \label{eq:282}
  Y(u):=r_0-r(u,0).
\end{align}
Therefore,
\begin{align}
  \label{eq:283}
  \alpha(u,0)=-(H\circ Y)(u).
\end{align}
Substituting this into \eqref{eq:276} we obtain
\begin{align}
  \label{eq:284}
  \alpha_\ast(\tau)=-\left[(H\circ Y(\tau))+M_\ast(\tau)\right]e^{-N_\ast(\tau)}.
\end{align}
Substituting this into \eqref{eq:271} we obtain
\begin{align}
  \label{eq:285}
  E=a_{-\ast}^2\left(\frac{\left[(H\circ Y)+M_\ast\right]e^{-N_\ast}-\phi_\ast}{r_\ast}\right)^2-1.
\end{align}
For the asymptotic form of $E$ we need the asymptotic forms of $a_{-\ast}$, $(H\circ Y)$, $M_\ast$, $N_\ast$, $\phi_\ast$ and $r_\ast$.

We first expand $H$ up to second order. Using the fact that $\zeta_0=1/a_{-0}$ (see \eqref{eq:133}), we get
\begin{align}
  \label{eq:286}
  H(0)=-\frac{r_0}{a_{-0}}.
\end{align}
Now,
\begin{align}
  \label{eq:287}
  \frac{dH}{ds}(0)=2\frac{d\Phi}{ds}(0)-r_0\frac{d^2\Phi}{ds^2}(0).
\end{align}
Using $\partial \zeta/\partial u=\nu\xi$ we get from \eqref{eq:280}
\begin{align}
  \label{eq:288}
  \frac{d^2\Phi}{ds^2}=\xi.
\end{align}
Therefore,
\begin{align}
  \label{eq:289}
  \frac{dH}{ds}(0)=\frac{2}{a_{-0}}-r_0\xi_0.
\end{align}
Using \eqref{eq:133}, \eqref{eq:135}, we obtain
\begin{align}
  \label{eq:290}
  \frac{dH}{ds}(0)=\frac{3}{a_{-0}}-\frac{a_{+0}}{a_{-0}^2}+\frac{4\pi r_0^2}{a_{-0}^3}.
\end{align}
Now,
\begin{align}
  \label{eq:291}
  \frac{d^2H}{ds^2}(0)=3\frac{d^2\Phi}{ds^2}(0)-r_0\frac{d^3\Phi}{ds^3}(0).
\end{align}
Using $\partial\xi/\partial u=\nu\xi_-$ we get from \eqref{eq:288}
\begin{align}
  \label{eq:292}
  \frac{d^3\Phi}{ds^3}=\xi_{-}.
\end{align}
Therefore
\begin{align}
  \label{eq:293}
  \frac{d^2H}{ds^2}(0)=3\xi_0-r_0\xi_{-0}.
\end{align}
Using \eqref{eq:135} and \eqref{eq:143} we obtain
\begin{align}
  \label{eq:294}
  \frac{d^2H}{ds^2}(0)=&-\frac{4ir_0}{a_{-0}^3}+\frac{1}{r_0a_{-0}^3}\left(1-36\pi r_0^2-3a_{+0}^2\right)\notag\\
&\qquad+\frac{1}{r_0 a_{-0}^5}\left(-48\pi^2r_0^4-3a_{-0}^4+20\pi r_0^2a_{-0}a_{+0}+5a_{-0}^3a_{+0}\right).
\end{align}

From (4.6c) of \cite{II} we have
\begin{align}
  \label{eq:295}
  \frac{d\phi_\ast}{d\tau}=\frac{1}{2}(f_++f_-),
\end{align}
where (recall that $x=2(1-\rho_\ast)$)
\begin{align}
  \label{eq:296}
  f_+:=(1+\beta)\sqrt{1-\frac{\beta x}{1+\beta}},\qquad f_-:=(1-\beta)\sqrt{1+\frac{\beta x}{1-\beta}}.
\end{align}
Let
\begin{align}
  \label{eq:297}
  g(s):=\sqrt{1+s}-\left(1+\tfrac{1}{2}s-\tfrac{1}{8}s^2\right).
\end{align}
The Taylor expansion of $g$ begins with cubic terms and we can express
\begin{align}
  \label{eq:298}
  \frac{1}{2}(f_++f_-)=1-\frac{1}{8}\frac{\beta^2x^2}{1-\beta^2}+f_0,
\end{align}
where
\begin{align}
  \label{eq:299}
  f_0:=\frac{1}{2}(1+\beta)g\left(-\frac{\beta x}{1+\beta}\right)+\frac{1}{2}(1-\beta)g\left(\frac{\beta x}{1-\beta}\right).
\end{align}
Using \eqref{eq:127}, we deduce
\begin{align}
  \label{eq:300}
  \frac{d\phi_\ast}{d\tau}=1+\mathcal{O}(\tau^4).
\end{align}
We define
\begin{align}
  \label{eq:301}
  \Psi(\tau):=\phi_\ast(\tau)-\tau
\end{align}
and we have
\begin{align}
  \label{eq:302}
  \Psi=O(\tau^5).
\end{align}

Let us define
\begin{align}
  \label{eq:303}
  \delta:=\tau-\frac{\chi_\ast}{2}-\frac{2Y}{a_{-0}}.
\end{align}

In the following we will use the notation $f^{(n)}$ to denote the $n$'th degree part of $f$.

The following proposition gives expressions for $M_\ast$, $N_\ast$ and $\delta$ up to cubic error terms.
\begin{proposition} The functions
  \begin{align}
    \label{eq:304}
    M_\ast(\tau)=\int_0^\tau\left(\phi \frac{\mu-4\pi r^2}{r}\kappa e^N\right)(\tau,v)dv,\qquad N_\ast(\tau)=\int_0^\tau\left(\frac{\mu-4\pi r^2}{r}\kappa\right)(\tau,v)dv,
  \end{align}
  \begin{align}
        \delta(\tau)=\tau-\frac{\chi_\ast}{2}-\frac{2Y(\tau)}{a_{-0}},
  \end{align}
where
\begin{align}
  \label{eq:305}
  Y(\tau)=r_0-r(\tau,0),
\end{align}
have the following asymptotic forms:
  \begin{align}
    \label{eq:306}
     M_\ast=M_\ast^{(2)}+\mathcal{O}(\tau^3),\qquad N_\ast=N_\ast^{(1)}+N_\ast^{(2)}+\mathcal{O}(\tau^3),\qquad\delta=\delta^{(2)}+\mathcal{O}(\tau^3),
  \end{align}
where
  \begin{align}
 \label{eq:307}
    M_\ast^{(2)}&:=\frac{\mu_0-4\pi r_0^2}{8a_{-0}r_0}\left(3\tau^2+\tau\chi_\ast-\frac{1}{4}\chi_\ast^2\right),\\
         \label{eq:308}
N_\ast^{(1)}&:=\frac{\mu_0-4\pi r_0^2}{2a_{-0}r_0}\left(\tau+\frac{\chi_\ast}{2}\right),\\
N_\ast^{(2)}&:=\frac{1}{2}A\tau^2+B\tau\chi_\ast+\frac{1}{2}C\chi_\ast^2,\\
    \label{eq:309}
    \delta^{(2)}&:=-\frac{1}{2r_0a_{-0}}\bigg\{2\pi r_0^2\tau^2+\left[a_{-0}(a_{+0}-a_{-0})-2\pi r_0^2\right]\tau\chi_\ast\notag\\
&\qquad\qquad\qquad\qquad+\frac{1}{4}\left[a_{-0}(a_{+0}-a_{-0})+2\pi r_0^2-2a_{-0}r_0k\right]\chi_\ast^2\bigg\},
\end{align}
where
\begin{align}
  \label{eq:310}
  A&:=\frac{1}{2a_{-0}^2r_0^2}\left[-\mu_0+\frac{\mu_0^2}{2}-4\pi r_0^2+4\pi r_0^2\mu_0+8\pi^2r_0^4+(2\mu_0+2\pi r_0^2)a_{-0}^2\right],\\
B&:=\frac{1}{8a_{-0}^2r_0^2}\left[-\mu_0\left(2-\frac{3}{2}\mu_0\right)+4\pi r_0^2-8\pi^2r_0^4-2\pi r_0^2a_{-0}^2+\mu_0a_{-0}^2\right],\\
C&:=\frac{1}{8a_{-0}^2r_0^2}\bigg[-4\mu_0+\frac{7}{2}\mu_0^2+16\pi r_0^2-8\pi r_0^2\mu_0-24\pi^2r_0^4\notag\\
&\qquad\qquad\qquad\qquad+2r_0(\mu_0-4\pi r_0^2)ka_{-0}+(\mu_0-10\pi r_0^2)a_{-0}^2\bigg].
\end{align}
\label{proposition10}
\end{proposition}
\begin{proof}
From \eqref{eq:273},
\begin{align}
  \label{eq:311}
  M_\ast(\tau)&=\int_0^\tau g(\tau,v)dv,\\
  \label{eq:312}
  \frac{d M_\ast}{d\tau}(\tau)&=g(\tau,\tau)+\int_0^\tau\frac{\p g}{\p u}(\tau,v)dv.
\end{align}
We recall the second of (4.5a) from \cite{II}:
\begin{align}
  \label{eq:313}
  \kappa_\ast=\frac{1}{2a_{-\ast}}(1+\beta).
\end{align}
Using \eqref{eq:301}, \eqref{eq:302}, we obtain
\begin{align}
  \label{eq:314}
   g(\tau,\tau)=\phi_\ast(\tau)f(\tau,\tau)e^{N_\ast(\tau)}&=\tau\frac{\mu_\ast-4\pi r_\ast^2}{2a_{-\ast} r_\ast}(1+\beta)+\mathcal{O}(\tau^2)\notag\\
&=\tau\frac{\mu_0-4\pi r_0^2}{2a_{-0}r_0}(1+\beta)+\mathcal{O}(\tau^2).
\end{align}
For the second term in \eqref{eq:312} we write
\begin{align}
  \label{eq:315}
  \int_0^\tau\frac{\p g}{\p u}(\tau,v)dv=I_1+I_2,
\end{align}
where
\begin{align}
  \label{eq:316}
  I_1:=\int_0^\tau\left(\frac{\p \phi}{\p u}fe^N\right)(\tau,v)dv,\quad I_2:=\int_0^\tau \left[\phi \,e^N\left(\frac{\p f}{\p u}+f\frac{\partial N}{\partial u}\right)\right](\tau,v)dv.
\end{align}

We rewrite $I_1$ as
\begin{align}
  \label{eq:317}
  I_1=\int_0^\tau (F\kappa)(\tau,v)dv,
\end{align}
where
\begin{align}
  \label{eq:318}
  F:=\frac{\partial \phi}{\partial u}\frac{\mu-4\pi r^2}{r}e^N.
\end{align}
Now,
\begin{align}
  \label{eq:319}
  \left(\frac{\mu-4\pi r^2}{r}\right)(\tau,v)=\frac{\mu_0-4\pi r_0^2}{r_0}+\mathcal{O}(\tau).
\end{align}
Let us recall the first of (4.5b) of \cite{II}:
\begin{align}
  \label{eq:320}
  \zeta_\ast=\frac{1}{a_{-\ast}}\sqrt{\frac{1+\beta-2\rho_\ast\beta}{1-\beta}},
\end{align}
as well as the first of (4.5a) of \cite{II}:
\begin{align}
  \label{eq:321}
  \nu_\ast=\frac{a_{-\ast}}{2}(1-\beta).
\end{align}
Using now $\partial \phi/\partial u=\nu\zeta$ (see the first of \eqref{eq:64}) and recalling \eqref{eq:127}, it follows
\begin{align}
  \label{eq:322}
  \frac{\p \phi}{\p u}(\tau,\tau)=\frac{1}{2}(1-\beta(\tau))+\mathcal{O}(\tau^2).
\end{align}
Since the mixed derivative of $\phi$ can be expressed through the basic equations of section 6 of \cite{I} in terms of $r$, $m$, $\nu$, $\kappa$, $\zeta$ and $\eta$, it is bounded in $\mathscr{U}(\hat{\tau})$ and we get
\begin{align}
  \label{eq:323}
  \frac{\p \phi}{\p u}(\tau,v)=\frac{1}{2}(1-\beta(\tau))+\mathcal{O}(\tau).
\end{align}
Therefore we can write
\begin{align}
  \label{eq:324}
  I_1=(1-\beta)\frac{\mu_0-4\pi r_0^2}{2r_0}\int_0^\tau\kappa(\tau,v)dv+\mathcal{O}(\tau^2).
\end{align}
Now, integrating (6.6a) of \cite{I} with respect to $u$ from $u=v$ to $u=\tau$ we get
\begin{align}
  \label{eq:325}
  \kappa(\tau,v)&=\kappa_\ast(v)e^{-K(\tau,v)}\notag\\
&=\kappa_\ast(v)+\mathcal{O}(\tau).
\end{align}
By \eqref{eq:313},
\begin{align}
  \label{eq:326}
  \kappa_\ast(v)=\frac{1}{2a_{-0}}(1+\beta(v))+\mathcal{O}(v).
\end{align}
Hence
\begin{align}
  \label{eq:327}
  \int_0^\tau\kappa(\tau,v)dv=\frac{1}{2a_{-0}}\int_0^\tau(1+\beta(v))dv+\mathcal{O}(\tau^2).
\end{align}
Now (recall that $e^{\omega_0}=\frac{1}{2}$),
\begin{align}
  \label{eq:328}
  \beta(v)=e^{\omega_{\ast}(v)}\frac{d\chi_\ast}{dv}(v)=\frac{1}{2}\frac{d\chi_\ast}{dv}(v)+\mathcal{O}(v).
\end{align}
Therefore,
\begin{align}
  \label{eq:329}
  \int_0^\tau\kappa(\tau,v)dv=\frac{1}{2a_{-0}}\left(\tau+\frac{1}{2}\chi_\ast\right)+\mathcal{O}(\tau^2).
\end{align}

Using $\partial\phi/\partial v=\kappa\eta$ (see the second of \eqref{eq:64}), we have
\begin{align}
  \label{eq:330}
  \phi(\tau,v')=\phi(\tau,\tau)+\int_\tau^{v'}(\kappa\eta)(\tau,v'')dv''=\phi(\tau,\tau)+\mathcal{O}(\tau).
\end{align}
Using \eqref{eq:301}, \eqref{eq:302} it follows
\begin{align}
  \label{eq:331}
  I_2=\mathcal{O}(\tau^2).
\end{align}
Using now \eqref{eq:329} in \eqref{eq:324} we obtain for the second term of \eqref{eq:312}
\begin{align}
  \label{eq:332}
  \int_0^\tau \frac{\p g}{\p u}(\tau,v)dv=\frac{\mu_0-4\pi r_0^2}{4a_{-0}r_0}(1-\beta)\left(\tau+\frac{1}{2}\chi_\ast\right)+\mathcal{O}(\tau^2).
\end{align}
Combining \eqref{eq:314} and \eqref{eq:332} (and using \eqref{eq:328}) yields
\begin{align}
  \label{eq:333}
  \frac{dM_\ast}{d\tau}&=\frac{\mu_0-4\pi r_0^2}{a_{-0}r_0}\left[\frac{\tau}{2}\left(1+\frac{1}{2}\frac{d\chi_\ast}{d\tau}\right)+\frac{1}{4}\left(1-\frac{1}{2}\frac{d\chi_\ast}{d\tau}\right)\left(\tau+\frac{\chi_\ast}{2}\right)\right]+\mathcal{O}(\tau^2)\notag\\
&=\frac{\mu_0-4\pi r_0^2}{a_{-0}r_0}\frac{d}{d\tau}\left(\frac{3}{8}\tau^2+\frac{1}{8}\tau\chi_\ast-\frac{1}{32}\chi_\ast^2\right)+\mathcal{O}(\tau^2).
\end{align}
The result for $M_\ast$ follows.

From \eqref{eq:204},
\begin{align}
  \label{eq:334}
  N_\ast(\tau)&=\int_0^\tau f(\tau,v)dv,\\
  \label{eq:335}
  \frac{dN_\ast}{d\tau}(\tau)&=f(\tau,\tau)+\int_0^\tau \frac{\p f}{\p u}(\tau,v)dv,
\end{align}
where $f$ is given by the second of \eqref{eq:274}. By \eqref{eq:313}, the first term is
\begin{align}
  \label{eq:336}
  f(\tau,\tau)=\frac{\mu_\ast-4\pi r_\ast^2}{2a_{-\ast}r_\ast}(1+\beta).
\end{align}
From the basic equations of section 6 of \cite{I} we obtain
\begin{align}
  \label{eq:337}
  \frac{\partial f}{\partial u}=\frac{\p }{\p u}\left(\frac{\mu-4\pi r^2}{r}\kappa\right)=G\kappa,
\end{align}
where
\begin{align}
  \label{eq:338}
  G:=-\frac{\nu}{r^2}\left[4\pi\zeta^2r^2 (1-4\pi r^2)-2\mu\right].
\end{align}
From the same basic equations in section 6 of \cite{I}, $\partial G/\partial v$ is bounded in $\mathscr{U}(\hat{\tau})$, hence
\begin{align}
  \label{eq:339}
  G(\tau,v)=G(\tau,\tau)+\mathcal{O}(v-\tau).
\end{align}
Using \eqref{eq:329} we get
\begin{align}
  \label{eq:340}
  \int_0^\tau \frac{\p f}{\p u}(\tau,v')dv'=\frac{G_\ast}{2a_{-0}}\left(\tau+\frac{1}{2}\chi_\ast\right)+\mathcal{O}(\tau^2).
\end{align}
Hence, from \eqref{eq:335}, \eqref{eq:336},
\begin{align}
  \label{eq:341}
  \frac{dN_\ast}{d\tau}=\frac{\mu_\ast-4\pi r_\ast^2}{2a_{-\ast}r_\ast}(1+\beta)+\frac{G_\ast}{2a_{-0}}\left(\tau+\frac{1}{2}\chi_\ast\right)+\mathcal{O}(\tau^2).
\end{align}
From \eqref{eq:320}, \eqref{eq:321} we obtain
\begin{align}
  \label{eq:342}
  G_\ast=-\frac{2\pi(1-4\pi r_\ast^2)}{a_{-\ast}}(1+\beta-2\rho_\ast\beta)+\frac{\mu_\ast a_{-\ast}}{r_\ast^2}(1-\beta).
\end{align}
We introduce the two functions
\begin{align}
  \label{eq:343}
  \psi:=\frac{\mu-4\pi r^2}{ra_-},\qquad \tilde{\psi}:=-\frac{2\pi(1-4\pi r^2)}{a_-}+\frac{\mu a_-}{r^2}.
\end{align}
Taking into account \eqref{eq:127}, \eqref{eq:342} becomes
\begin{align}
  \label{eq:344}
  G_\ast=\tilde{\psi}_\ast (1-\beta)+\mathcal{O}(\tau^2).
\end{align}
Substituting into \eqref{eq:341} we obtain
\begin{align}
  \label{eq:345}
  \frac{dN_\ast}{d\tau}=\frac{\psi_\ast}{2}(1+\beta)+\frac{\tilde{\psi}_\ast}{2a_{-0}}(1-\beta)\left(\tau+\frac{1}{2}\chi_\ast\right)+\mathcal{O}(\tau^2)
\end{align}
To obtain an expression for $N_\ast$ up to cubic error terms we need $dN_\ast/d\tau$ up to quadratic error terms. Using \eqref{eq:328} and expanding $\psi_\ast$ and $e^{\omega_\ast}$ to the appropriate order we obtain
\begin{align}
  \label{eq:346}
  \frac{dN_\ast}{d\tau}&=\frac{1}{2}\left[\psi_0+\left(\frac{\partial\psi}{\partial\tau}\right)_0\tau+\left(\frac{\partial\psi}{\partial\chi}\right)_0\chi_\ast\right]\left[1+\left(e^{\omega_0}+\left(\frac{\partial e^\omega}{\partial \tau}\right)_0\tau+\left(\frac{\partial e^\omega}{\partial \chi}\right)_0\chi_\ast\right)\frac{d\chi_\ast}{d\tau}\right]\notag\\
&\qquad+\frac{\tilde{\psi}_0}{2a_{-0}}\left(1-\frac{1}{2}\frac{d\chi_\ast}{d\tau}\right)\left(\tau+\frac{1}{2}\chi_\ast\right)+\mathcal{O}(\tau^2).
\end{align}
We now determine the first order partial derivatives of $\psi$ and $\omega$. From the second of \eqref{eq:29} together with the first of \eqref{eq:61} we have (recall that $\mu=2m/r$)
\begin{align}
  \label{eq:347}
  e^{-\omega}\frac{\partial r}{\partial \chi}=\sqrt{1-\frac{2m}{r}+\dot{r}^2}.
\end{align}
Taking the partial derivative with respect to $\tau$ and $\chi$ and using equations \eqref{eq:101}, \eqref{eq:102}, \eqref{eq:103} as well as the mass equations given by \eqref{eq:26}, we obtain
\begin{align}
  \label{eq:349}
  \left(\frac{\p e^{\omega}}{\p \tau}\right)_0=\frac{a_{-0}-a_{+0}}{2r_0},\qquad \left(\frac{\p e^{\omega}}{\p \chi}\right)_0=\frac{a_{-0}-a_{+0}}{2r_0}+\frac{k}{2},
\end{align}
where $k$ is the constant defined in \eqref{eq:49}. The partial derivatives of $\psi$ involve the partial derivatives of $a_-$. Using the first of \eqref{eq:62} together with \eqref{eq:101}, \eqref{eq:102} as well as \eqref{eq:26} we find
\begin{align}
  \label{eq:350}
  \left(\frac{\p a_-}{\p\tau}\right)_0=\frac{\mu_0}{2r_0},\qquad\left(\frac{\p a_-}{\p\chi}\right)_0=\frac{a_{-0}}{2r_0}(a_{+0}-a_{-0})+\frac{\mu_0}{4r_0}-2\pi r_0.
\end{align}
Using theses expressions and again \eqref{eq:101}, \eqref{eq:102}, \eqref{eq:103} for the partial derivatives of $r$ and \eqref{eq:26} for the partial derivatives of $m$ we obtain
\begin{align}
  \label{eq:351}
  \left(\frac{\p \psi}{\p \tau}\right)_0&=-\frac{\mu_0+2\pi r_0^2}{a_{-0}r_0^2}(a_{+0}-a_{-0})-\frac{\mu_0(\mu_0-4\pi r_0^2)}{2a_{-0}^2r_0^2},\\
  \label{eq:352}
  \left(\frac{\p \psi}{\p \chi}\right)_0&=\frac{-\mu_0+2\pi r_0^2}{2a_{-0}r_0^2}(a_{+0}+a_{-0})-\frac{\mu_0-4\pi r_0^2}{2a_{-0}^2r_0^2}\left[(a_{+0}-a_{-0})a_{-0}+\frac{\mu_0}{2}-4\pi r_0^2\right].
\end{align}
Substituting \eqref{eq:349}, \eqref{eq:351} and \eqref{eq:352} into \eqref{eq:346} we arrive at
\begin{align}
  \label{eq:353}
  \frac{dN_\ast}{d\tau}=\frac{\mu_0-4\pi r_0^2}{2a_{-0}r_0}\left(1+\frac{1}{2}\frac{d\chi_\ast}{d\tau}\right)+A \tau+B\frac{d}{d\tau}(\tau\chi_\ast)+C \chi_\ast\frac{d\chi_\ast}{d\tau}+\mathcal{O}(\tau^2),
\end{align}
where $A$, $B$ and $C$ are given in the statement of the proposition. The result for $N_\ast$ follows by integration.

We now turn to $\delta$. Integrating (6.6b) of \cite{I} with respect to $v$ from $v=0$ up to $v=u$ we get
\begin{align}
  \label{eq:354}
  \nu_\ast(u)=\nu(u,0)e^{N_\ast(u)}.
\end{align}
Together with \eqref{eq:321} we can express the derivative of $Y$ (recall that $Y(u)=r_0-r(u,0)$) and we obtain
\begin{align}
  \label{eq:355}
  \frac{d\delta}{d\tau}&=1-\frac{1}{2}e^{-\omega_\ast}\beta-\frac{a_{-\ast}}{a_{-0}}(1-\beta)e^{-N_\ast}\notag\\
&=\left(1-\frac{1}{2}e^{-\omega_\ast}\right)\beta+\left(1-\frac{a_{-\ast}}{a_{-0}}e^{-N_\ast}\right)(1-\beta).
\end{align}
Recalling again that $e^{\omega_0}=\frac{1}{2}$ we rewrite
\begin{align}
  \label{eq:356}
  \frac{d\delta}{d\tau}=\left(e^{\omega_\ast}-e^{\omega_0}\right)\frac{d\chi_\ast}{d\tau}+\left(1-\frac{a_{-\ast}}{a_{-0}}e^{-N_\ast}\right)\left(1-e^{\omega_\ast}\frac{d\chi_\ast}{d\tau}\right).
\end{align}
From \eqref{eq:349} we have
\begin{align}
  \label{eq:357}
  e^{\omega_\ast}-e^{\omega_0}=\frac{a_{-0}-a_{+0}}{2r_0}\tau+\left(\frac{a_{-0}-a_{+0}}{2r_0}+\frac{k}{2}\right)\chi_\ast+\mathcal{O}(\tau^2),
\end{align}
while from \eqref{eq:350}, together with the result for $N_\ast$ we obtain
\begin{align}
  \label{eq:359}
  1-\frac{a_{-\ast}}{a_{-0}}e^{-N_\ast}=\frac{1}{2r_0a_{-0}}\left\{-4\pi r_0^2\tau-\left[(a_{+0}-a_{-0})a_{-0}-2\pi r_0^2\right]\chi_\ast\right\}+\mathcal{O}(\tau^2).
\end{align}
Substituting \eqref{eq:357}, \eqref{eq:359} in \eqref{eq:356} we arrive at
\begin{align}
  \label{eq:360}
  \frac{d\delta}{d\tau}&=\frac{1}{2}\frac{d\chi_\ast}{d\tau}\Bigg\{-\frac{1}{a_{-0}r_0}\left[(a_{+0}-a_{-0})a_{-0}-2\pi r_0^2\right]\tau-\frac{1}{2a_{-0}r_0}\left[(a_{+0}-a_{-0})a_{-0}+2\pi r_0^2-2ka_{-0}r_0\right]\chi_\ast\Bigg\}\notag\\
&\qquad-\frac{1}{2r_0a_{-0}}\left[4\pi r_0^2\tau+\left((a_{+0}-a_{-0})a_{-0}-2\pi r_0^2\right)\chi_\ast\right]+\mathcal{O}(\tau^2).
\end{align}
Integrating we get
\begin{align}
  \label{eq:361}
  \delta=\delta^{(2)}+\mathcal{O}(\tau^3).
\end{align}
The proof of the proposition is complete.
\end{proof}

We note
\begin{align}
  \label{eq:362}
Y=\frac{a_{-0}}{2}\left(\tau-\frac{\chi_\ast}{2}\right)-\frac{a_{-0}}{2}\delta^{(2)}+\mathcal{O}(\tau^3).
\end{align}

In the following we will use the notation $f^{[n]}=\sum_{i=0}^n f^{(i)}$. The following proposition expresses $E$ up to cubic error terms.
\begin{proposition} The function $E=1/\gamma^2-1$ has the following asymptotic form
  \begin{align}
    \label{eq:363}
    E=E^{[2]}+\mathcal{O}(\tau^3),
\end{align}
where
\begin{align}
      \label{eq:364}
    E^{[2]}&=(i+2i_0)\tau^2+(2i_0-i)\tau\chi_\ast+(2i_0-3i)\frac{\chi_\ast^2}{4},
  \end{align}
where $i$ and $i_0$ are the positive real numbers defined by \eqref{eq:117} and \eqref{eq:124} respectively.
  \label{proposition4}  
\end{proposition}
\begin{proof}
Let
\begin{align}
  \label{eq:365}
  B:=e^{-N_\ast}((H\circ Y)+M_\ast).
\end{align}
Using the fact that $N_\ast^{(0)}=M_\ast^{(0)}=M_\ast^{(1)}=0$ (cf. proposition \ref{proposition10}) it follows
\begin{align}
  \label{eq:366}
  B^{(0)}&=(H\circ Y)^{(0)},\\
  \label{eq:367}
  B^{(1)}&=-N_\ast^{(1)}(H\circ Y)^{(0)}+(H\circ Y)^{(1)},\\
  \label{eq:368}
  B^{(2)}&=\left(\frac{1}{2}\left(N_\ast^{(1)}\right)^2-N_\ast^{(2)}\right)(H\circ Y)^{(0)}-N_\ast^{(1)}(H\circ Y)^{(1)}+(H\circ Y)^{(2)}+M_\ast^{(2)}.
\end{align}
Let
\begin{align}
  \label{eq:369}
  G:=r_\ast^2E.
\end{align}
Then, from \eqref{eq:285},
\begin{align}
  \label{eq:370}
  G=a_{-\ast}^2(B-\phi_\ast)^2-r_\ast^2.
\end{align}
In view of \eqref{eq:286}, \eqref{eq:301}, \eqref{eq:302} we obtain
\begin{align}
  \label{eq:371}
  G^{(0)}=0.
\end{align}
Using \eqref{eq:366} we obtain
\begin{align}
  \label{eq:372}
  G^{(1)}=2\left(a_{-\ast}^{(0)}\right)^2\left(B^{(1)}(H\circ Y)^{(0)}-B^{(0)}\tau\right)+2a_{-\ast}^{(1)}a_{-\ast}^{(0)}\left((H\circ Y)^{(0)}\right)^2-2r_\ast^{(0)}r_\ast^{(1)},
\end{align}
where we made use of $\phi_\ast^{(1)}=\tau$. For $B^{(1)}$ we use \eqref{eq:367} together with \eqref{eq:286}, \eqref{eq:308} and
\begin{align}
  \label{eq:373}
  (H\circ Y)^{(1)}=\frac{dH}{ds}(0)\,Y^{(1)}
\end{align}
together with \eqref{eq:290} and \eqref{eq:362}. We obtain
\begin{align}
  \label{eq:374}
  B^{(1)}&=-N_\ast^{(1)}(H\circ Y)^{(0)}+\frac{dH}{ds}(0)\,Y^{(1)}\notag\\
&=\left(3-\frac{2a_{+0}}{a_{-0}}+\frac{1}{a_{-0}^2}\right)\frac{\tau}{2}+\left(-3+\frac{1}{a_{-0}^2}-\frac{8\pi r_0^2}{a_{-0}^2}\right)\frac{\chi_\ast}{4}.
\end{align}
Using \eqref{eq:101} we have
\begin{align}
  \label{eq:375}
  r_\ast^{(1)}=\frac{1}{2}(a_{+0}-a_{-0})\tau+\frac{1}{4}(a_{+0}+a_{-0})\chi_\ast,
\end{align}
while from \eqref{eq:350},
\begin{align}
  \label{eq:376}
  a_{-\ast}^{(1)}=\frac{1}{2r_0}(1-a_{+0}a_{-0})\tau+\frac{1}{4r_0}(a_{+0}a_{-0}-2a_{-0}^2+1-8\pi r_0^2)\chi_\ast.
\end{align}
Substituting into \eqref{eq:372} we find
\begin{align}
  \label{eq:377}
  G^{(1)}=0.
\end{align}

From \eqref{eq:370} it follows
\begin{align}
  \label{eq:378}
  G^{(2)}=&\left(B^{(0)}\right)^2\left(2a_{-\ast}^{(2)}a_{-\ast}^{(0)}+\left(a_{-\ast}^{(1)}\right)^2\right)+4a_{-\ast}^{(0)}a_{-\ast}^{(1)}\left(B^{(0)}B^{(1)}-B^{(0)}\tau\right)\notag\\
&\qquad+\left(a_{-\ast}^{(0)}\right)^2\left(\left(B^{(1)}-\tau\right)^2+2B^{(0)}B^{(2)}\right)-2r_\ast^{(0)}r_\ast^{(2)}-\left(r_\ast^{(1)}\right)^2.
\end{align}
We need expressions for $r_\ast^{(2)}$, $a_{-\ast}^{(2)}$, $B^{(2)}$. Using the partial derivatives of second order of $r_\ast$ at $N^-$ as given by \eqref{eq:102}, \eqref{eq:103} yields
\begin{align}
  \label{eq:379}
  r_\ast^{(2)}&=\frac{1}{2}\left(\frac{\p^2 r}{\p\tau^2}\right)_0\tau^2+\left(\frac{\p^2 r}{\p\tau\p\chi}\right)_0\tau\chi_\ast+\frac{1}{2}\left(\frac{\p^2 r}{\p\chi^2}\right)_0\chi_\ast^2\notag\\
&=-\frac{1-a_{+0}a_{-0}}{4r_0}\tau^2-\frac{a_{+0}^2-a_{-0}^2}{4r_0}\tau\chi_\ast\notag\\
&\qquad+\left(\frac{1}{16r_0}\left(-3a_{+0}^2+a_{+0}a_{-0}+a_{-0}^2+1-8\pi r_0^2\right)+\frac{k}{8}(a_{+0}+a_{-0})\right)\chi_\ast^2.
\end{align}
$a_{-\ast}^{(2)}$ can be written as
\begin{align}
  \label{eq:380}
  a_{-\ast}^{(2)}=\frac{1}{2}\left(\frac{\p^2 a_-}{\p\tau^2}\right)_0\tau^2+\left(\frac{\p^2 a_-}{\p\tau\p\chi}\right)_0\tau\chi_\ast+\frac{1}{2}\left(\frac{\p^2 a_{-}}{\p\chi^2}\right)_0\chi_\ast^2.
\end{align}
Recalling the fact that $2m/r+\dot{r}^2$ is constant along flow lines together with the fact that $m$ is constant along the flow lines (the first of \eqref{eq:26}), equation \eqref{eq:62} yields
\begin{align}
  \label{eq:381}
  \left(\frac{\p^2a_-}{\p\tau^2}\right)_{0}=-\frac{2m_0\dot{r}_0}{r_0^3}=\frac{a_{+0}a_{-0}-1}{2r_0^2}(a_{+0}-a_{-0}).
\end{align}
Equation \eqref{eq:62} together with the second of \eqref{eq:26} yield
\begin{align}
  \label{eq:382}
  \left(\frac{\p^2a_-}{\p\tau\p \chi}\right)_{0}=\frac{a_{+0}+a_{-0}}{4}\left(4\pi-\frac{2m_0}{r_0^3}\right)=\frac{1}{4r_0^2}(4\pi r_0^2-1+a_{+0}a_{-0})(a_{+0}+a_{-0}).
\end{align}
Using again equation \eqref{eq:62} we obtain
\begin{align}
  \label{eq:383}
  \frac{\p a_-}{\p\chi}=\frac{2}{a_++a_-}\left[-a_-\frac{\p\dot{r}}{\p\chi}+\left(\frac{m}{r^2}-4\pi \rho r\right)\frac{\p r}{\p\chi}\right].
\end{align}
Therefore,
\begin{align}
  \label{eq:384}
  \left(\frac{\p^2 a_-}{\p\chi^2}\right)_{0}=&
-\frac{2}{a_{+0}+a_{-0}}\left(\frac{\p a_-}{\p \chi}\right)_{0}\left[\left(\frac{\p a_+}{\p\chi}\right)_{0}+\left(\frac{\p a_-}{\p \chi}\right)_{0}\right]\nonumber\\
&+\frac{2}{a_{+0}+a_{-0}}\Bigg\{-\left(\frac{\p a_-}{\p\chi}\right)_{0}\left(\frac{\p\dot{r}}{\p\chi}\right)_{0}-a_{-0}\left(\frac{\p^2\dot{r}}{\p\chi^2}\right)_{0}\nonumber\\
&\hspace{23mm}+\left[\frac{1}{r_0^2}\left(\frac{\p m}{\p\chi}\right)_{0}-\frac{2m_0}{r_0^3}\left(\frac{\p r}{\p\chi}\right)_{0}-4\pi\left(\frac{\p r}{\p\chi}\right)_{0}\right]\left(\frac{\p r}{\p\chi}\right)_{0}\nonumber\\
&\hspace{70mm}+\left(\frac{m_0}{r_0^2}-4\pi r_0\right)\left(\frac{\p^2 r}{\p\chi^2}\right)_{0}\Bigg\}.
\end{align}
Since $(\p^2\dot{r}/\p\chi^2)_{0}$ is given by \eqref{eq:121}, all terms are determined and we have
\begin{align}
  \label{eq:385}
\left(\frac{\p^2a_-}{\p\chi^2}\right)_{0}=\frac{1}{4 r_0^2}\bigg(kr_0(1-8\pi r_0^2)+a_{+0}(4\pi r_0^2-1)+2a_{-0}^3-a_{-0}^2(2kr_0+a_{+0})+r_0a_{-0}(ka_{+0}-4ir_0)\bigg).
\end{align}

 Now in the expression for $B^{(2)}$ as given by \eqref{eq:368} there appear $N_\ast^{(1)}$, $N_\ast^{(2)}$ $M_\ast^{(2)}$. Those are given by proposition \ref{proposition10}. Since
\begin{align}
  \label{eq:386}
  (H\circ Y)^{(2)}=\frac{dH}{ds}(0)\,Y^{(2)}+\frac{1}{2}\frac{d^2H}{dY^2}(0)\left(Y^{(1)}\right)^2,
\end{align}
using \eqref{eq:290}, \eqref{eq:294}, and
\begin{align}
  \label{eq:387}
  Y^{(2)}=-\frac{a_{-0}}{2}\delta^{(2)}
\end{align}
(cf.~\eqref{eq:362}), and since $\delta^{(2)}$ is given by proposition \ref{proposition10}, $(H\circ Y)^{(2)}$ can be expressed. Putting things together we arrive at the following expression for $G^{(2)}$
\begin{align}
  \label{eq:388}
  G^{(2)}=&\left[\frac{3}{4}(a_{-0}-a_{+0})^2+(i+2\pi)r_0^2\right]\tau^2+\left[\frac{3}{4}(a_{-0}-a_{+0})^2-(i+2\pi)r_0^2\right]\tau\chi_\ast\\\notag
&\qquad+\frac{1}{4}\left[\frac{3}{4}(a_{-0}-a_{+0})^2+(2\pi-3i)r_0^2\right]\chi_\ast^2.
\end{align}

We expand $E=G/r_\ast^2$:
\begin{align}
  \label{eq:389}
  E=\frac{1}{r_0^2}\left(1-\frac{2}{r_0}\left(\frac{\p r}{\p \tau}\right)_0\tau-\frac{2}{r_0}\left(\frac{\p r}{\p \chi}\right)_0\chi_\ast\right)\left(G^{(1)}+G^{(2)}\right)+\mathcal{O}(\tau^3)=E^{(1)}+E^{(2)}+\mathcal{O}(\tau^3),
\end{align}
where
\begin{align}
  \label{eq:390}
  E^{(1)}=\frac{1}{r_0^2}G^{(1)},\qquad E^{(2)}=\frac{1}{r_0^2}G^{(2)}-\frac{2}{r_0^3}\left(\left(\frac{\p r}{\p \tau}\right)_0\tau+\left(\frac{\p r}{\p \chi}\right)_0\chi_\ast\right)G^{(1)}.
\end{align}
However,  since $G^{(1)}=0$, we have
\begin{align}
  \label{eq:391}
  E^{[2]}=\frac{G^{(2)}}{r_0^2}.
\end{align}
Substituting \eqref{eq:388} into \eqref{eq:391} yields the expression for $E^{[2]}$. 
\end{proof}

Let us introduce the soft phase coordinates $t$, $x$ by
\begin{align}
  \label{eq:392}
  t:=\tau+\frac{1}{2}\chi,\qquad x:=\tau-\frac{1}{2}\chi.
\end{align}
We note that
\begin{align}
  \label{eq:393}
  \liminf_{\tau\rightarrow 0}\frac{t}{\tau}\geq 1.
\end{align}
Equation \eqref{eq:127} yields
\begin{align}
  \label{eq:394}
  X^{[2]}=(2i-2i_0)t^2-2itx_\ast.
\end{align}
Here, $\mathscr{B}$ is represented by $t\mapsto (t,x_\ast(t))$. By proposition \ref{proposition4},
\begin{align}
  \label{eq:395}
  E^{[2]}=(2i_0-i)t^2+2itx_\ast.
\end{align}
We note that
\begin{align}
  \label{eq:396}
  \frac{E^{[2]}+X^{[2]}}{t^2}=i>0.
\end{align}
Let
\begin{align}
  \label{eq:397}
  F(t,x):=\frac{e^{-\omega_0}-e^{-\omega(t,x)}}{e^{-\omega_0}+e^{-\omega(t,x)}}.
\end{align}
We note that $F$ is a smooth function of $t$, $x$ whose Taylor expansion begins wih terms of degree one. Since $\beta=e^{\omega_\ast}\frac{d\chi_\ast}{d\tau}$, in view of \eqref{eq:268}, \eqref{eq:269} it follows
\begin{align}
  \label{eq:398}
  \frac{dx_\ast}{dt}=z^2\frac{1+F_\ast z^{-2}}{1+F_\ast z^2}=\frac{X^{[2]}}{E^{[2]}+X^{[2]}}+\mathcal{O}(t)=\frac{(2i-2i_0)t-2ix_\ast}{it}+\mathcal{O}(t).
\end{align}
Hence
\begin{align}
  \label{eq:399}
  x_\ast (t)=\frac{2}{3}\left(1-\frac{i_0}{i}\right)t+\mathcal{O}(t^2).
\end{align}
It follows
\begin{align}
  \label{eq:400}
  \chi_\ast (\tau)=\frac{2i+4i_0}{5i-2i_0}\tau+\mathcal{O}(\tau^2).
\end{align}
We have thus shown
\begin{align}
  \label{eq:401}
  \lim_{\tau\rightarrow 0}\beta(\tau)=\beta_0,
\end{align}
where $\beta_0$ is given by \eqref{eq:159}. Hence \eqref{eq:266} and \eqref{eq:267} hold.

Let the shock curve $\mathscr{B}$ be represented by
\begin{align}
  \label{eq:402}
  K:[0,\hat{\tau}]\rightarrow \mathscr{V},\qquad \tau\mapsto (\tau,\chi_\ast(\tau)),
\end{align}
where the domain $\mathscr{V}$ was defined in \eqref{eq:156}. In the following we will denote by $Q_n$ a smooth function with arguments in $\mathscr{V}$ whose Taylor expansion begins with $n$'th degree terms. We will denote by $P_n$ a polynomial in $\tau$, $\chi$ of $n$'th degree and by $A_n$ a homogeneous such polynomial. Furthermore, we define
\begin{align}
  \label{eq:403}
  \bar{Q}_n:=Q_n\circ K,\qquad \bar{P}_n:=P_n\circ K,\qquad \bar{A}_n:=A_n\circ K.
\end{align}
We will use the following notation for the remainders
\begin{align}
  \label{eq:404}
  \{f\}_n:=f-f^{[n-1]}.
\end{align}

The following proposition will be used in the uniqueness proof.
\begin{proposition}
The remainder $\{H\circ Y\}_3$ can be expressed as
\begin{align}
    \label{eq:405}
    \{H\circ Y\}_3=A_3+A_4+P_2\{\delta\}_3+P_0\left(\{\delta\}_3\right)^2+Y^3(G\circ Y),
\end{align}
where $Y$ is considered as a function of $\tau$ and $G$ is given by
\begin{align}
  \label{eq:406}
  H\circ Y=H_0+\frac{dH}{ds}(0)\,Y+\frac{1}{2}\frac{d^2H}{dY^2}(0)\,Y^2+Y^3(G\circ Y).
\end{align}
(Recall that $H$ is a smooth function, hence so is $G$). The remainders $\{B\}_3$, $\{E\}_3$ and $d\{\delta\}_3/d\tau$ can be expressed as
\begin{align}
    \label{eq:407}
    \{B\}_3&=e^{-N_\ast^{[2]}-\{N_\ast\}_3}\left(\{H\circ Y\}_3+\{M_\ast\}_3\right)+\bar{Q}_0\left(e^{-\{N_\ast\}_3}-1\right)+\bar{Q}_3,\\
    \label{eq:408}
    \{E\}_3&=\bar{Q}_3+\bar{Q}_0(\{B\}_3-\Psi)+\bar{Q}_0(\{B\}_3-\Psi)^2,\\
    \label{eq:409}
\frac{d}{d\tau}\{\delta\}_3&=\bar{Q}_2+\bar{Q}_1\left(e^{-N_\ast^{[2]}-\{N_\ast\}_3}-1\right)-\left(e^{-N_\ast^{[2]}-\{N_\ast\}_3}-1+N_\ast^{(1)}\right)\notag\\
&\qquad+\bigg[\bar{Q}_2+\bar{Q}_1\left(e^{-N_\ast^{[2]}-\{N_\ast\}_3}-1\right)+\frac{1}{2}\left(e^{-N_\ast^{[2]}-\{N_\ast\}_3}-1+N_\ast^{(1)}\right)\notag\\
&\hspace{40mm}+\bar{Q}_1\left(1-e^{-N_\ast^{[2]}-\{N_\ast\}_3}\right)+\bar{Q}_2e^{-N_\ast^{[2]}-\{N_\ast\}_3} \bigg]\frac{d\chi_\ast}{d\tau}.
\end{align}
The derivative of the function $\Psi$ can be written as
\begin{align}
    \label{eq:410}
    \frac{d\Psi}{d\tau}&=\frac{\tilde{Y}}{1+z^2}\left[h\left(\frac{\tilde{Y}}{z^2}\right)-h(-\tilde{Y})\right],
\end{align}
where
\begin{align}
  \label{eq:411}
  \tilde{Y}:=X(1-z^2),\qquad h(u):=\frac{\sqrt{1+u}-1}{u}.
\end{align}
\label{proposition6}
\end{proposition}
\begin{proof}
Substituting (cf. \eqref{eq:303})
\begin{align}
  \label{eq:412}
  Y=\frac{a_{-0}}{2}\left(\tau-\frac{\chi_\ast}{2}-\delta^{[2]}-\{\delta\}_3\right)
\end{align}
into \eqref{eq:406} the result for $\{H\circ Y\}_3$ follows.

Using \eqref{eq:365} together with
\begin{align}
  \label{eq:611}
  N_\ast=N_\ast^{[2]}+\{N_\ast\}_2,\qquad(H\circ Y)=(H\circ Y)^{[2]}+\{H\circ Y\}_3,\qquad M_\ast=M_\ast^{[2]}+\{M_\ast\}_3, 
\end{align}
we can write
\begin{align}
  \label{eq:413}
  \{B\}_3&=B-B^{[2]}\notag\\
&=e^{-N_\ast^{[2]}-\{N_\ast\}_3}\left(\{H\circ Y\}_3+\{M_\ast\}_3\right)+e^{-N_\ast^{[2]}}\left(e^{-\{N_\ast\}_3}-1\right)\left((H\circ Y)^{[2]}+M_\ast^{[2]}\right)\notag\\
&\qquad+e^{-N_\ast^{[2]}}\left((H\circ Y)^{[2]}+M_\ast^{[2]}\right)-B^{[2]}.
\end{align}
Using \eqref{eq:366}, \eqref{eq:367}, \eqref{eq:368} we obtain
\begin{align}
  \label{eq:414}
  B^{[2]}&=B^{(0)}+B^{(1)}+B^{(2)}\notag\\
&=(H\circ Y)^{(0)}\left(1-N_\ast^{(1)}+\tfrac{1}{2}\left(N_\ast^{(1)}\right)^2-N_\ast^{(2)}\right)+(H\circ Y)^{(1)}\left(1-N_\ast^{(1)}\right)+(H\circ Y)^{(2)}+M_\ast^{(2)}.
\end{align}
Therefore,
\begin{align}
  \label{eq:415}
  e^{-N_\ast^{[2]}}\left((H\circ Y)^{[2]}+M_\ast^{[2]}\right)-B^{[2]}&=\left(e^{-N_\ast^{[2]}}-1\right)(H\circ Y)^{(2)}+\left(e^{-N_\ast^{[2]}}-1\right)M_\ast^{(2)}\notag\\
&\hspace{0.5cm}+\left(e^{-N_\ast^{[2]}}-1+N_\ast^{(1)}\right)(H\circ Y)^{(1)}\notag\\
&\hspace{0.5cm}+\left(e^{-N_\ast^{[2]}}-1+N_\ast^{(1)}-\tfrac{1}{2}\left(N_\ast^{(1)}\right)^2+N_\ast^{(2)}\right)(H\circ Y)^{(0)}
\end{align}
and the result for $\{B\}_3$ follows.

Using \eqref{eq:285}, \eqref{eq:301}, \eqref{eq:365} we can write
\begin{align}
  \label{eq:416}
  E=\left(\frac{a_{-\ast}}{r_\ast}\right)^2\left(B^{[2]}-\tau\right)^2-1+2\left(\frac{a_{-\ast}}{r_\ast}\right)^2\left(B^{[2]}-\tau\right)\left(\{B\}_3-\Psi\right)+\left(\frac{a_{-\ast}}{r_\ast}\right)^2\left(\{B\}_3-\Psi\right)^2
\end{align}
and the result for $\{E\}_3$ follows.

We now consider the remainder $\{\delta\}_3$. We have (see \eqref{eq:356})
\begin{align}
  \label{eq:417}
  \frac{d\delta}{d\tau}=\left(e^{\omega_\ast}-e^{\omega_0}\right)\frac{d\chi_\ast}{d\tau}+\left(1-\frac{a_{-\ast}}{a_{-0}}e^{-N_\ast}\right)\left(1-e^{\omega_\ast}\frac{d\chi_\ast}{d\tau}\right),
\end{align}
while
\begin{align}
  \label{eq:418}
  \frac{d\delta^{[2]}}{d\tau}&=\left[\left(\frac{\p e^{\omega}}{\p\tau}\right)_0\tau+\left(\frac{\p e^{\omega}}{\p\chi}\right)_0\chi_\ast\right]\frac{d\chi_\ast}{d\tau}\notag\\
&\qquad-\frac{1}{a_{-0}}\left[\left(\frac{\p a_{-}}{\p\tau}\right)_0\tau+\left(\frac{\p a_{-}}{\p \chi}\right)_0\chi_\ast-a_{-0}N_\ast^{(1)}\right]\left(1-e^{\omega_0}\frac{d\chi_\ast}{d\tau}\right).
\end{align}
Subtracting we obtain
\begin{align}
  \label{eq:419}
  \frac{d}{d\tau}\{\delta\}_3=I_1\frac{d\chi_\ast}{d\tau}-\frac{I_2}{a_{-0}}\left(1-e^{\omega_0}\frac{d\chi_\ast}{d\tau}\right)-\left(e^{\omega_\ast}-e^{\omega_0}\right)\left(1-\frac{a_{-\ast}}{a_{-0}}e^{-N_\ast}\right)\frac{d\chi_\ast}{d\tau},
\end{align}
where
\begin{align}
  \label{eq:420}
  I_1&:=e^{\omega_\ast}-e^{\omega_0}-\left[\left(\frac{\p e^{\omega}}{\p\tau}\right)_0\tau+\left(\frac{\p e^{\omega}}{\p\chi}\right)_0\chi_\ast\right],\\
  \label{eq:421}
  I_2&:=a_{-\ast}-a_{-0}-\left[\left(\frac{\p a_{-\ast}}{\p\tau}\right)_0\tau+\left(\frac{\p a_{-\ast}}{\p \chi}\right)_0\chi_\ast\right]\notag\\
&\qquad+(a_{-\ast}-a_{-0})\left(e^{-N_\ast}-1\right)+a_{-0}\left(e^{-N_\ast}-1+N_\ast^{(1)}\right).
\end{align}
Therefore, we can write
\begin{align}
  \label{eq:422}
  \frac{d}{d\tau}\{\delta\}_3&=\{e^{\omega_\ast}\}_2\frac{d\chi_\ast}{d\tau}-\frac{1}{a_{-0}}\bigg[\{a_{-\ast}\}_2+(a_{-\ast}-a_{-0})\left(e^{-N_\ast^{[2]}-\{N_\ast\}_3}-1\right)\notag\\
&\hspace{50mm}+a_{-0}\left(e^{-N_\ast^{[2]}-\{N_\ast\}_3}-1+N_\ast^{(1)}\right)\bigg]\left(1-e^{\omega_0}\frac{d\chi_\ast}{d\tau}\right)\notag\\
&\qquad-(e^{\omega_\ast}-e^{\omega_0})\left(1-\frac{a_{-\ast}}{a_{-0}}e^{-N_\ast^{[2]}-\{N_\ast\}_3}\right)\frac{d\chi_\ast}{d\tau}.
\end{align}
The expression \eqref{eq:409} then follows.

Using (4.6c) of \cite{II} and taking into account the first of \eqref{eq:268}, it follows
\begin{align}
  \label{eq:423}
  \frac{d\phi_\ast}{d\tau}=\frac{z}{1+z^2}\sqrt{z^2+X(1-z^2)}+\frac{1}{1+z^2}\sqrt{1-X(1-z^2)}.
\end{align}
Using the definition of $\Psi$ as given by \eqref{eq:301},  expression \eqref{eq:410} follows.
\end{proof}

We are now ready to prove the following existence theorem.
\begin{theorem}
Let $(r,\omega,\rho)$ be a soft-phase solution corresponding to smooth initial data on $\Sigma$ and let $\mathscr{V}$ be the domain
\begin{align*}
  \mathscr{V}=\{(\tau,\chi):\chi\geq 0,\tau_-(\chi)\leq\tau\leq\hat{\tau}(\chi)\}
\end{align*}
where $(0,0)=N^-$ is an incoming boundary null point of $\Sigma$, $\tau=\tau_-(\chi)$ is the equation of $C^-$, the outgoing null curve issuing from $N^-$, while $\tau=\hat{\tau}(\chi)$ is the equation of the curve issuing at $N^-$, where, corresponding to the given soft-phase initial data on $\Sigma$, $\rho(\tau,\chi)$ along each flow line first becomes equal to 1. Let $R$ be a function defined on an interval $[0,\hat{t}]$ and representing $r$ as a function of $\phi$ along $C^{\ast -}$, the incoming null curve issuing from $N^-$. Then there is a $\hat{\tau}>0$ and a solution to the problem of formation of a free phase boundary such that $\mathscr{B}$ is a $C^1$ curve: $[0,\hat{\tau}]\mapsto \mathscr{V}$, $\tau\mapsto \chi_\ast(\tau)$, issuing from $N^-$, which has positive velocity $\beta$ relative to the soft-phase flow lines, is strictly timelike $(\beta<1)$ and contained in the interior of $\mathscr{V}$ in $(0,\hat{\tau}]$. Furthermore $\beta(0)=(5+l)/(5l+1)$ where $l$ is given by
\begin{align}
  \label{eq:424}
  l=-4\frac{\rho_{\chi\chi}|_{N^-}}{\rho_{\tau\tau}|_{N^-}}.
\end{align}
Also $r$, $m$ and $\phi$ are $C^1$ functions while $\nu$, $\kappa$, $\zeta$, $\eta$ are $C^0$ functions on $\mathscr{U}(\hat{\tau})$, defining a solution corresponding to genuine hard phase, except at $N^-$ where $\sigma=1$.
\label{theorem_1}
\end{theorem}
\begin{proof}
By virtue of lemma \ref{lemma_sufficiently} and using \eqref{eq:201}, \eqref{eq:203}, \eqref{eq:204} $\nu$ is continuous in
\begin{align}
  \label{eq:425}
  \mathscr{U}_{(0,\hat{\tau}]}^+:=\mathscr{U(\hat{\tau})}\setminus (0,0)
\end{align}
and $\kappa$ is continuous in
\begin{align}
  \label{eq:426}
  \mathscr{U}_{(0,\hat{\tau}]}^-:=\mathscr{U}(\hat{\tau})\setminus ([0,\hat{\tau}]\times\{0\}).
\end{align}
Since
\begin{align}
  \label{eq:427}
 \mathscr{U}_{(0,\hat{\tau}]}^+\cap \mathscr{U}_{(0,\hat{\tau}]}^-=\mathscr{U}_{(0,\hat{\tau}]}^-,
\end{align}
the functions $\nu$, $\zeta$ are thus continuously differentiable with respect to $v$ while the functions $\kappa$, $\eta$ are continuously differentiable with respect to $u$ in $\mathscr{U}_{(0,\hat{\tau}]}^-$ (cf. (6.6a,b), (6.7a,b) of \cite{III}). Also the functions $r$, $\phi$, $m$ are $C^1$ functions in $\mathscr{U}_{(0,\hat{\tau}]}^-$. The results of the previous section show that $\beta\rightarrow \beta_0$ as $\tau\rightarrow 0$, where $\beta_0$ is given by \eqref{eq:159}. This implies that $\nu_\ast$ and $\kappa_\ast$ (see \eqref{eq:201}) extend continuously to $\tau=0$ and
\begin{align}
  \label{eq:428}
  \nu_\ast (0)=\frac{a_{-0}}{2}(1-\beta_0),\qquad \kappa_\ast(0)=\frac{1}{2a_{-0}}(1+\beta_0),
\end{align}
which in turn implies that $\nu$ extends continuously to the point $(0,0)$, $\kappa$ extends continuously to $[0,\hat{\tau}]\times \{0\}$, the past boundary of $\mathscr{U}(\hat{\tau})$, the partial derivatives of the functions $\nu$, $\zeta$ with respect to $v$ and of the functions $\kappa$, $\eta$ with respect to $u$ extend continuously to $[0,\hat{\tau}]\times \{0\}$, while the functions $r$, $m$ and $\phi$ are $C^1$ on the whole of $\mathscr{U}(\hat{\tau})$.

It remains to show that the solution obtained in $\mathscr{U}(\hat{\tau})$ corresponds to a genuine hard phase. More precisely, it suffices to show that we can find a $\hat{\tau}'\in(0,\hat{\tau}]$ such that the restriction of the solution to $\mathscr{U}(\hat{\tau}')$ satisfies $\sigma>1$ except at $(0,0)$. We shall do this with the help of the (interior) barrier function $e$ of \eqref{eq:129}. Now, on $(0,\hat{\tau}]$ we have $\sigma_\ast>1$ (see (4.20a,b) of \cite{II}). Thus if, on the contrary, such a $\hat{\tau}'$ cannot be found, then for every $\tau_1\in (0,\hat{\tau}]$, the set
\begin{align}
  \label{eq:429}
  \mathscr{E}(\tau_1):=\{(u,v)\in\mathscr{U}(\tau_1)\setminus(0,0):\sigma(u,v)\leq1\}
\end{align}
is not empty; consequently, there is a $\check{v}\in [0,\tau_1)$ such that $\mathscr{E}(\tau_1)$ has a non-empty intersection with the incoming null segment
\begin{align}
  \label{eq:430}
  I_{\tau_1}^-(\check{v}):=[\check{v},\tau_1]\times \{\check{v}\}.
\end{align}
Suppose that $\check{v}>0$. Then since $\sigma_\ast(\check{v})>1$, there is a first point, $(\check{u},\check{v})$ along $I_{\tau_1}^-(\check{v})$ at which $\sigma=1$, i.e.
\begin{align}
  \label{eq:431}
  \check{u}=\sup\{u:\sigma(u,\check{v})>1\}.
\end{align}
Thus
\begin{align}
  \label{eq:432}
  \frac{\partial \sigma}{\partial u}(\check{u},\check{v})\leq 0
\end{align}
must hold. However, by (1.23b) of \cite{III},
\begin{align}
  \label{eq:433}
  \frac{\partial \sigma}{\partial u}(\check{u},\check{v})=2(\nu\zeta e)(\check{u},\check{v}).
\end{align}
Using proposition \ref{proposition2}, it follows that
\begin{align}
  \label{eq:434}
  \inf_{I_{\tau_1}^-(\check{v})}e>0
\end{align}
if $\tau_1$ is sufficiently small. It follows that $\check{v}=0$; hence $\mathscr{E}(\tau_1)\subset I_{\tau_1}^-(0)\setminus (0,0)$ and, by continuity, $\sigma=1$ on $\mathscr{E}(\tau_1)$. Consequently, either there is an open segment $(u_0,u_1)\times \{0\}\subset I_{\tau_1}^-(0)$ not intersecting $\mathscr{E}(\tau_1)$ such that its future end point $(u_1,0)$ belongs to $\mathscr{E}(\tau_1)$ or there is a $\tau_0\in(0,\tau_1]$ such that $\mathscr{E}(\tau_1)$ coincides with $I_{\tau_0}^-(0)\setminus (0,0)$. However, according to the first alternative, $\sigma>1$ on $(u_0,u_1)\times \{0\}$ and $\sigma(u_1,0)=1$, whence $(\partial\sigma/\partial u)(u_1,0)\leq 0$, while according to the second alternative $\sigma=1$ on $I_{\tau_0}^-(0)\setminus(0,0)$, whence $\partial\sigma/\partial u=0$ on $I_{\tau_0}^-(0)\setminus(0,0)$. Since proposition \ref{proposition2} implies that, for $\tau$ suitably small,
\begin{align}
  \label{eq:435}
  e|_{I_{\tau_1}^-(0)\setminus(0,0)}>0,
\end{align}
we obtain, in conjunction with \eqref{eq:433}, a contradiction in both alternatives.
\end{proof}

\section{The Local Form of an Expansion Shock\label{local_form}}
We recall \eqref{eq:400}
\begin{align}
  \label{eq:436}
  \chi_\ast(\tau)=2\beta_0\tau+\mathcal{O}(\tau^2)=\frac{2i+4i_0}{5i-2i_0}\tau+\mathcal{O}(\tau^2).
\end{align}
Using now the expression for $\rho(\tau,\chi)$ as given by \eqref{eq:127}
\begin{align}
  \label{eq:437}
    \rho(\tau,\chi)=1+2i_0\left(\tau+\frac{\chi}{2}\right)\left(\tau-\frac{l\chi}{2}\right)+\mathcal{O}(\tau^3,\tau^2\chi,\tau\chi^2,\chi^3),
\end{align}
it follows that
\begin{align}
  \label{eq:438}
  \rho_\ast(\tau)=\rho(\tau,\chi_\ast(\tau))=1-\frac{24i^2(i-i_0)}{(5i-2i_0)^2}\tau^2+\mathcal{O}(\tau^3).
\end{align}
Since $C^{-}$, the null outgoing curve at $N_-$, is given by $\chi_-(\tau)=2\tau+\mathcal{O}(\tau^2)$ we obtain
\begin{align}
  \label{eq:439}
  \rho_-(\tau)=\rho(\tau,\chi_-(\tau))=1-8(i-i_0)\tau^2+\mathcal{O}(\tau^3).
\end{align}
Using \eqref{eq:363}, \eqref{eq:400} we obtain
\begin{align}
  \label{eq:440}
  \gamma(\tau)=1-\frac{6i^2(i+2i_0)}{(5i-2i_0)^2}\tau^2+\mathcal{O}(\tau^3).
\end{align}
Using \eqref{eq:438} in \eqref{eq:298} we obtain, through \eqref{eq:295},
\begin{align}
  \label{eq:441}
  \phi_\ast(\tau)=\tau-\frac{12i^3(i-i_0)(i+2i_0)^2}{5(5i-2i_0)^4}\tau^5+\mathcal{O}(\tau^6).
\end{align}
Using \eqref{eq:282} together with \eqref{eq:362} we obtain
\begin{align}
  \label{eq:442}
  r(\tau,0)=r_0-a_{-0}\frac{2(i-i_0)}{5i-2i_0}\tau+\mathcal{O}(\tau^2).
\end{align}

\noindent\textit{Remark.}  Moreover, any solution of the free boundary problem as stated in the end of section \ref{section1} with the regularity properties as stated in the existence theorem has the above local form.

\begin{center}
\begin{tikzpicture}
\draw [dashed] (0,0) -- (4,4);
\draw [dashed] (0,0) -- (-4,4);
\node at (-3,2.5) {$C^{\ast -}$};
\node at (3,2.5) {$C^-$};
\node at (4.4,4.4) {$\tau=\tau_-(\chi)$};
\draw [very thick] (0,0) .. controls (1.5,2) and (1.8,4) .. (2,4.5);
\node at (2,3.3) {$\mathscr{B}$};
\draw [very thick] (0,0) .. controls (2,-2) and (4,-1.8) .. (4.5,-2);
\node at (5,-2) {$\Sigma$};
\draw (0,0) .. controls (0.8,2).. (1,4.5);
\node at (1,4.75) {$\tau=\hat{\tau}(\chi)$};
\node at (1.6,1.5) [fill=white] {Soft};
\node at (3.4,-1) {Soft};
\node at (0.3,2) [fill=white] {Hard};
\node at (-2,-1) {Hard};
\filldraw (0,0) [fill=white] circle (2pt);
\node at (-0.4,-0.3) {$N^-$};
\end{tikzpicture}
\end{center}

\section{Uniqueness}
\begin{theorem}
  Let $\mathscr{B'}$, $(r',m',\nu',\kappa',\zeta',\eta')$ and $\mathscr{B''}$, $(r'',m'',\nu'',\kappa'',\zeta'',\eta'')$, both defined on $[0,\tau_1]$, $\mathscr{U}(\tau_1)$, be two solutions of the problem of formation of a free phase boundary in the phase transition from hard to soft, as formulated in Section \ref{section1}. Suppose that the two solutions have the same hard-phase initial data along $C^{\ast -}$ and also the same soft-phase initial data along $C^-$, and hence correspond to the same soft-phase solution $r,\omega,\rho$. Then the two solutions coincide.\label{theorem_2}
\end{theorem}

In the following, there can always be made the additional assumption that $\tau_1$ be sufficiently small.

Since
\begin{align}
  \label{eq:452}
  \frac{d\chi_\ast}{d\tau}=e^{-\omega_\ast}\beta=e^{-\omega_\ast}\frac{1-z^2}{1+z^2},
\end{align}
it follows that
\begin{align}
  \frac{d}{d\tau}|\chi_\ast''-\chi_\ast'|&=\textrm{sgn}(\chi_\ast''-\chi_\ast')\frac{d}{d\tau}(\chi_\ast''-\chi_\ast')\notag\\
&=\textrm{sgn}(\chi_\ast''-\chi_\ast')\left[\frac{-2e^{-\omega_\ast''}}{(1+z''^2)(1+z'^2)}(z''^2-z'^2)+\frac{1-z'^2}{1+z'^2}\left(e^{-\omega_\ast''}-e^{-\omega_\ast'}\right)\right].\label{eq:453}
\end{align}
From \eqref{eq:394}, \eqref{eq:395} together with the expression for $\{E\}_3$ given by \eqref{eq:408}, we have
\begin{align}
  \label{eq:454}
  E+X=it^2+U,
\end{align}
where
\begin{align}
  \label{eq:455}
  U=\bar{Q}_3+\bar{Q}_0(\{B\}_3-\Psi)+\bar{Q}_0(\{B\}_3-\Psi)^2.
\end{align}
Recalling $z^2=X/(E+X)$, we have
\begin{align}
  \label{eq:456}
  z''^2-z'^2&=\frac{1}{(it''^2+U'')(it'^2+U')}\Big\{-2i^2(t''x''t'^2-t'x't''^2)+i(t'^2Q_3(t'',x'')-t''^2Q_3(t',x'))\notag\\
&\qquad+2(i-i_0)(t''^2U'-t'^2U'')-2i(t''x''U'-t'x'U'')+U'Q_3(t'',x'')-U''Q_3(t',x')\Big\}.
\end{align}
In view of \eqref{eq:392} we see that $t''-t'+x''-x'=0$ (we compare quantities at the same value of $\tau$), it follows
\begin{align}
  \label{eq:457}
  t''x''t'^2-t'x't''^2=-2t't''\tau(\chi_\ast''-\chi_\ast'),
\end{align}
therefore the first term on the right hand side of \eqref{eq:456} is
\begin{align}
  \label{eq:458}
  \frac{4i^2t't''\tau(\chi_\ast''-\chi_\ast')}{(it''^2+U'')(it'^2+U')}.
\end{align}
Because of the factor
\begin{align}
  \label{eq:459}
  \textrm{sgn}(\chi_\ast''-\chi_\ast')\frac{-2e^{-\omega_\ast''}}{(1+z''^2)(1+z'^2)}
\end{align}
in \eqref{eq:453}, it follows that the term \eqref{eq:458} does contribute negatively in \eqref{eq:453} and can therefore be omitted in the estimates to follow. Making use of estimates of the type
\begin{align}
  \label{eq:460}
  |\bar{Q}_n''-\bar{Q}_n'|\leq C\tau^{n-1}|\chi_\ast''-\chi_\ast'|,
\end{align}
we can estimate
\begin{align}
  \label{eq:461}
\frac{d}{d\tau}|\chi_\ast''-\chi_\ast'|\leq C\left(|\chi_\ast''-\chi_\ast'|+\frac{|U''-U'|}{\tau^2}\right).  
\end{align}
Using \eqref{eq:455}, we can estimate
\begin{align}
  \label{eq:462}
  \frac{|U''-U'|}{\tau^2}\leq C\left(|\chi_\ast''-\chi_\ast'|+\frac{|\{B\}_3''-\{B\}_3'|}{\tau^2}+\frac{|\Psi''-\Psi'|}{\tau^2}\right).
\end{align}
Using proposition \ref{proposition6} we can estimate
\begin{align}
  \label{eq:463}
  \frac{|\{B\}_3''-\{B\}_3'|}{\tau^2}&\leq C\,\bigg(|\chi_\ast''-\chi_\ast'|+\frac{|\{N_\ast\}_3''-\{N_\ast\}_3'|}{\tau^2}\notag\\
&\qquad\qquad+\frac{|\{H\circ Y\}_3''-\{H\circ Y\}_3'|}{\tau^2}+\frac{|\{M_\ast\}_3''-\{M_\ast\}_3'|}{\tau^2}\bigg),\\
  \label{eq:464}
  \frac{|\{H\circ Y\}_3''-\{H\circ Y\}_3'|}{\tau^2}&\leq C\left(|\chi_\ast''-\chi_\ast'|+\frac{|\{\delta\}_3''-\{\delta\}_3'|}{\tau^2}\right).
\end{align}
In the arguments to follows we will use the notation
\begin{align}
  \label{eq:465}
  \Delta f:=f''-f'.
\end{align}
Now we define
\begin{align}
  \label{eq:466}
  a:=\sup_{[0,\tau_1]}|\Delta\chi_\ast|,\qquad b:=\sup_{[0,\tau_1]}\left(\frac{|\Delta\Psi|}{\tau^2}\right),\qquad c:=\sup_{[0,\tau_1]}\left(\frac{|\Delta\{\delta\}_3|}{\tau^2}\right),
\end{align}
\begin{align}
  \label{eq:467}
  l_1:=\sup_{[0,\tau_1]}\left(\frac{|\Delta\{N_\ast\}_3|}{\tau^2}\right),\qquad l_2:=\sup_{[0,\tau_1]}\left(\frac{|\Delta\{M_\ast\}_3|}{\tau^2}\right).
\end{align}
It follows from \eqref{eq:405}, \eqref{eq:407}
\begin{align}
  \label{eq:468}
\frac{|\Delta\{H\circ Y\}_3|}{\tau^2}&\leq C(a+c),\\
  \label{eq:469}
  \frac{|\Delta\{B\}_3|}{\tau^2}&\leq C(a+c+l_1+l_2).
\end{align}
We therefore can estimate
\begin{align}
  \label{eq:470}
  \frac{d}{d\tau}|\Delta\chi_\ast|\leq C(a +b+c+l_1+l_2),
\end{align}
and it follows
\begin{align}
  \label{eq:471}
  |\Delta\chi_\ast|\leq C\tau(a+b+c+l_1+l_2).
\end{align}
Going back to equation \eqref{eq:456}, we can estimate
\begin{align}
  \label{eq:472}
  |\Delta z^2|&\leq C \left(\frac{|\Delta\chi_\ast|}{\tau}+a+b+c+l_1+l_2\right)\notag\\
  &\leq C(a+b+c+l_1+l_2).
\end{align}

For an estimate of $|\Delta\{\delta\}_3|$ we use the equation for $\delta$ of proposition \ref{proposition6}.
Defining
\begin{align}
  \label{eq:473}
  Q_3(\tau,\chi):=\int_0^{\chi} Q_2(\tau,\tilde{\chi})d\tilde{\chi},
\end{align}
it follows
\begin{align}
  \label{eq:474}
  \frac{d\bar{Q}_3}{d\tau}(\tau)=\frac{dQ_3}{d\tau}(\tau,\chi_\ast(\tau))&=\frac{\p Q_3}{\p \tau}(\tau,\chi_\ast(\tau))+\frac{\p Q_3}{\p \chi}(\tau,\chi_\ast(\tau))\frac{d\chi_\ast}{d\tau}(\tau)\notag\\
&=\bar{Q}_2(\tau)+\bar{Q}_2(\tau)\frac{d\chi_\ast}{d\tau}(\tau).
\end{align}
Transferring the term $\bar{Q}_2\frac{d\chi_\ast}{d\tau}$ from the right hand side of \eqref{eq:409} to the left hand side using \eqref{eq:474}, we obtain
\begin{align}
  \label{eq:475}
  \frac{d}{d\tau}(\{\delta\}_3-\bar{Q}_3)&=\bar{Q}_2+\bar{Q}_1\left(e^{-N_\ast^{[2]}-\{N_\ast\}_3}-1\right)-\left(e^{-N_\ast^{[2]}-\{N_\ast\}_3}-1+N_\ast^{(1)}\right)\notag\\
&\qquad+\bigg[\bar{Q}_1\left(e^{-N_\ast^{[2]}-\{N_\ast\}_3}-1\right)+\frac{1}{2}\left(e^{-N_\ast^{[2]}-\{N_\ast\}_3}-1+N_\ast^{(1)}\right)\notag\\
&\hspace{29mm}+\bar{Q}_1\left(1-e^{-N_\ast^{[2]}-\{N_\ast\}_3}\right)+\bar{Q}_2e^{-N_\ast^{[2]}-\{N_\ast\}_3} \bigg]e^{-\omega_\ast}\frac{1-z^2}{1+z^2}.
\end{align}
where we also used \eqref{eq:452}. So we can estimate
\begin{align}
  \label{eq:476}
  \frac{d}{d\tau}\Big|\Delta(\{\delta\}_3-\bar{Q}_3)\Big|&\leq C(\tau|\Delta\chi_\ast|+|\Delta\{N_\ast\}_3|+\tau^2|\Delta z^2|)\notag\\
&\leq C\tau^2(a+b+c+l_1+l_2).
\end{align}
Hence
\begin{align}
  \label{eq:477}
  |\Delta\{\delta\}_3-\bar{Q}_3|\leq C\tau^3(a+b+c+l_1+l_2).
\end{align}
Since
\begin{align}
  \label{eq:478}
  |\Delta\bar{Q}_3|\leq C\tau^2|\Delta\chi_\ast|\leq C\tau^3(a+b+c+l_1+l_2),
\end{align}
it follows
\begin{align}
  \label{eq:479}
  |\Delta\{\delta\}_3|\leq C\tau^3(a+b+c+l_1+l_2).
\end{align}
Using the definition of $c$, it follows
\begin{align}
  \label{eq:480}
  c\leq C\tau_1(a+b+c+l_1+l_2),
\end{align}
which, for $\tau_1$ sufficiently small, implies
\begin{align}
  \label{eq:481}
  c\leq C\tau_1(a+b+l_1+l_2)
\end{align}
Using \eqref{eq:471}, \eqref{eq:472} it follows
\begin{align}
  \label{eq:482}
|\Delta\chi_\ast|\leq C\tau(a+b+l_1+l_2),\qquad  |\Delta z^2|\leq C(a+b+l_1+l_2).
\end{align}
Using the equation for $d\Psi/d\tau$ from proposition \ref{proposition6} we get
\begin{align}
  \label{eq:483}
  \frac{d}{d\tau}|\Delta\Psi|&\leq C(\tau^4|\Delta z^2|+\tau^3|\Delta\chi_\ast|)\notag\\
&\leq C\tau^4(a+b+l_1+l_2).
\end{align}
Hence
\begin{align}
  \label{eq:484}
  |\Delta\Psi|\leq C\tau^5(a+b+l_1+l_2).
\end{align}
The definition of $b$ yields
\begin{align}
  \label{eq:485}
  b\leq C\tau_1^3(a+b+l_1+l_2),
\end{align}
which, for $\tau_1$ small enough, implies
\begin{align}
  \label{eq:486}
  b\leq C\tau_1(a+l_1+l_2).
\end{align}
Therefore
\begin{align}
  \label{eq:487}
  |\Delta\chi_\ast|\leq C\tau(a+l_1+l_2).
\end{align}
Using the definition of $a$ gives
\begin{align}
  \label{eq:488}
  a\leq C\tau_1(a+l_1+l_2),
\end{align}
which, for $\tau_1$ small enough, implies
\begin{align}
  \label{eq:489}
  a\leq C\tau_1(l_1+l_2).
\end{align}
Using \eqref{eq:486} and \eqref{eq:481} we can estimate
\begin{align}
  \label{eq:490}
  b\leq C(l_1+l_2),\qquad c\leq C(l_1+l_2).
\end{align}
It follows from \eqref{eq:487}, \eqref{eq:484} and \eqref{eq:479} that
\begin{align}
  \label{eq:491}
  |\Delta\chi_\ast|\leq C\tau(l_1+l_2),\qquad |\Delta\Psi|\leq C\tau^5(l_1+l_2),\qquad |\Delta\{\delta\}_3|\leq C\tau^3(l_1+l_2).
\end{align}
Using the first of \eqref{eq:268} together with the second of \eqref{eq:482}, \eqref{eq:489} and the first of \eqref{eq:490}, it follows
\begin{align}
  \label{eq:492}
  |\Delta\beta|\leq C|z^2|\leq C(l_1+l_2).
\end{align}
In addition, we note that the above estimates imply
\begin{align}
  \label{eq:493}
  \left|\frac{d\Delta\chi_\ast}{d\tau}\right|\leq C(l_1+l_2).
\end{align}

\begin{proposition}
Let $\mathscr{U}(\tau):=\{(u,v):\tau\geq u\geq v\geq 0\}$ and let
\begin{align}
  \label{eq:494}
  Q:=\max\left\{\sup_{\mathscr{U}(\tau)}|\Delta r|,\sup_{\mathscr{U}(\tau)}|\Delta m|,\sup_{\mathscr{U}(\tau)}|\Delta \nu|,\sup_{\mathscr{U}(\tau)}|\Delta \kappa|,\sup_{\mathscr{U}(\tau)}|\Delta\zeta|,\sup_{\mathscr{U}(\tau)}|\Delta\eta|\right\}.
\end{align}
It then follows
\begin{align}
  \label{eq:495}
  Q\leq C(l_1+l_2).
\end{align}
Furthermore, for the differences of the functions $N$, $K$, $\phi$, it holds
\begin{align}
  \sup_{\mathscr{U}(\tau)}| \Delta N|\leq C\tau(l_1+l_2),\qquad\sup_{\mathscr{U}(\tau)}| \Delta K|\leq C\tau(l_1+l_2),\qquad \sup_{\mathscr{U}(\tau)}| \Delta \phi|\leq C\tau(l_1+l_2),\label{eq:496}
\end{align}
where
\begin{align}
  \label{eq:497}
  N(u,v)=\int_0^v\left(\left(\mu-4\pi r^2\right)\frac{\kappa}{r}\right)(u,v')dv',\qquad K(u,v)=\int_v^u (r\nu\zeta^2)(u',v)du'.
\end{align}

  \label{proposition8}
\end{proposition}
\begin{proof}
For a smooth function $g$ defined by the soft-phase solution, such as $r$, $\rho$, $a_-$ let $g_\ast'$ and $g_\ast''$ be the restrictions of $g$ to $\mathscr{B}'$ and $\mathscr{B}''$ respectively. Then we have
\begin{align}
  \label{eq:498}
  |\Delta g_\ast|\leq\Bigg|\int_{\chi_\ast''(\tau)}^{\chi_\ast'(\tau)}\left(\frac{\p g}{\p \chi}\right)(\tau,\chi)d\chi\Bigg|\leq C|\Delta\chi_\ast|\leq C\tau(l_1+l_2),
\end{align}
by the first of \eqref{eq:491}.

Integrating (6.6b) of \cite{I} with respect to $v$ yields
\begin{align}
  \label{eq:499}
  \nu(u,v)=\nu_\ast (u)e^{N(u,v)-N_\ast(u)},
\end{align}
While from \eqref{eq:325}
\begin{align}
  \label{eq:500}
  \kappa(u,v)=\kappa_\ast (v)e^{-K(u,v)}.
\end{align}
Therefore,
\begin{align}
  \label{eq:501}
\Delta\nu(u,v)&=e^{N'(u,v)-N_\ast'(u)}\Delta\nu_\ast(u)+\nu_\ast''(u)\Delta\left(e^{N(u,v)-N_\ast(u)}\right),\\
\Delta\kappa(u,v)&=e^{-K'(u,v)}\Delta\kappa_\ast(v)+K_\ast''(v)\Delta\left(e^{-K(u,v)}\right).  \label{eq:502}
\end{align}
\eqref{eq:313}, \eqref{eq:321} give
\begin{align}
  \label{eq:503}
\Delta\nu_\ast&=\frac{1}{2}(1-\beta')\Delta a_{-\ast}-\frac{1}{2}a_{-\ast}''\Delta\beta,\\
\Delta\kappa_\ast&=-\frac{\Delta a_{-\ast}}{2a_{-\ast}'a_{-\ast}''}(1+\beta')+\frac{1}{2a_{-\ast}''}\Delta\beta.\label{eq:504}
\end{align}
Using \eqref{eq:492} and the fact that by \eqref{eq:498}
\begin{align}
  \label{eq:505}
  |\Delta a_{-\ast}|\leq C\tau(\l_1+l_2),
\end{align}
we deduce
\begin{align}
  \label{eq:506}
  |\Delta\nu_\ast|,|\Delta\kappa_\ast|\leq C(l_1+l_2).
\end{align}
We then obtain
\begin{align}
  \label{eq:507}
  \sup_{\mathscr{U}(\tau)}|\Delta \nu|\leq C(l_1+l_2)+C\sup_{\mathscr{U}(\tau)}|\Delta N|,\qquad  \sup_{\mathscr{U}(\tau)}|\Delta \kappa|\leq C(l_1+l_2)+C\sup_{\mathscr{U}(\tau)}|\Delta K|.
\end{align}
Using the expressions for  $N$ and $K$ as given by \eqref{eq:204}, we get
\begin{align}
  \label{eq:508}
  \sup_{\mathscr{U}(\tau)}|\Delta N|&\leq C\tau\left(\sup_{\mathscr{U}(\tau)}|\Delta r|+\sup_{\mathscr{U}(\tau)}|\Delta m|+\sup_{\mathscr{U}(\tau)}|\Delta \kappa|\right),\\
  \sup_{\mathscr{U}(\tau)}|\Delta K|&\leq C\tau\left(\sup_{\mathscr{U}(\tau)}|\Delta r|+\sup_{\mathscr{U}(\tau)}|\Delta \zeta|+\sup_{\mathscr{U}(\tau)}|\Delta \nu|\right).  \label{eq:509}
\end{align}
Substituting into \eqref{eq:507} we obtain
\begin{align}
  \label{eq:510}
  \sup_{\mathscr{U}(\tau)}|\Delta \nu|&\leq C(l_1+l_2)+C\tau\left(\sup_{\mathscr{U}(\tau)}|\Delta r|+\sup_{\mathscr{U}(\tau)}|\Delta m|+\sup_{\mathscr{U}(\tau)}|\Delta \kappa|\right),\\
  \sup_{\mathscr{U}(\tau)}|\Delta \kappa|&\leq C(l_1+l_2)+C\tau\left(\sup_{\mathscr{U}(\tau)}|\Delta r|+\sup_{\mathscr{U}(\tau)}|\Delta \zeta|+\sup_{\mathscr{U}(\tau)}|\Delta \nu|\right).\label{eq:511}
\end{align}
When $\tau$ is sufficiently small these inequalities imply
\begin{align}
  \label{eq:512}
  \sup_{\mathscr{U}(\tau)}|\Delta \nu|&\leq C(l_1+l_2)+C\tau\left(\sup_{\mathscr{U}(\tau)}|\Delta r|+\sup_{\mathscr{U}(\tau)}|\Delta m|+\tau\sup_{\mathscr{U}(\tau)}|\Delta \zeta|\right),\\
  \sup_{\mathscr{U}(\tau)}|\Delta \kappa|&\leq C(l_1+l_2)+C\tau\left(\sup_{\mathscr{U}(\tau)}|\Delta r|+\sup_{\mathscr{U}(\tau)}|\Delta \zeta|+\tau\sup_{\mathscr{U}(\tau)}|\Delta m|\right).\label{eq:513}
\end{align}
Substituting into \eqref{eq:508}, \eqref{eq:509} then yields
\begin{align}
  \label{eq:514}
\sup_{\mathscr{U}(\tau)}|\Delta N|&\leq C\tau(l_1+l_2)+C\tau\left(\sup_{\mathscr{U}(\tau)}|\Delta r|+\sup_{\mathscr{U}(\tau)}|\Delta m|+\tau \sup_{\mathscr{U}(\tau)}|\Delta \zeta|\right),\\
\sup_{\mathscr{U}(\tau)}|\Delta K|&\leq C\tau(l_1+l_2)+C\tau\left(\sup_{\mathscr{U}(\tau)}|\Delta r|+\sup_{\mathscr{U}(\tau)}|\Delta \zeta|+\tau \sup_{\mathscr{U}(\tau)}|\Delta m|\right).\label{eq:515}
\end{align}
Using
\begin{align}
  \label{eq:516}
  r(u,v)=r_\ast(v)-\int_v^u\nu(u,v')dv'
\end{align}
together with \eqref{eq:498} we obtain
\begin{align}
  \label{eq:517}
\sup_{\mathscr{U}(\tau)}|\Delta r|\leq C\tau(l_1+l_2)+\tau \sup_{\mathscr{U}(\tau)}|\Delta \nu|.
\end{align}

Now, equation (6.5a) of \cite{I} is equivalent to
\begin{align}
  \label{eq:518}
  m(u,v)=m_\ast(v)e^{K(u,v)}-F(u,v),
\end{align}
where
\begin{align}
  F(u,v)=2\pi\int_v^u e^{K(u,v)-K(u',v)}(r^2\nu(\zeta^2+1))(u',v)du'.\notag
\end{align}
We can estimate
\begin{align}
  \label{eq:519}
  \sup_{\mathscr{U}(\tau)}|\Delta F|\leq C\tau\left(\sup_{\mathscr{U}(\tau)}|\Delta K|+\sup_{\mathscr{U}(\tau)}|\Delta r|+\sup_{\mathscr{U}(\tau)}|\Delta \nu|+\sup_{\mathscr{U}(\tau)}|\Delta \zeta|\right).
\end{align}
Together with \eqref{eq:498} in the case $g=m$ we obtain
\begin{align}
  \label{eq:520}
  \sup_{\mathscr{U}(\tau)}|\Delta m|\leq C\tau(l_1+l_2)+C\sup_{\mathscr{U}(\tau)}|\Delta K|+C\tau\left(\sup_{\mathscr{U}(\tau)}|\Delta r|+\sup_{\mathscr{U}(\tau)}|\Delta \nu|+\sup_{\mathscr{U}(\tau)}|\Delta \zeta|\right).
\end{align}

Substituting \eqref{eq:512} into \eqref{eq:517} and \eqref{eq:515}  into \eqref{eq:520} we deduce, when $\tau$ is small enough,
\begin{align}
  \label{eq:521}
  \sup_{\mathscr{U}(\tau)}|\Delta r| &\leq C\tau(l_1+l_2)+C\tau^2\left(\sup_{\mathscr{U}(\tau)}|\Delta m|+\tau\sup_{\mathscr{U}(\tau)}|\Delta \zeta|\right),\\
\sup_{\mathscr{U}(\tau)}|\Delta m|&\leq C\tau(l_1+l_2)+C\tau\left(\sup_{\mathscr{U}(\tau)}|\Delta r|+\sup_{\mathscr{U}(\tau)}|\Delta \zeta|\right).
\end{align}
When $\tau$ is suitably small these inequalities imply
\begin{align}
  \label{eq:522}
  \sup_{\mathscr{U}(\tau)}|\Delta r|&\leq C\tau(l_1+l_2) +C\tau^3\sup_{\mathscr{U}(\tau)}|\Delta \zeta|,\\
  \sup_{\mathscr{U}(\tau)}|\Delta m|&\leq C\tau(l_1+l_2) +C\tau\sup_{\mathscr{U}(\tau)}|\Delta \zeta|.\label{eq:523}
\end{align}
Substituting these into \eqref{eq:512}, \eqref{eq:513}, \eqref{eq:514}, \eqref{eq:515} yields
\begin{align}
  \label{eq:524}
\sup_{\mathscr{U}(\tau)}|\Delta \nu|&\leq C(l_1+l_2) +C\tau^2\sup_{\mathscr{U}(\tau)}|\Delta \zeta|,\\
\sup_{\mathscr{U}(\tau)}|\Delta \kappa|&\leq C(l_1+l_2) +C\tau\sup_{\mathscr{U}(\tau)}|\Delta \zeta|,\label{eq:525}\\
\sup_{\mathscr{U}(\tau)}|\Delta N|&\leq C\tau(l_1+l_2) +C\tau^2\sup_{\mathscr{U}(\tau)}|\Delta \zeta|,\label{eq:526}\\
\sup_{\mathscr{U}(\tau)}|\Delta K|&\leq C\tau(l_1+l_2) +C\tau\sup_{\mathscr{U}(\tau)}|\Delta \zeta|.\label{eq:527}
\end{align}

We have
\begin{align}
  \label{eq:528}
  \phi(u,v)=\phi_\ast(v)+\int_v^u(\nu\zeta)(u',v)du'.
\end{align}
Using the second of \eqref{eq:491}, we can estimate
\begin{align}
  \label{eq:529}
  \sup_{\mathscr{U}(\tau)}|\Delta \phi|\leq C\tau^5(l_1+l_2)+C\tau \sup_{\mathscr{U}(\tau)}|\Delta \zeta|+C\tau\sup_{\mathscr{U}(\tau)}|\Delta \nu|.
\end{align}
Taking into account \eqref{eq:524} we obtain
\begin{align}
  \label{eq:530}
  \sup_{\mathscr{U}(\tau)}|\Delta \phi|\leq C\tau(l_1+l_2)+C\tau \sup_{\mathscr{U}(\tau)}|\Delta \zeta|.
\end{align}

We now estimate the difference $\Delta \zeta$ in $\mathscr{U}(\tau)$. We recall \eqref{eq:320}:
\begin{align}
  \label{eq:531}
  \zeta_\ast=\frac{1}{a_{-\ast}}\sqrt{\frac{1+\beta-2\rho_\ast\beta}{1-\beta}}.
\end{align}
Using \eqref{eq:492} as well as \eqref{eq:498} in the cases $g=a_-,\rho$, we find
\begin{align}
  \label{eq:532}
|\Delta\zeta_\ast|\leq C\tau(l_1+l_2).
\end{align}
Recalling that $\alpha_\ast=r_\ast\zeta_\ast-\phi_\ast$, and using \eqref{eq:498} in the case $g=r$, we get
\begin{align}
  \label{eq:533}
  |\Delta\alpha_\ast|\leq \tau C(l_1+l_2).
\end{align}
We recall \eqref{eq:278}:
\begin{align}
  \label{eq:534}
  \alpha(u,v)=\left(\alpha_\ast(u)e^{N_\ast(u)}+M_\ast(u)-M(u,v)\right)e^{-N(u,v)}.
\end{align}
Using \eqref{eq:533} we deduce
\begin{align}
  \label{eq:535}
  \sup_{\mathscr{U}(\tau)}|\Delta\alpha|\leq C\tau(l_1+l_2)+\sup_{\mathscr{U}(\tau)}| \Delta N|+\sup_{\mathscr{U}(\tau)}| \Delta M|.
\end{align}
In view of the fact that $r\zeta=\alpha+\phi$ and \eqref{eq:530}, this implies
\begin{align}
  \label{eq:536}
  \sup_{\mathscr{U}(\tau)}| \Delta\zeta|\leq C\left(\tau(l_1+l_2)+\sup_{\mathscr{U}(\tau)}| \Delta N|+\sup_{\mathscr{U}(\tau)}| \Delta M|+\sup_{\mathscr{U}(\tau)}| \Delta r|+\sup_{\mathscr{U}(\tau)}| \Delta \phi|\right).
\end{align}
Substituting the estimates \eqref{eq:522}, \eqref{eq:526}, \eqref{eq:530} into \eqref{eq:536} yields
\begin{align}
  \label{eq:537}
  \sup_{\mathscr{U}(\tau)}| \Delta \zeta|\leq C\left(\tau(l_1+l_2)+\tau\sup_{\mathscr{U}(\tau)}| \Delta \zeta|+\sup_{\mathscr{U}(\tau)}| \Delta M|\right).
\end{align}

We have to consider the integral $M$:
\begin{align}
  \label{eq:538}
  M(u,v)=\int_0^v\left(\phi(\mu-4\pi r^2)\frac{\kappa}{r}\right)(u,v')e^{N(u,v')}dv'.
\end{align}
Since $\sup_{\mathscr{U}(\tau)}| \phi|\leq C\tau$ we deduce
\begin{align}
  \label{eq:539}
  \sup_{\mathscr{U}(\tau)}| \Delta M|\leq C\left(\tau\sup_{\mathscr{U}(\tau)}| \Delta \phi|+\tau^2\sup_{\mathscr{U}(\tau)}| \Delta \kappa|+\tau^2\sup_{\mathscr{U}(\tau)}| \Delta r|+\tau^2\sup_{\mathscr{U}(\tau)}| \Delta m|+\tau^2\sup_{\mathscr{U}(\tau)}| \Delta N|\right).
\end{align}
Substituting the estimates \eqref{eq:522}, \eqref{eq:523}, \eqref{eq:525}, \eqref{eq:526}, \eqref{eq:530} into \eqref{eq:539} we then obtain
\begin{align}
  \label{eq:540}
  \sup_{\mathscr{U}(\tau)}| \Delta M|\leq C\tau^2(l_1+l_2)+C\tau^2\sup_{\mathscr{U}(\tau)}| \Delta \zeta|.
\end{align}
Substituting in turn \eqref{eq:540} into \eqref{eq:537} yields the inequality
\begin{align}
  \label{eq:541}
  \sup_{\mathscr{U}(\tau)}| \Delta \zeta|\leq C\tau(l_1+l_2)+C\tau\sup_{\mathscr{U}(\tau)}| \Delta \zeta|,
\end{align}
which for $\tau$ suitably small implies
\begin{align}
  \label{eq:542}
  \sup_{\mathscr{U}(\tau)}| \Delta \zeta|\leq C\tau(l_1+l_2).
\end{align}
In view of this estimate, the inequalities \eqref{eq:522}, \eqref{eq:523}, \eqref{eq:524}, \eqref{eq:525}, \eqref{eq:526}, \eqref{eq:527}, \eqref{eq:530}, \eqref{eq:540} reduce to
\begin{alignat}{3}
  \sup_{\mathscr{U}(\tau)}| \Delta r|&\leq C(l_1+l_2),\qquad & \sup_{\mathscr{U}(\tau)}| \Delta m|&\leq C(l_1+l_2),\label{eq:543}\\
  \sup_{\mathscr{U}(\tau)}| \Delta \nu|&\leq C(l_1+l_2),\qquad &\sup_{\mathscr{U}(\tau)}| \Delta \kappa|&\leq C(l_1+l_2),    \label{eq:544}\\
  \sup_{\mathscr{U}(\tau)}| \Delta N|&\leq C\tau(l_1+l_2),\qquad &\sup_{\mathscr{U}(\tau)}| \Delta K|&\leq C\tau(l_1+l_2),  \label{eq:545}\\
  \sup_{\mathscr{U}(\tau)}| \Delta \phi|&\leq C\tau(l_1+l_2),\qquad &\sup_{\mathscr{U}(\tau)}| \Delta M|&\leq C\tau^2(l_1+l_2).\label{eq:546}
\end{alignat}

We write (6.7a) of \cite{I} in the form
\begin{align}
  \label{eq:547}
  \frac{\partial(r\eta)}{\partial u}-4\pi r\nu\zeta^2(r\eta)=-\nu(1-\mu)\zeta.
\end{align}
From this we have
\begin{align}
  \label{eq:548}
  (r\eta)(u,v)=e^{K(u,v)}\left[(r_\ast \eta_\ast)(v)-J(u,v)\right],
\end{align}
where
\begin{align}
  \label{eq:549}
  J(u,v)=\int_u^ve^{-K(u',v)}((1-\mu)\nu\zeta)(u',v)du',\qquad\eta_\ast=a_{-\ast}\sqrt{\frac{1-\beta+2\rho_\ast\beta}{1+\beta}},
\end{align}
(the second is (4.5b) of \cite{II}). We can estimate
\begin{align}
  \label{eq:550}
  \sup_{\mathscr{U}(\tau)}|\Delta J|\leq C\tau\left(\sup_{\mathscr{U}(\tau)}|\Delta K|+\sup_{\mathscr{U}(\tau)}|\Delta m|+\sup_{\mathscr{U}(\tau)}|\Delta r|+\sup_{\mathscr{U}(\tau)}|\Delta \nu|+\sup_{\mathscr{U}(\tau)}|\Delta \zeta|\right).
\end{align}
Using \eqref{eq:492}, \eqref{eq:498}, \eqref{eq:543}, \eqref{eq:544}, \eqref{eq:545}, it follows
\begin{align}
  \label{eq:551}
  \sup_{\mathscr{U}(\tau)}| \Delta \eta|\leq C(l_1+l_2).
\end{align}
This concludes the proof of the proposition.
\end{proof}

To complete the uniqueness proof, we must estimate $l_1$ and $l_2$.
\begin{proposition}
It holds
  \begin{align}
    \label{eq:552}
    \left|\frac{d\Delta\{M_\ast\}_3}{d\tau}\right|&\leq C\tau^2(l_1+l_2),\\
    \label{eq:553}
    \left|\frac{d\Delta\{N_\ast\}_3}{d\tau}\right|&\leq C\tau^2(l_1+l_2).
  \end{align}
  \label{proposition7}
\end{proposition}

\begin{proof}
We recall \eqref{eq:312}:
\begin{align}
  \label{eq:554}
  \frac{dM_\ast}{d\tau}=g(\tau,\tau)+\int_0^\tau \frac{\partial g}{\partial u}(\tau,v)dv,
\end{align}
where $g=\phi fe^N$, $f=\frac{\mu-4\pi r^2}{r}\kappa$ (see \eqref{eq:274}). Recalling the definition of $\psi$ (see the first of \eqref{eq:343}) as well as \eqref{eq:301} and \eqref{eq:313} we have
\begin{align}
  \label{eq:555}
  g(\tau,\tau)=(\tau+\Psi)e^{N_\ast}\frac{\psi_\ast}{2}(1+\beta),
\end{align}
which we rewrite as
\begin{align}
  \label{eq:556}
  g(\tau,\tau)=\left\{\frac{\psi_0}{2}\tau+\left(\frac{\psi_\ast}{2}-\frac{\psi_0}{2}\right)\tau+\frac{\psi_\ast}{2}\left[\left(e^{N_\ast}-1\right)\tau+e^{N_\ast}\Psi\right]\right\}(1+\beta).
\end{align}
Since
\begin{align}
  \label{eq:557}
  \beta-\frac{1}{2}\frac{d\chi_\ast}{d\tau}=\left(e^{\omega_\ast}-e^{\omega_0}\right)\frac{d\chi_\ast}{d\tau}=\bar{Q}_1\frac{d\chi_\ast}{d\tau},
\end{align}
we have
\begin{align}
  \label{eq:558}
  g(\tau,\tau)=\frac{\psi_0}{2}\left(1+\frac{1}{2}\frac{d\chi_\ast}{d\tau}\right)\tau+\bar{Q}_2+\bar{Q}_2\frac{d\chi_\ast}{d\tau}+R_0,
\end{align}
where
\begin{align}
  \label{eq:559}
  R_0:=\left[\left(e^{N_\ast}-1\right)\tau+e^{N_\ast}\Psi\right]\frac{\psi_\ast}{2}(1+\beta).
\end{align}
Using now the second of \eqref{eq:491} together with \eqref{eq:492} and the first of \eqref{eq:496} we get
\begin{align}
  \label{eq:560}
  |\Delta R_0|\leq C\tau^2(l_1+l_2),
\end{align}
which, together with \eqref{eq:460}, \eqref{eq:493}, yields
\begin{align}
  \label{eq:561}
  \left|\Delta\left[g(\tau,\tau)-\frac{\psi_0}{2}\left(1+\frac{1}{2}\frac{d\chi_\ast}{d\tau}\right)\tau\right]\right|\leq C\tau^2(l_1+l_2).
\end{align}

Now we look at
\begin{align}
  \label{eq:562}
  \int_0^\tau\frac{\partial g}{\partial u}(\tau,v)dv=I_1+I_2,
\end{align}
where (see \eqref{eq:316})
\begin{align}
  \label{eq:563}
  I_1:=\int_0^\tau\left(\frac{\p \phi}{\p u}fe^N\right)(\tau,v)dv,\qquad I_2:=\int_0^\tau \left[\phi \,e^N\left(\frac{\p f}{\p u}+f\frac{\partial N}{\partial u}\right)\right](\tau,v)dv.
\end{align}
We first look at $I_1$. We recall the function $F$ (see \eqref{eq:318})
\begin{align}
  \label{eq:564}
  F:=\frac{\partial \phi}{\partial u}\frac{\mu-4\pi r^2}{r}e^N
\end{align}
and rewrite
\begin{align}
  \label{eq:565}
  I_1=F_\ast\int_0^\tau \kappa(\tau,v)dv+\tilde{I}_1,
\end{align}
where
\begin{align}
  \label{eq:566}
  \tilde{I}_1:=\int_0^\tau\left[F(\tau,v)-F(\tau,\tau)\right]\kappa(\tau,v)dv.
\end{align}
We define $R_1$, $R_2$, $R_3$ by
\begin{align}
  \label{eq:567}
  R_1&:=\left(\frac{\partial\phi}{\partial u}\right)_\ast-\frac{1}{2}(1-\beta),\\
  R_2&:=\frac{\mu_\ast-4\pi r_\ast^2}{r_\ast}-\frac{\mu_0-4\pi r_0^2}{r_0},\\
  R_3&:=e^{N_\ast}-1.
\end{align}
We note that $R_i=\mathcal{O}(\tau):i=1,2,3$ and $R_2=\bar{Q}_1$. We have
\begin{align}
  \label{eq:568}
  F_\ast=\frac{1}{2}(1-\beta)\frac{\mu_0-4\pi r_0^2}{r_0}+R,
\end{align}
where
\begin{align}
  \label{eq:569}
  R=R_1\frac{\mu_\ast-4\pi r_\ast^2}{r_\ast}+\frac{1}{2}(1-\beta)R_2e^{N_\ast}+\left(\frac{\partial \phi}{\partial u}\right)_\ast\frac{\mu_0-4\pi r_0^2}{r_0}R_3+R_1R_2R_3.
\end{align}
From (see \eqref{eq:500})
\begin{align}
  \label{eq:570}
  \kappa(\tau,v)=\kappa_\ast (v)e^{-K(\tau,v)},
\end{align}
together with (see \eqref{eq:313} with \eqref{eq:328})
\begin{align}
  \label{eq:571}
  \kappa_\ast=\frac{1}{2a_{-0}}\left(1+\frac{1}{2}\frac{d\chi_\ast}{d\tau}\right)+\bar{Q}_1+\bar{Q}_1\frac{d\chi_\ast}{d\tau},
\end{align}
we deduce
\begin{align}
  \label{eq:572}
  \kappa(\tau,v)=\frac{1}{2a_{-0}}\left(1+\frac{1}{2}\frac{d\chi_\ast}{d\tau}(v)\right)+\bar{\kappa}(\tau,v),
\end{align}
where
\begin{align}
  \label{eq:573}
  \bar{\kappa}(\tau,v):=\bar{Q}_1(v)+\bar{Q}_1(v)\frac{d\chi_\ast}{d\tau}(v)-\kappa_\ast(v)\left(1-e^{-K(\tau,v)}\right).
\end{align}
We note that $\bar{\kappa}=\mathcal{O}(\tau)$. Substituting \eqref{eq:568}, \eqref{eq:572} into \eqref{eq:565} it follows
\begin{align}
  \label{eq:574}
  I_1&=\frac{\mu_0-4\pi r_0^2}{4a_{-0}r_0}\left(1-\frac{1}{2}\frac{d\chi_\ast}{d\tau}\right)\left(\tau+\frac{\chi_\ast}{2}\right)\notag\\
&\qquad+\frac{\mu_0-4\pi r_0^2}{2r_0}\left(1+(\bar{Q}_1-\tfrac{1}{2})\frac{d\chi_\ast}{d\tau}\right)\int_0^\tau\bar{\kappa}(\tau,v)dv+\bar{Q}_2\frac{d\chi_\ast}{d\tau}+R\int_0^\tau\kappa(\tau,v)dv+\tilde{I}_1.
\end{align}

What is now needed are estimates for the differences $\Delta R$, $\Delta\bar{\kappa}$ and $\Delta \tilde{I}_1$. We start with $|\Delta R|$. From \eqref{eq:569}, using proposition \ref{proposition8} (recall that $(\partial\phi/\partial u)_\ast=\nu_\ast\zeta_\ast$),
\begin{align}
  \label{eq:575}
  |\Delta R|&\leq C\left\{|\Delta R_1|+|\Delta R_2|+|\Delta R_3|+\tau\left[(l_1+l_2)+|\Delta\beta|+|\Delta N_\ast|\right]\right\}\notag\\
&\leq C\left\{|\Delta R_1|+|\Delta R_2|+|\Delta R_3|+\tau(l_1+l_2)\right\},
\end{align}
where for the second inequality we used \eqref{eq:492} and the first of \eqref{eq:496}. Using $\partial \phi/\partial u=\nu\zeta$ together with \eqref{eq:320}, \eqref{eq:321} we have (recall that $X=1-\rho_\ast=\mathcal{O}(\tau^2)$)
\begin{align}
  \label{eq:576}
  R_1=\frac{1}{2}(1-\beta)\left(\sqrt{1+\frac{2\beta X}{1-\beta}}-1\right).
\end{align}
Therefore,
\begin{align}
  \label{eq:577}
  |\Delta R_1|\leq C\left(\tau^2|\Delta\beta|+|\Delta X|\right).
\end{align}
From (using \eqref{eq:460})
\begin{align}
  \label{eq:578}
  |\Delta X|=|\bar{Q}_2|\leq C \tau|\Delta\chi_\ast|\leq C\tau^2(l_1+l_2),
\end{align}
we obtain
\begin{align}
  \label{eq:579}
  |\Delta R_1|\leq C\tau^2(l_1+l_2).
\end{align}
For the difference of $R_2$ we have, using the first of \eqref{eq:491}
\begin{align}
  \label{eq:580}
  |\Delta R_2|= |\Delta \bar{Q}_1|\leq C|\Delta\chi_\ast|\leq C\tau(l_1+l_2),
\end{align}
while for the difference of $R_3$ the first of \eqref{eq:496} implies
\begin{align}
  \label{eq:581}
  |\Delta R_3|\leq C|\Delta N_\ast|\leq C\tau (l_1+l_2).
\end{align}
Therefore,
\begin{align}
  \label{eq:582}
  |\Delta R|\leq C\tau(l_1+l_2).
\end{align}
Together with $|\Delta \kappa|\leq C(l_1+l_2)$, we obtain
\begin{align}
  \label{eq:583}
  \left|\Delta\left(R\int_0^\tau\kappa(\tau,v)dv\right)\right|\leq C\tau^2(l_1+l_2).
\end{align}
From \eqref{eq:573},
\begin{align}
  \label{eq:584}
  |\Delta\bar{\kappa}|&\leq C\left\{|\Delta\bar{Q}_1|+\tau\left|\frac{d\Delta\chi_\ast}{d\tau}\right|+\tau|\Delta\chi_\ast|+|\Delta K|\right\}\notag\\
&\leq C\tau(l_1+l_2),
\end{align}
where for the second inequality we used \eqref{eq:460}, \eqref{eq:493}, the first of \eqref{eq:491} and the second of \eqref{eq:496}. Now we turn to $\Delta\tilde{I}_1$. We write
\begin{align}
  \label{eq:585}
  F(\tau,v)-F(\tau,\tau)=\int_\tau^v\frac{\partial F}{\partial v}(\tau,v')dv'.
\end{align}
Now, $\partial F/\partial v$ can be expressed, through the basic equations of section 6 of \cite{I} (and also using the expression for $N$ as given by \eqref{eq:204}), in terms of $r$, $m$, $\nu$, $\kappa$, $\zeta$, $\eta$ and $N$. Therefore, by proposition \ref{proposition8} we obtain
\begin{align}
  \label{eq:586}
  \left|\Delta\left[F(\tau,v)-F(\tau,\tau)\right]\right|\leq C\tau(l_1+l_2).
\end{align}
Hence
\begin{align}
  \label{eq:587}
  |\Delta\tilde{I}_1|\leq C\tau^2(l_1+l_2).
\end{align}
Using now \eqref{eq:583}, \eqref{eq:584} and \eqref{eq:587} in \eqref{eq:574} together with \eqref{eq:460} and the first of \eqref{eq:491} we find
\begin{align}
  \label{eq:588}
  \left|\Delta\left[I_1-\frac{\mu_0-4\pi r_0^2}{4a_{-0}r_0}\left(1-\frac{1}{2}\frac{d\chi_\ast}{d\tau}\right)\left(\tau+\frac{\chi_\ast}{2}\right)\right]\right|\leq C\tau^2(l_1+l_2).
\end{align}

We now look at $I_2$. Since $f$ as well as $\partial f/\partial u$ can be expressed in terms of $r$, $m$, $\nu$, $\kappa$, $\zeta$, and $\eta$ we have, by proposition \ref{proposition8},
\begin{align}
  \label{eq:589}
  |\Delta f|\leq C(l_1+l_2),\qquad \left|\frac{\partial\Delta f}{\partial u}\right|\leq C(l_1+l_2).
\end{align}
From the second of these we obtain
\begin{align}
  \label{eq:590}
  \left|\frac{\partial\Delta N}{\partial u}(u,v)\right|\leq\int_0^v\left|\frac{\partial \Delta f}{\partial u}\right|(u,v')dv'\leq Cv(l_1+l_2),
\end{align}
together with \eqref{eq:589} and the first and third of \eqref{eq:496} we obtain
\begin{align}
  \label{eq:591}
  |\Delta I_2|\leq C\tau^2(l_1+l_2).
\end{align}

In view of the estimates \eqref{eq:561}, \eqref{eq:588}, \eqref{eq:591} and the expression for $dM_\ast^{[2]}/d\tau$ (see the first line of \eqref{eq:333}) it follows
\begin{align}
  \label{eq:592}
  \left|\frac{d\Delta\{M_\ast\}_3}{d\tau}\right|\leq C\tau^2(l_1+l_2).
\end{align}

We now turn to $|d\Delta\{N_\ast\}_3/d\tau|$. Let us recall the functions $\psi$ and $G$:
\begin{align}
  \label{eq:593}
  \psi:=\frac{\mu-4\pi r^2}{ra_-},\qquad G:=-\frac{\nu}{r^2}\left[4\pi \zeta^2r^2(1-4\pi r^2)-2\mu\right].
\end{align}
We recall \eqref{eq:335}:
\begin{align}
  \label{eq:608}
  \frac{dN_\ast}{d\tau}(\tau)&=f(\tau,\tau)+\int_0^\tau \frac{\p f}{\p u}(\tau,v)dv.
\end{align}
The first term is equal to $(\psi_\ast/2)(1+\beta)$, while the second term is equal to $\int_0^\tau (G\kappa)(\tau,v)dv$ (see \eqref{eq:336}, \eqref{eq:337}). We rewrite \eqref{eq:608} as
\begin{align}
  \label{eq:594}
  \frac{dN_\ast}{d\tau}=\frac{1}{2}\psi_\ast(1+\beta)+G_\ast\int_0^\tau \kappa(\tau,v)dv+I_3,
\end{align}
where
\begin{align}
  \label{eq:595}
  I_3:=\int_0^\tau\left[G(\tau,v)-G(\tau,\tau)\right]\kappa(\tau,v)dv.
\end{align}
We note that
\begin{align}
  \label{eq:596}
  \psi_\ast-\left(\psi_0+\left(\frac{\partial\psi}{\partial\tau}\right)_0\tau+\left(\frac{\partial\psi}{\partial\chi}\right)_0\chi_\ast\right)&=\bar{Q}_2,\\
e^{\omega_\ast}-\left(e^{\omega_0}+\left(\frac{\partial e^\omega}{\partial\tau}\right)_0\tau+\left(\frac{\partial e^\omega}{\partial\chi}\right)_0\chi_\ast\right)&=\bar{Q}_2.
\end{align}
Therefore
\begin{align}
  \label{eq:597}
  \frac{1}{2}\psi_\ast(1+\beta)&=\frac{1}{2}\left[\psi_0+\left(\frac{\partial\psi}{\partial\tau}\right)_0\tau+\left(\frac{\partial\psi}{\partial\chi}\right)_0\chi_\ast\right]\left[1+\left(e^{\omega_0}+\left(\frac{\partial e^\omega}{\partial \tau}\right)_0\tau+\left(\frac{\partial e^\omega}{\partial \chi}\right)_0\chi_\ast\right)\frac{d\chi_\ast}{d\tau}\right]\notag\\
&\qquad+\bar{Q}_2+\bar{Q}_2\frac{d\chi_\ast}{d\tau}.
\end{align}
We recall \eqref{eq:342}:
\begin{align}
  \label{eq:598}
  G_\ast=-\frac{2\pi(1-4\pi r_\ast^2)}{a_{-\ast}}(1+\beta-2\rho_\ast\beta)+\frac{\mu_\ast a_{-\ast}}{r_\ast^2}(1-\beta).
\end{align}
Therefore
\begin{align}
  \label{eq:599}
  G_\ast=\tilde{\psi}_0\left(1-\frac{1}{2}\frac{d\chi_\ast}{d\tau}\right)+\bar{Q}_1+\bar{Q}_1\frac{d\chi_\ast}{d\tau},
\end{align}
where (see the second of \eqref{eq:343})
\begin{align}
  \label{eq:610}
  \tilde{\psi}:=-\frac{2\pi(1-4\pi r^2)}{a_-}+\frac{\mu a_-}{r^2}.
\end{align}
Substituting \eqref{eq:572}, \eqref{eq:597}, \eqref{eq:599} into \eqref{eq:594}, using \eqref{eq:346} and taking into account the above estimates to estimate $d\Delta\chi_\ast/d\tau$ as well as \eqref{eq:460} we obtain
\begin{align}
  \label{eq:600}
  \left|\frac{d\Delta\{N\}_3}{d\tau}\right|\leq C\tau^2(l_1+l_2)+|\Delta I_3|+C\left|\Delta\int_0^\tau \kappa_\ast(v)\left(1-e^{-K(\tau,v)}\right)dv\right|.
\end{align}
By proposition \ref{proposition8} the last term in \eqref{eq:600} is bounded from above in absolute value by $C\tau^2(l_1+l_2)$. We conclude
\begin{align}
  \label{eq:601}
  \left|\frac{d\Delta\{N\}_3}{d\tau}\right|\leq C\tau^2(l_1+l_2)+|\Delta I_3|.
\end{align}

We now consider the difference of $I_3$. We rewrite
\begin{align}
  \label{eq:602}
  G(\tau,v)-G(\tau,\tau)=-\int_v^\tau\frac{\partial G}{\partial v}(\tau,v')dv'.
\end{align}
$\partial G/\partial v$ can be expressed through the basic equations of section 6 of \cite{I} in terms of $r$, $m$, $\nu$, $\kappa$, $\zeta$ and $\eta$. Therefore, from proposition \ref{proposition8},
\begin{align}
  \label{eq:603}
  \left|\Delta\left(G(\tau,v)-G(\tau,\tau)\right)\right|\leq C\tau(l_1+l_2),
\end{align}
which implies
\begin{align}
  \label{eq:604}
  |\Delta I_3|\leq C\tau^2(l_1+l_2).
\end{align}
We conclude
\begin{align}
  \label{eq:605}
  \left|\frac{d\Delta\{N\}_3}{d\tau}\right|\leq C\tau^2(l_1+l_2).
\end{align}
\end{proof}

From proposition \ref{proposition7} we deduce
\begin{align}
  \label{eq:606}
  |\Delta \{M_\ast\}_3|\leq C\tau^3(l_1+l_2),\qquad |\Delta \{N_\ast\}_3|\leq C\tau^3(l_1+l_2).
\end{align}
Together with the definitions of $l_1$, $l_2$, this yields
\begin{align}
  \label{eq:607}
l_1+l_2=\sup_{[0,\tau_1]}\left(\frac{|\Delta \{M_\ast\}_3|}{\tau^2}\right)+\sup_{[0,\tau_1]}\left(\frac{|\Delta \{N_\ast\}_3|}{\tau^2}\right)\leq C\tau_1(l_1+l_2),
\end{align}
which when $C\tau_1<1$ implies that $l_1+l_2=0$ which in turn implies $l_1=l_2=0$.

In view of propostition \ref{proposition8} together with the first of \eqref{eq:491} we conclude that for $\tau_1$ suitably small, the solutions $\mathscr{B'}$, $(r',m',\nu',\kappa',\zeta',\eta')$ and $\mathscr{B''}$, $(r'',m'',\nu'',\kappa'',\zeta'',\eta'')$ coincide on $[0,\tau_1]$, $\mathscr{U}(\tau_1)$. The uniqueness without a smallness condition follows immediately from the uniqueness in the large of the solution of the continuation problem, theorem 3.2 of \cite{II}. This proves theorem \ref{theorem_2}.

\bibliographystyle{plain}
\bibliography{bibliography}

\end{document}